\documentclass{article}
\usepackage{ijcai20}

\newif\iflong
\newif\ifshort

\longtrue

\iflong
\else
\shorttrue
\fi

\usepackage{times}
\usepackage{soul}
\usepackage{url}
\usepackage[utf8]{inputenc}
\usepackage[small]{caption}
\usepackage{graphicx}
\usepackage{amsmath}
\usepackage{amsthm}
\usepackage{booktabs} 

\usepackage{amsmath,amssymb,graphicx,amsthm,dsfont}
\usepackage{color,soul}
\usepackage{xspace}
\usepackage{enumerate}
\usepackage{bbm}
\usepackage{mathtools}
\usepackage{comment}
\usepackage{paralist} 
\usepackage{nicefrac}
\usepackage{multirow} 

\usepackage{colortbl}
 
\usepackage{subcaption} 
\usepackage[noend,ruled,linesnumbered]{algorithm2e}

\SetAlFnt{\scriptsize}
\SetAlCapFnt{\footnotesize}
\SetAlCapNameFnt{\footnotesize}
\SetAlCapHSkip{0pt} 

\newcommand{\algor}{{\normalfont \textbf{or}}\xspace}
\newcommand{\algand}{{\normalfont \textbf{and}}\xspace}

\usepackage{xifthen}

\newcommand{\bigoh}{\mathcal{O}}

\usepackage[linecolor=blue!70!black, backgroundcolor=blue!10, bordercolor=black, textsize=scriptsize, obeyFinal,disable]{todonotes}

\newcommand{\todoH}[1]{\todo[linecolor=yellow!70!black, backgroundcolor=yellow!10]{H: #1}}

\newcommand{\todoT}[1]{\todo[linecolor=red!70!black, backgroundcolor=red!10]{T: #1}}
\newcommand{\todoR}[1]{\todo[linecolor=green!70!black, backgroundcolor=green!10]{R: #1}}
\usepackage{complexity}
\usepackage{makecell}

\usepackage[%
pagebackref,
breaklinks=true,
colorlinks=true,
linkcolor=green!40!black!90,
citecolor=darkblue,
urlcolor=purple,
pdfpagelayout=SinglePage,
pdfstartview=Fit]{hyperref}
\renewcommand*{\backref}[1]{}
\renewcommand*{\backrefalt}[4]{%
\ifcase #1%
\marginpar{\tiny no cite}
\or
$\rightarrow$~p.~#2.%
\else
 $\rightarrow$~pp.~#2.%
\fi
}

\newcommand{\St}{\mathsf{S}}
\newcommand{\worst}{\mathsf{worst}}
\newcommand{\Wst}{\mathsf{wst}}

 \newcommand{\lit}{\mathsf{lit}}
\newcommand{\enn}{r}
\newcommand{\emm}{s}
 
\newcommand{\capR}{\mathsf{Q}}

\newcommand{\decprob}[3]{
  \smallskip
  {\centering
    \begin{minipage}{0.94\linewidth}%
      \textsc{#1}\\[0.2ex]
      \textbf{Input:} #2\\[0.2ex]
      \textbf{Question:} #3
    \end{minipage}%
  }
  \smallskip
}

\newcommand{\decprobnormal}[3]{
  \begin{center}%
    \begin{minipage}{0.9\linewidth}%
      \textsc{#1}\\[0.2ex]
      \textbf{Input:} #2\\[0.2ex]
      \textbf{Question:} #3
    \end{minipage}%
  \end{center}
}

\usepackage{tikz}
\usetikzlibrary{decorations,arrows,petri,topaths,backgrounds,shapes,positioning,fit,calc,decorations.pathreplacing,patterns,intersections,decorations.pathmorphing,matrix}

\makeatletter
\newcommand{\gettikzxy}[3]{%
  \tikz@scan@one@point\pgfutil@firstofone#1\relax
  \edef#2{\the\pgf@x}%
  \edef#3{\the\pgf@y}%
}
\makeatother

 \tikzstyle{profilestyle} = [ampersand replacement=\&,column sep=-4pt, row sep=-2pt]

\usepackage{cleveref}
\newtheorem{theorem}{Theorem}
\newtheorem{corollary}{Corollary}
\newtheorem{lemma}{Lemma}
\newtheorem{claim}{Claim}
\newtheorem{observation}{Observation}
\newtheorem{proposition}{Proposition}

\theoremstyle{definition}

\newtheorem{example}{Example}

\crefname{table}{Table}{Tables}
\crefname{figure}{Figure}{Figures}
\crefname{theorem}{Theorem}{Theorems}
\crefname{definition}{Definition}{Definitions}
\crefname{corollary}{Corollary}{Corollaries}
\crefname{observation}{Observation}{Observations}
\crefname{lemma}{Lemma}{Lemmas}
\crefname{example}{Example}{Examples}
\crefname{reduction}{Reduction}{Reductions}
\crefname{construction}{Construction}{Constructions}
\crefname{subsection}{Subsection}{Subsections}
\crefname{section}{Section}{Sections}
\crefname{proposition}{Proposition}{Propositions}
\crefname{algorithm}{Algorithm}{Algorithms}
\crefname{drule}{Rule}{Rules}
\crefname{claim}{Claim}{Claims}

\newcommand{\diverseties}{\textsc{SMTI-Diverse}\xspace}
\newcommand{\diverse}{\textsc{SMI-Diverse}\xspace}
\newcommand{\feasible}{\textsc{FI-Diverse}\xspace}
\newcommand{\CSM}{\textsc{Classified Stable Matching}\xspace}

\newcommand{\threesat}{\textsc{$\forall\exists$3SAT}\xspace}
\newcommand{\ethreesat}{\textsc{3SAT}\xspace}
\newcommand{\oneinthree}{\textsc{1-in-3-$\forall\exists$3SAT}\xspace}
\newcommand{\notoneinthree}{\textsc{Not-1-in-3-$\exists\forall$3SAT}\xspace}

\newcommand{\myemph}[1]{\emph{#1}}

\usepackage{xcolor}
\definecolor{dargray}{rgb}{0.18, 0.18, 0.18}
\definecolor{darkgreen}{rgb}{0.01,0.6,0.1}
\definecolor{lightrose}{rgb}{0.996,0.75,0.793}
\definecolor{rose}{cmyk}{0.75, 0.75, 0,0}
\definecolor{winered}{rgb}{0.6,0.1,0.1}
\definecolor{darkyellow}{rgb}{.99, .87, 0.04}
\definecolor{lightyellow}{rgb}{1, 1, 0.6}
\definecolor{transparent}{rgb}{1,1,1}
\definecolor{lightlightgray}{rgb}{0.88, 0.88, 0.88}
\definecolor{lightgray}{rgb}{0.8, 0.8, 0.8}
\definecolor{lightblue}{rgb}{0.527,0.805,0.977}
\definecolor{lightgreen}{rgb}{.74,1,0}
\definecolor{darkblue}{rgb}{0,0,0.4}

\usepackage{graphicx}

\newcommand{\known}[1]{\textcolor{darkgray}{#1}}

\newcommand{\npremark}{\small $^\clubsuit$}
\newcommand{\aamascite}{\small $^\diamondsuit$}
\newcommand{\presult}[1]{\cellcolor{green!20}#1}

\newcommand{\fptresult}[1]{\cellcolor{green!20}#1}
\newcommand{\xpresult}[1]{\cellcolor{green!20}#1}
\newcommand{\npresult}[1]{\cellcolor{red!20}#1}
\newcommand{\sigmapresult}[1]{\cellcolor{red!40}#1}

\newcommand{\highlight}[1]{\cellcolor{lightgray}#1}
\newcommand{\true}{\mathsf{true}}
\newcommand{\false}{\mathsf{false}}

\newcommand{\lowervec}{\ensuremath{\mathsf{\ell}}}
\newcommand{\uppervec}{\ensuremath{\mathsf{u}}}
\newcommand{\umax}{\ensuremath{\mathsf{u}_{\infty}}}
\newcommand{\lmax}{\ensuremath{\mathsf{\ell}_{\infty}}}
\newcommand{\qmax}{\ensuremath{\mathsf{q}_{\infty}}}

\newcommand{\typevec}{\ensuremath{\tau}}
\newcommand{\acset}{\ensuremath{\mathsf{A}}}
\newcommand{\osucc}{\ensuremath{\, }}

\makeatletter
\g@addto@macro\bfseries{\boldmath}
\makeatother

\tikzstyle{blueline} = [thick, blue, dotted]
\tikzstyle{redline} = [line width=2pt, red, dashed]
\tikzstyle{blackline} = [thick, black]

\title{Stable Matchings with Diversity Constraints: Affirmative Action is beyond NP}

\author{
Jiehua Chen
\and
Robert Ganian
\and
Thekla Hamm
\affiliations
Vienna University of Technology\\
\emails
jiehua.chen@tuwien.ac.at, \{rganian, thamm\}@ac.tuwien.ac.at,
}

\begin{document}

\maketitle

\begin{abstract}
   We investigate the following many-to-one stable matching problem with diversity constraints~(\diverseties): Given a set of students and a set of colleges which have preferences over each other, where the students have overlapping types, and the colleges each have a total capacity as well as quotas for individual types (the diversity constraints),
   is there a matching satisfying all diversity constraints such that no unmatched student-college pair has an incentive to deviate?
%
  
  \diverseties{} is known to be NP-hard. 
  However, as opposed to the NP-membership claims in the literature~\cite{AzizGaspersSunWalsh2019aamas,Huang2010classifiedSM}, we prove that it is beyond NP: it is complete for the complexity class~$\Sigma^{\text{P}}_2$.
  In addition, we provide a comprehensive
  analysis of the problem's complexity from the viewpoint of natural restrictions to inputs and obtain new algorithms for the problem.
\end{abstract}

\section{Introduction}\label{sec:intro}
\iflong
Stability is a classic and central property of assignments, or \myemph{matchings}, of agents to each other,
describing that no two agents actively prefer each other to their respective situations in the matching. Stability is desirable in many
scenarios and spawned numerous works in various context~\cite{Manlove2013}. 
In this work we investigate the notion of stability in combination with diversity, which is key
in many real-world matching applications, ranging from education, through health-care systems, to job and
housing markets~[\citeauthor{Abdul2005-college-affirmative}, \citeyear{Abdul2005-college-affirmative}; \citeauthor{Huang2010classifiedSM}, \citeyear{Huang2010classifiedSM}; \citeauthor{KamKoj2015}, \citeyear{KamKoj2015}; \citeauthor{KurHamIwakiYok2017controlledschoolchoice}, \citeyear{KurHamIwakiYok2017controlledschoolchoice}; \citeauthor{AhmeDickFuge2017diverseBmatching}, \citeyear{AhmeDickFuge2017diverseBmatching}; \citeauthor{BenChaHoSliZic2019houseallocation}, \citeyear{BenChaHoSliZic2019houseallocation}; \citeauthor{GonNisKovRom2019}, \citeyear{GonNisKovRom2019}; \citeauthor{AzizGaspersSunWalsh2019aamas}, \citeyear{AzizGaspersSunWalsh2019aamas}].

For this we conceptually distinguish two sets---a set of \myemph{students}
which should be matched to a set of \myemph{colleges} (each with a maximum \myemph{capacity} to accommodate students) with the additional
constraint that the set of students matched to any single college has to
be diverse. The diversity requirements are captured by \myemph{types} which are
attributes that a student may or may not have, and \myemph{upper and lower
quotas} that specify how many students of a certain type may be matched to a given college.
The terminology arises from the context of controlled public school choice, a typical application of this paradigm where it is desirable to match colleges to students to ensure stability as well as demographic, socio-economic, and ethnic diversity (see also \emph{affirmative action}).

As an illustration, assume that there are four students~$u_1,\ldots, u_4$ and two colleges~$w_1,w_2$ as depicted in Table~\ref{tab:example}.
\noindent
\begin{table}
\begin{tikzpicture}[scale=.95, every node/.style={scale=.95}]
  \matrix[profilestyle, column 1/.style={nodes={text width=5.5ex, align=center, minimum height=3ex}}, column 2/.style={nodes={align=left}}, column 8/.style={nodes={text width=3.5ex, align=center,minimum height=3ex}}] (ex) 
  {\node (T) {Types}; \& \node (SPref) {S.Pref.}; \& \node[] (S) {S.}; \& \node[text width=4ex] {}; \& \node (C) {C.}; \& \node (CPref) {C.Pref.}; \& \node (Q) {Quotas}; \& \node (C) {C.}; \\
    \node (T1) {--L}; \& \node (prefu1) {$w_1 \!\succ\! w_2$}; \& \node[] (u1) {$~u_1~$};
    \& \&  \\
    \node (T2) {--L}; \& \node (prefu2) {$w_1 \!\succ\! w_2$}; \& \node[] (u2) {$~u_2~$}; \& \& \node[] (w1) {$~w_1~$};  \& \node (prefw1) {$u_3 \!\succ\! u_1 \!\succ\! u_2$}; \& \node (LQ1) {\small $\ge$1F, $\ge$1L}; \& \node (q1) {$2$}; \\
    \node[text width=5ex] (T3) {FL}; \& \node (prefu3) {$w_2\!\succ\!w_1$}; \& \node[] (u3) {$~u_3~$};
     \& \& \node[] (w2) {$~w_2~$};  \& \node (prefw2) {$u_1\!\succ\!u_3\!\succ\!u_4 \!\succ\! u_2$};  \& \node (LQ2) {\small =1F, =1L}; \& \node[] (q2) {$2$}; \\
     \node (T4) {F--}; \& \node (prefu4) {$w_2$}; \& \node[] (u4) {$~u_4~$}; \& \& \& \& \& \node (d) {};\\
  };
     
  \foreach \s/\t in { u1/w2,
    u2/w1,
    u3/w1,u4/w2%
  } {
    \draw[redline] ($(\s.east)+(-.1,0)$) -- ($(\t.west)+(0.1,0)$);
  }

  \foreach \s/\t in {u1/w1,u3/w1,u2/w2,u4/w2} {
    \draw[blackline] ($(\s.east)+(-0.1,0)$) -- ($(\t.west)+(0.1,0)$);
  }
  
\end{tikzpicture}
\vspace{-0.5cm}
\caption{Each of the students $u_1,\ldots,u_4$ has preferences over the two colleges (where $w_1\succ w_2$ indicates that $w_1$ is preferred to $w_2$) and their types (\textbf{F}emale, \textbf{L}ocal). Each of the colleges has preferences over the students, \myemph{Quotas} for individual types and a \myemph{Capacity} to accommodate students.}
\label{tab:example}
\vspace{-0.5cm}
\end{table}
In terms of \myemph{classical stability}, one could match~$u_1$ and~$u_2$ to~$w_1$, and~$u_3$ and~$u_4$ to~$w_2$,
without inducing any \myemph{blocking pairs}, i.e., an unmatched student-college pair~$\{u,w\}$
such that
\begin{inparaenum}[(1)]
  \item $u$ is unmatched or strictly prefers~$w$ to its assigned college~$M(u)$ and
  \item $w$ can either accommodate~$u$ (without exceeding the capacity) or strictly prefers~$u$ to at least one member in $M(w)$.
\end{inparaenum}
However, to ensure affirmative action, among all assignees to~$w_1$ there must be \emph{at least one} female student and \emph{one} local student,
and~$w_2$ must receive \emph{exactly one} female student and \emph{exactly one} local (as indicated by the respective quotas in Table~\ref{tab:example}).
A \myemph{feasible} matching~$M_1$ (see the black solid lines) could be to match~$u_1$ and~$u_3$ to~$w_1$, and~$u_2$ and $u_4$ to~$w_2$.
However, this is \emph{not} stable since $w_2$ may substitute its assignees~$u_2$ and~$u_4$ with~$u_3$ so as to obtain a more preferred student while maintaining its diversity constraints.
Another feasible matching~$M_2$~(see the red dashed lines), which matches~$u_2$ and~$u_3$ to~$w_1$, and~$u_1$ and~$u_4$ to~$w_2$, fulfills diversity constraints and is stable (see \cref{sec:prelim} for formal definitions).
Note that stable matchings fulfilling the diversity constraints are not guaranteed to exist.
For instance, if student~$u_2$ in the above example does not find $w_1$ acceptable at all, then matching~$M_1$ is the only feasible matching with diversity.
However, it is not stable.
\fi
\ifshort
Stability is a classic and central property of assignments, or \myemph{matchings}, of agents to each other,
describing that no two agents actively prefer each other to their respective situations in the matching. Stability is desirable in many
scenarios and spawned numerous works in various context~\cite{Manlove2013}. 
In this work we investigate the notion of stability in combination with diversity, which is key
in many real-world matching applications, ranging from education, through health-care systems, to job and
housing markets~[\citeauthor{Abdul2005-college-affirmative}, \citeyear{Abdul2005-college-affirmative}; \citeauthor{Huang2010classifiedSM}, \citeyear{Huang2010classifiedSM}; \citeauthor{KamKoj2015}, \citeyear{KamKoj2015}; \citeauthor{KurHamIwakiYok2017controlledschoolchoice}, \citeyear{KurHamIwakiYok2017controlledschoolchoice}; \citeauthor{AhmeDickFuge2017diverseBmatching}, \citeyear{AhmeDickFuge2017diverseBmatching}; \citeauthor{BenChaHoSliZic2019houseallocation}, \citeyear{BenChaHoSliZic2019houseallocation}; \citeauthor{GonNisKovRom2019}, \citeyear{GonNisKovRom2019}; \citeauthor{AzizGaspersSunWalsh2019aamas}, \citeyear{AzizGaspersSunWalsh2019aamas}].
\todoR{Remove green emph before submission...}

For this we conceptually distinguish two sets---a set of \myemph{students}
which should be matched to a set of \myemph{colleges} (each with a maximum \myemph{capacity} to accommodate students) with the additional
constraint that the set of students matched to any single college has to
be diverse. The diversity requirements are captured by \myemph{types} which are
attributes that a student may or may not have, and \myemph{upper and lower
quotas} that specify how many students of a certain type may be matched to a given college.
The terminology arises from the context of controlled public school choice, a typical application of this paradigm where it is desirable to match colleges to students to ensure stability as well as demographic, socio-economic, and ethnic diversity (see also \emph{affirmative action}).
\fi

The study of stable matchings with diversity constraints was initiated by Abdulkadiro\v glu~\shortcite{Abdul2005-college-affirmative} in the context of college admissions.
It has since become an ongoing and actively researched topic among economists and computer scientists, covered for example by two chapters
~\cite{Heo2019equity-diversity,Kojima2019constraints} in the recently published book ``On the Future of Economic Design''~\cite{LasMouSanZwi2019future}.
One of the fundamental questions in this area is whether there is a diverse and stable matching between students and colleges; 
the corresponding computational problem is called \diverseties~(see \cref{sec:prelim} for formal definitions).

As has already been observed in the pioneering work of Aziz, Gaspers, Sun and Walsh~\shortcite{AzizGaspersSunWalsh2019aamas} and hinted at in Huang's earlier work on a closely related problem~\shortcite{Huang2010classifiedSM}, \diverseties\ is \NP-hard.
The authors further claimed that the problem(s) under consideration belong to NP (see also \cite[Chapter 5.2.5]{Manlove2013}).
We disprove this claim by presenting an involved reduction showing that the problem is in fact complete for the complexity class~$\Sigma^{\text{P}}_2$~(\cref{thm:diverse:lq>0:sigma2p-c}), even under severe restrictions to the input instances.

Complementing this hardness finding, we systematically analyze the complexity of the problem by considering natural relaxations (such as dropping lower quotas or dropping stability) or restrictions (such as bounding the number~$n$ of students, the number~$t$ of types, the number~$m$ of colleges, and/or the maximum upper quota~$\umax$, and the maximum capacity~$\qmax$). The outcome of our analysis is a full classification of the complexity of \diverseties\ w.r.t.\ the considered restrictions and relaxations, presented in Table~\ref{tab:overview}.
We highlight three key technical contributions of our work:

\newcommand{\The}[1]{{\small Th~\ref{#1}}}
\newcommand{\Pro}[1]{{\small Pr~\ref{#1}}}
\newcommand{\Cor}[1]{{\small Co~\ref{#1}}}
\newcommand{\Obs}[1]{{\small Ob~\ref{#1}}}
\iflong
\begin{table}[t!]
  \centering
 \extrarowheight=.8\aboverulesep
  \addtolength{\extrarowheight}{\belowrulesep}
  \aboverulesep=0pt
  \belowrulesep=0pt
  \resizebox{\columnwidth}{!}{
\begin{tabular}{|@{\,}l|l@{}>{\columncolor{white}[0pt][\tabcolsep]}l@{}|l@{}l|l@{}l@{}>{\columncolor{white}[\tabcolsep][2pt]}l@{\,}|}
  \toprule
  Problems & \multicolumn{2}{@{\;}c@{\;}|}{\feasible} & \multicolumn{5}{c|}{\diverseties}  \\\midrule
  Constraints  & \multicolumn{2}{c|}{$(\lmax \geq 0)$} & \multicolumn{2}{c|}{$(\lmax \geq 0)$} &  \multicolumn{3}{c|}{$(\lmax = 0)$} \\\midrule
  Complexity & \multicolumn{2}{c|}{\known{NP-c\aamascite}} &  \sigmapresult{$\Sigma^{\text{P}}_2$-c} & \sigmapresult{[\The{thm:diverse:lq>0:sigma2p-c}]} & \multicolumn{2}{l}{\npresult{NP-c}} & \npresult{[\The{thm:diverse:NP-c:t=2}]} \\\hline
  $m+\umax$ & \multicolumn{2}{c|}{\known{NP-c\aamascite}} & \sigmapresult{$\Sigma^{\text{P}}_2$-c}  & \sigmapresult{[\The{thm:diverse:lq>0:sigma2p-c}]} & \multicolumn{2}{l}{\npresult{NP-c}} & \npresult{[\The{thm:lq=0:m=4:diverse-NP-c}]} \\\hline
  $t+\umax+\qmax$ & \npresult{NP-c} & \npresult{[\Pro{prop:feasible-t+umax+qmax-NP-h}]}  & \npresult{NP-c\npremark} & \npresult{[\The{thm:diverse:NP-c:t=2},\Obs{obs:t-or-qmax-NP}]}  & \multicolumn{2}{l}{\npresult{NP-c}} & \npresult{[\The{thm:diverse:NP-c:t=2}]} \\\hline
  $n$ & \presult{P} & \fptresult{[\The{thm:fptn}]}& \presult{P} & \fptresult{[\The{thm:fptn}]}& \multicolumn{2}{l}{\presult{P}} & \fptresult{[\The{thm:fptn}]}\\\hline
  $m + t$ & \presult{P} & \xpresult{[\Cor{cor:feasible-m+t-fpt}]} & \presult{P} &\xpresult{[\The{thm:diverse-m+t-XP}]} & \multicolumn{2}{l}{\presult{P}} &\xpresult{[\The{thm:diverse-m+t-XP}]} \\\hline
  $m+\qmax$ & \presult{P} & \xpresult{[\Pro{prop:diverseties-m+qmax:XP}]}  & \presult{P} & \xpresult{[\Pro{prop:diverseties-m+qmax:XP}]}  & \multicolumn{2}{l}{\presult{P}} & \fptresult{[\Pro{prop:diverseties-m+qmax:XP}]}\\ 
  \bottomrule
 \end{tabular}}
\caption{A complete picture of the complexity results for \feasible and \diverseties (see \cref{sec:prelim} for the definitions).
  Results marked with \aamascite{} are due to~\protect\cite[Proposition~5.1]{AzizGaspersSunWalsh2019aamas} while the remaining ones are new.
  All hardness results hold even for preferences with \emph{no} ties, even if the corresponding measures are upper-bounded by a constant. 
    The NP-containment result marked with \npremark{} holds already when either $t$ or $\qmax$ is a constant.
  The problem variants for which fixed-parameter tractability results (FPT) exist do not admit polynomial-size kernels~(see \cref{prop:nopoly-n,prop:nopoly-m+t+qmax}).}\label{tab:overview} 
\end{table}
\fi
\ifshort
\begin{table}
  \extrarowheight=.8\aboverulesep
  \addtolength{\extrarowheight}{\belowrulesep}
  \aboverulesep=0pt
  \belowrulesep=0pt
  \resizebox{\columnwidth}{!}{
    \begin{tabular}{|@{}l|@{}>{\columncolor{white}[0pt][\tabcolsep]}l@{\;\;}>{\columncolor{white}[0pt][\tabcolsep]}l@{}|@{}>{\columncolor{white}[0pt][\tabcolsep]}l@{\;}>{\columncolor{white}[0pt][\tabcolsep]}l@{}|@{}>{\columncolor{white}[0pt][\tabcolsep]}l@{\;\;}>{\columncolor{white}[\tabcolsep][0pt]}l@{}|}
    \toprule
  Problems & \multicolumn{2}{@{\;}c@{\;}|}{\feasible} & \multicolumn{4}{c|}{\diverseties}  \\\midrule
  Constraints  & \multicolumn{2}{c|}{$(\lmax \geq 0)$} & \multicolumn{2}{c|}{$(\lmax \geq 0)$} &  \multicolumn{2}{c|}{$(\lmax = 0)$} \\\midrule
  Complexity & \multicolumn{2}{c|}{\known{~NP-c\aamascite}} &  \sigmapresult{~$\Sigma^{\text{P}}_2$-c} & \sigmapresult{[\The{thm:diverse:lq>0:sigma2p-c}]} & \npresult{~NP-c} & \npresult{[\The{thm:diverse:NP-c:t=2}]} \\\hline
  $m+\umax$ & \multicolumn{2}{c|}{\known{~NP-c\aamascite}} & \sigmapresult{~$\Sigma^{\text{P}}_2$-c}  & \sigmapresult{[\The{thm:diverse:lq>0:sigma2p-c}]} & \npresult{~NP-c} & \npresult{[\The{thm:lq=0:m=4:diverse-NP-c}]} \\\hline
  $t\!+\!\umax\!+\!\qmax$ & \npresult{~NP-c} & \npresult{[\Pro{prop:feasible-t+umax+qmax-NP-h}]}  & \npresult{~NP-c\npremark} & \npresult{[\The{thm:diverse:NP-c:t=2},\Obs{obs:t-or-qmax-NP}]}  & \npresult{~NP-c} & \npresult{[\The{thm:diverse:NP-c:t=2}]} \\\hline
  $n$ & \presult{~P} & \presult{[\The{thm:fptn}]}  & \presult{~P} & \presult{[\The{thm:fptn}]}  & \presult{~P} & \presult{[\The{thm:fptn}]}  \\ \hline
  $m + t$ & \presult{~P} & \presult{[\The{thm:diverse-m+t-XP}]} & \presult{~P} &\presult{[\The{thm:diverse-m+t-XP}]} & \presult{~P} &\presult{[\The{thm:diverse-m+t-XP}]} \\\hline
    $m+\qmax$ & \presult{~P} & \presult{[\Pro{prop:diverseties-m+qmax:XP}]}  & \presult{~P} & \presult{[\Pro{prop:diverseties-m+qmax:XP}]}  & \presult{~P} & \presult{[\Pro{prop:diverseties-m+qmax:XP}]}\\ 
\bottomrule
\end{tabular}}
\caption{A complete picture of the complexity results for \feasible and \diverseties (see \cref{sec:prelim} for the definitions). Results marked with \aamascite{} are due to~\protect\cite[Prop~5.1]{AzizGaspersSunWalsh2019aamas} while the remaining ones are new.
  All hardness results hold even for preferences with \emph{no} ties, even if the corresponding measures are upper-bounded by a constant.
  The NP-containment result marked with \npremark{} holds already when either $t$ or $\qmax$ is a constant~(see \cref{obs:t-or-qmax-NP}).\vspace{-0.5cm}}\label{tab:overview} 
\end{table}
\fi
\iflong
\begin{compactenum}[(1)]
  \item \diverseties{} is $\Sigma^{\text{P}}_2$-complete even when the preferences do not have ties and there are only four colleges, while two natural relaxations of the problem (either dropping the lower quotas or the stability requirement) lower the complexity to being NP-complete.
  \item When (a) the number of students, (b) the number of colleges and types, or (c) the number of colleges and maximum capacity is bounded by a constant, \diverseties\ can be solved in polynomial time. 
  \item \diverseties{} is NP-complete even when the preferences do not have ties, lower quotas are all zero, and the number of types and the maximum upper quota and maximum capacity are bounded by a constant. Our reduction showing this result also fixes a technical flaw 
in the reduction presented in Aziz et al.'s recent work~\cite[Proposition~5.3]{AzizGaspersSunWalsh2019aamas}.
\end{compactenum}
\fi

\ifshort
\paragraph{Related Work.}
For one type, \diverseties is equivalent to the Hospitals/Residents with Lower Quotas problem where no hospital is allowed to be closed (HR-LQ-2), as studied by \cite{HamadaIwamaMiyazaki2016algorithmica}.
This problem is polynomial-time solvable when no ties are allowed~\cite[Chapter 5.2.3]{Manlove2013}.
We show that \diverseties{} becomes NP-hard even for only two types.

Huang~\shortcite{Huang2010classifiedSM} introduced the closely related \CSM~(CSM) problem, which asks for a matching that fulfills the diversity constraints and does not admit \emph{blocking coalitions}. 
We show that our $\Sigma^{\text{P}}_2$-hardness reduction can be adapted to show that CSM is also $\Sigma^{\text{P}}_2$-complete.

Aziz et al.~\shortcite{AzizGaspersSunWalsh2019aamas} studied school choice with diversity constraints, but with a slightly different stability condition: an unmatched student-college pair~$\{u,w\}$ is \myemph{d-blocking} if it is a blocking pair (in our sense) and the new solution fulfills the diversity constraints for \emph{all} remaining colleges instead of~$w$ only.
This means that a d-blocking pair is also a blocking pair, but the converse is not true. 
However, dropping the lower quotas requirements renders both concepts equivalent.
Our model of blocking pairs is a natural extension of the HR-LQ-2 problem where a student and a college already form a blocking pair once the new solution is better for them, regardless of whether the other colleges' lower quotas.
Such kind of model assumes that blocking condition is tested based on local information of whether the deviating college's diversity constraints are fulfilled after the rematching.
This is a standard assumption in many controlled school choice articles~\cite{Abdul2005-college-affirmative,KurHamIwakiYok2017controlledschoolchoice,HamadaIwamaMiyazaki2016algorithmica}.
Nevertheless, our $\Sigma^{\text{P}}_2$-hardness reduction also establishes the same hardness of their variant.

Other related work includes recent papers
by Nguyen and Vohra~\shortcite{NguVoh2019} (who studied stable matching with proportionality constraints on the lower quotas), Kurata et al.~\shortcite{KurHamIwakiYok2017controlledschoolchoice} (who studied stable matching with overlapping types and consider the diversity constraints as soft bounds), and Ismaele et al.~\shortcite{IsmHamZhaSuzYok2019} (who introduced and studied weighted matching markets with budget constraints).

\smallskip

{\noindent {\emph{%
      Some proofs are replaced with proof sketches due to space constraints.
      We can provide them immediately upon request.
    }}}

\fi

\iflong
\subsection{Related Work}
If the number of types is equal to one, then \diverseties is equivalent to the Hospitals/Residents with Lower Quotas problem where no hospital is allowed to be closed (HR-LQ-2), as studied by \cite{HamadaIwamaMiyazaki2016algorithmica}.
This problem is polynomial-time solvable when no ties are allowed~\cite[Chapter 5.2.3]{Manlove2013}.
We show that \diverseties{} becomes NP-hard even for only two types.
The problem variant (HR-LQ-1) where each hospital is allowed to be closed (i.e., receive no residents) was introduced by \cite{BiroFleinerIrvingManlove2010} and  proven to be NP-complete even if each upper and lower quota is equal to three. Our \diverseties{} is different from this problem as we do not allow colleges to be closed.

Huang~\shortcite{Huang2010classifiedSM} introduced the closely related \CSM~(CSM) problem, which asks for a matching that fulfills the diversity constraints and does not admit the so-called \emph{blocking coalitions}.
Huang showed that CSM is NP-hard and further claimed that it is in NP. However, while blocking coalitions and blocking pairs are not comparable in general, we show that our $\Sigma^{\text{P}}_2$-hardness reduction can be adapted to show that CSM is indeed beyond NP: it is also $\Sigma^{\text{P}}_2$-complete.

Aziz et al.~\shortcite{AzizGaspersSunWalsh2019aamas} studied school choice with diversity constraints, but with a slightly different stability condition: an unmatched student-college pair~$\{u,w\}$ is \myemph{d-blocking} if it is a blocking pair (in our sense) and the new solution fulfills the diversity constraints for \emph{all} remaining colleges instead of~$w$ only.
This means that a d-blocking pair is also a blocking pair, but the converse is not true. 
However, dropping the lower quotas requirements renders both concepts equivalent.
Our model of blocking pairs is a natural extension of the HR-LQ-2 problem where a student and a college already form a blocking pair once the new solution is better for them, regardless of whether the other colleges' lower quotas.
Such kind of model assumes that blocking condition is tested based on local information of whether the deviating college's diversity constraints are fulfilled after the rematching.
This is a standard assumption in many controlled school choice articles~\cite{Abdul2005-college-affirmative,KurHamIwakiYok2017controlledschoolchoice,HamadaIwamaMiyazaki2016algorithmica}.
Nevertheless, our $\Sigma^{\text{P}}_2$-hardness reduction also establishes the same hardness of their variant.

\citeauthor{Abdul2005-college-affirmative}~\shortcite{Abdul2005-college-affirmative}, \citeauthor{HafTYenYil2013}~\shortcite{HafTYenYil2013}, \citeauthor{Tomoeda2018typespecific-sm}~\shortcite{Tomoeda2018typespecific-sm}, and \citeauthor{Bo2016typespecific-diverse}~\shortcite{Bo2016typespecific-diverse} studied \diverseties{} for the specific case where every student and every college has complete and strict preferences and every student belongs to exactly one type.

\citeauthor{NguVoh2019}~\shortcite{NguVoh2019} studied stable matching with proportionality constraints on the lower quotas.
\citeauthor{KurHamIwakiYok2017controlledschoolchoice}~\shortcite{KurHamIwakiYok2017controlledschoolchoice} study stable matching with overlapping types and consider the diversity constraints as soft bounds.
\citeauthor{IsmHamZhaSuzYok2019}~\shortcite{IsmHamZhaSuzYok2019} introduce and study weighted matching markets with budget constraints in labor market, where the students have preferences over the colleges and have wage requirements, and the colleges have preferences over subsets of the students and each have a budget limit.
The goal is to decide whether there is a coalitionally stable matching, a matching which satisfies the budget constraints and does not admit blocking coalitions.
Our \diverseties{} problem differs from theirs in two ways: (1) while they allow colleges to specify preferences over subsets of students which increases the complexity in the input, we do not; (2) they allow both the students and the colleges to have some form of weights, we do not. They show that their problem is in general $\Sigma^{\text{P}}_2$-complete. The reduction, however, heavily utilizes (1) and (2). Thus, it is not obvious how to encode numbers and preferences over subsets to make their reduction adaptable for us.

Finally, we mention that work in combination with diversity has been done on hedonic games~\cite{BreElkIga2019}, multi-winner elections~\cite{BreFalIgaLacSko2018}, and matchings~\cite{AhmeDickFuge2017diverseBmatching}.

\fi

\section{Preliminaries}\label{sec:prelim}
Given an integer~$z$, we use~\myemph{$[z]$} to denote the set $\{1,2,\ldots, z\}$.
Given two integer vectors~$\boldsymbol{x}, \boldsymbol{y}$ of the same dimension, i.e., $\boldsymbol{x},\boldsymbol{y} \in \mathds{Z}^{z}$ for a non-negative integer~$z$, we write \myemph{$\boldsymbol{x}\le \boldsymbol{y}$} if for each index~$i\in [z]$ it holds that $\boldsymbol{x}[i]\le \boldsymbol{y}[i]$; otherwise, we write \myemph{$\boldsymbol{x} \not\le \boldsymbol{y}$.} 
\iflong

\fi
A \myemph{preference list}~$\succeq$ over a set~$A$ is a complete, transitive binary relation on~$A$. 
We use \myemph{$\succ$} to denote the asymmetric part ({i.e.}, $x\succeq y$ and $\neg (y\succeq x)$)
and \myemph{$\sim$} to denote the symmetric part of $\succeq$ ({i.e.}, $x\succeq y$ and $y \succeq x$).
We say that~$x$ is \myemph{(strictly) preferred} (resp.\ \myemph{weakly preferred}) to~$y$ if $x\succ y$ (resp.\ $x\succeq y$),
and that $x$ and $y$ are \myemph{tied} in $\succeq$ if $x\sim y$; $\succeq$ is said to \myemph{contain ties} in this case.
We write~$[A]$ to denote an arbitrary but fixed linear order on~$A$ (i.e., a preference list without ties).
The expression ``\myemph{$x\succeq Y$}'' (resp.\ ``\myemph{$x\succ Y$}'') means that $x$ is weakly (resp.\ strictly) preferred to every one in $Y$.

\paragraph{Problem-Specific Terminology.}
The problem we study has as input a set $T\coloneqq [t]$ of types,
a set $U\coloneqq \{u_1,u_2,\ldots, u_n\}$ of $n$~students and a set $W\coloneqq \{w_1,w_2,\ldots,w_m\}$ of $m$~colleges together with the following information.
\iflong

\fi
\noindent Each student~$u$$\in$\,$U$ has
\ifshort\begin{inparaenum}[(i)]\else \begin{compactenum}[(i)]\fi
    \item a \myemph{preference list~$\succeq_{u}$} over a subset~$\acset(u)\subseteq W$ of the colleges, 
    and
    \item a \myemph{type vector~$\typevec_{u}$ $\in \{0,1\}^t$}, where $\typevec_{u}[z]=1$ means that $u$ has type~$z$.
    \ifshort\end{inparaenum} \else\end{compactenum}\fi
Each college~$w\in W$ has
\ifshort\begin{inparaenum}[(i)]\else \begin{compactenum}[(i)] \fi
    \item a \myemph{preference list~$\succeq_{w}$} over a subset~$\acset(w) \subseteq U$ of the students, and
    \item a \myemph{lower-quota} and \myemph{upper-quota} for each type which is described, respectively, via the vectors~\myemph{\(\lowervec_{w}\)} and \myemph{\(\uppervec_{w}\)} \(\in [n]^t\), where $\lowervec_{w}\le \uppervec_{w}$, and
    \item a \myemph{capacity~\(q_{w}\)} \(\in [n]\) which is the maximum number of students allowed to be admitted to~$w$. Note that while the capacity can be modeled by introducing an extra type, leaving it separate from types allows for a more refined analysis of the problem's complexity.
    \ifshort\end{inparaenum}\else \end{compactenum}
We say that the input contains \myemph{ties} if it contains a preference list~$\succeq$ with ties; otherwise it contains no ties. 
\fi


For each $x\in U\cup W$, we call \myemph{$\acset(x)$} the \myemph{acceptable set} of~$x$, which contains all students or colleges that are acceptable to~$x$.
Throughout the paper, we assume that no student or college has an empty acceptable set, and
for each student~$u$ and each college~$w$
it holds that $u\in \acset(w)$ iff.\ $w\in \acset(u)$.

A \myemph{matching}~$M$ is a set of student-college pairs of the form~$\{u,w\}$, where each student~$u$ is involved in \emph{at most one} pair in~$M$ and $w \in \acset(u)$. 
If $\{u,w\}\in M$, then we say that $u$ and $w$ are \myemph{assigned} to each other by~$M$.
Slightly abusing the notation, given a student~$u\in U$ if there exists a college~$w\in W$ with $\{u,w\}\in M$, then we let \myemph{\(M(u)\)} \(\coloneqq w\); otherwise we let $M(u)\coloneqq \bot$.
We assume that each student~$u$ prefers an acceptable college~$w\in \acset(u)$ to~$\bot$.
Similarly, given a college~$w\in W$, we write \myemph{\(M(w)\)} \(\coloneqq \{u\mid \{u, w\} \in M\}\) to denote the set consisting of all students assigned to~$w$ by~$M$.

\paragraph{Feasible and Stable Matchings.}
A matching~$M$ is \myemph{feasible} for an instance~$(U, W, T, (\typevec_u,\succeq_u)_{u\in U}, (\succeq_w,\lowervec_{w}, \uppervec_w, q_w)_{w\in W})$ if each college $w\in W$
\begin{inparaenum}[(i)]
  \item is assigned at most $q_{w}$ students, i.e., $|M(w)| \le q_{w}$, and
  \item meets the lower and upper quotas for each type, i.e.,
  $\lowervec_{w} \le \sum_{u\in M(w)}\typevec_u \le \uppervec_{w}$.
\end{inparaenum}


A student~$u$ and a college~$w$ form a \myemph{blocking pair} in a matching~$M$ if: 
\begin{compactenum}[(i)]
  \item $u \in \acset(w)$ and $\{u, w\}\notin M$,
  \item student~$u$ strictly prefers $w$ to~$M(u)$,
  \item there exists a (possibly empty) subset~$U'\subseteq M(w)$ 
  \iflong of students \fi
  such that $w$ strictly prefers~$u$ to each student from $U'$, and
  \item $M \cup \{\{u, w\}\} \setminus \{\{u', w\} \mid u' \in U'\}$ is feasible for~$w$. 
\end{compactenum}
Accordingly, we say that \(U'\) is a \myemph{witness} for \(\{u, w\}\) to block~\(M\).
A matching~$M$ is \myemph{stable} if it has \emph{no} blocking pairs.

\paragraph{Problem Variants.}
Now, we formally state our main problem of interest---the natural generalization of the classical \textsc{Many-to-One Stable Matching with Ties and Incomplete Preferences} (\textsc{SMTI})~\cite{Manlove2013} to incorporate diversity constraints:

\ifshort
\decprob{\diverseties}
{A set~$U$ of $n$ students, a set~$W$ of $m$ colleges, a set~$T$ of types,
  the type vectors and preference lists~$(\typevec_u,\succeq_u)_{u\in U}$ for the students,
  the preference lists, lower-quota vectors, upper-quota vectors, and capacities~$(\succeq_w, \lowervec_w,\uppervec_w,q_w)_{w\in W}$ for the colleges.}
{Is there a feasible and stable matching?}
\else \decprobnormal{\diverseties}
{A set~$U$ of $n$ students, a set~$W$ of $m$ colleges, a set~$T$ of types,
  the type vectors and preference lists~$(\typevec_u,\succeq_u)_{u\in U}$ for the students,
  the preference lists, lower-quota vectors, upper-quota vectors, and capacities~$(\succeq_w, \lowervec_w,\uppervec_w,q_w)_{w\in W}$ for the colleges.}
{Is there a feasible and stable matching?}
\fi

\noindent We use \diverse{} to denote the restriction of \diverseties to the case
\iflong where ties are not present.
\fi
Moreover, we use \feasible\ to denote the problem of deciding whether there is a feasible matching (representing a generalization for \textsc{Feasible Matching with Incomplete Preferences}).

\iflong
The introductory example from \cref{sec:intro} can be depicted as follows, where the first type is about ``being female'' and the second type is about ``being local''. 

\begin{example}\label{ex:intro->prelim}
  Below, we describe the type vectors~(T.) and the preference lists~(Pref.) of the students~(S.) as well as the preference lists~(Pref.), the lower quotas~(LQ.), the upper quotas~(UQ.), and the capacities~(C.) of the colleges from \cref{sec:intro}; \textbf{the preference lists do not have ties and are always ordered by $\succ$}:\\
  \begin{center}
  \begin{tabular}{@{}l@{}l@{\;\;}c@{}|@{\,}l@{\,}l@{\;\;}c@{\,}|@{\,}l@{}l@{\;}c@{\;}c@{\;}c}
    \toprule
    S. & Pref. & T. &    S. & Pref. & T. & C. & Pref. & LQ. & UQ. & C. \\\midrule 
    $u_1\colon$ & $w_1\osucc w_2$ & $01$ & $u_2\colon$ & $w_1\osucc w_2$ & $01$ &   $w_1\colon$& $u_3 \osucc u_1 \osucc u_2$ & $11$ & $22$ & $2$\\
    $u_3\colon$ & $w_2 \osucc w_1$ &$11$&     $u_4\colon$ & $w_2$ &$10$&   $w_2\colon$ & $u_1 \osucc u_3  \osucc u_4 \osucc u_2$ & $11$  & $11$ & $2$\\
    \bottomrule
  \end{tabular}
  \end{center}
As already discussed, there are two feasible matchings~$M_1$ and $M_2$
with $M_1(w_1)$$=$$\{u_1,u_3\}$, $M_1(w_2)$$=$$\{u_2,u_4\}$
and $M_2(w_1)$$=$$\{u_2,u_3\}$, $M_2(w_2)=\{u_1,u_3\}$.
But $M_1$ is blocked by $\{u_3,w_2\}$ because $w_2$ would prefer to replace $\{u_2,u_4\}$ with $u_3$ and $u_3$ prefers~$w_2$ to~$w_1$.
One can verify that $M_2$ is feasible and stable.
If $u_2$ does not accept~$w_2$, then no feasible and stable matching exists.
\end{example}
\fi

\section{How Hard is Diversity?}\label{sec:hard}
\iflong
\subsection{General Complexity}
\fi
\ifshort
\paragraph{General Complexity.}
\fi
Aziz et al.~\shortcite[Proposition 5.1]{AzizGaspersSunWalsh2019aamas} proved that \feasible\ is NP-complete; the hardness result holds even for a single college. 
\ifshort
They also claimed that determining a matching without d-blocking pairs is NP-complete~\cite[Proposition 5.3]{AzizGaspersSunWalsh2019aamas}. 
However, the proof used to show NP-membership is technically flawed---in particular, while the proof claims that ``\emph{Deciding whether a stable outcome exists is in NP, since we can guess an outcome X and check whether X admits blocking pair in polynomial time}'', by adapting the reduction of Aziz et al.~\shortcite[Proposition 5.1]{AzizGaspersSunWalsh2019aamas} we can show that this is impossible unless P~$\subseteq$~coNP. 
\else They also claimed that determining whether there exists a matching without d-blocking pairs is NP-complete~\cite[Proposition 5.3]{AzizGaspersSunWalsh2019aamas}.
However, there are two separate issues with Proposition 5.3 in the aforementioned paper.
First of all, the proof used to show NP-membership is technically flawed---in particular, while the proof claims that ``\emph{Deciding whether a stable outcome exists is in NP, since we can guess an outcome X and check whether X admits blocking pair in polynomial time}'', by adapting the reduction of Aziz et al.~\shortcite[Proposition 5.1]{AzizGaspersSunWalsh2019aamas} we can show that this is impossible unless P~$\subseteq$~coNP. 
\fi


\begin{proposition}\label{prop:checkingMstable:coNP-h}
 \iflong Given a \diverse\ instance and a feasible matching $M$, it is \textnormal{coNP}-hard to decide whether $M$ does \emph{not} admit any blocking pairs or d-blocking pairs.
 \fi
\end{proposition}

\iflong
\begin{proof}
  To show this, we adapt the reduction as given by Aziz et al.~\shortcite[Proposition 5.1]{AzizGaspersSunWalsh2019aamas}.
  (For the sake of completeness, we describe the whole construction here.)
  We remark that in the constructed instance there is only one single college.
  Since for one college the notion of d-blocking pairs as used by \citeauthor{AzizGaspersSunWalsh2019aamas} is equivalent to our notion of blocking pairs we only show the statement for our definition of blocking pairs.

  We reduce an NP-complete variant of the \textsc{Exact Cover by 3-Sets} problem~\cite{GJ79}.
  Note that in this variant each element of the universe appears in exactly three sets~\cite{Gonzalez1985}.
  \decprobnormal{Restricted Exact Cover by 3-Sets~(RX3C)}
  {A finite set~$X$ with $|X|=3q$ and a collection~\(\mathcal{S}\) of $3$-element subsets of $X$, where each element appears in exactly three sets.}
  {Does $\mathcal{S}$ contain an \myemph{exact cover}, i.e., a subcollection~$C\subseteq \mathcal{S}$ such that every element of $X$ occurs in exactly one member of~$C$?}
  
  Let $I=(X, \mathcal{S})$ be an instance of \textsc{X3C}, where $X=\{x_1,\ldots,x_{n^*}\}$.
  We construct an instance of \diverse as follows.
  Let $U=\{s \mid S \in \mathcal{S}\}\cup \{d\}$, $W=\{w\}$, $t=n^*$.
  For each set~$S\in \mathcal{S}$, the corresponding \myemph{set-student}~$s\in U$ has types that correspond to the three elements contained in $S$.
  Formally, for each type~$z\in [t]$ it holds that $\typevec_{s}[z]=1$ if and only if $x_z\in S$.
  The special student~$d$ has all types which is the most preferred student of the single college~$w$.
  The lower quota and upper quota of each type are exactly three.
  We ask whether the feasible matching~$M$, which assigns the single college~$w$ to all students but~$d$,
  admits no blocking pairs.
  Observe that $M$ is indeed feasible because each element in $I$ appeas in exactly three sets.
  

  To show the coNP-hardness, it suffices to show that $I$ has an exact cover if and only if $M$ admits a blocking pair.
 
  If there is an exact cover, then $\{d,w\}$ is a blocking pair because replacing the students corresponding to the exact cover with $d$ results in a feasible matching which is better for $w$ and~\(d\).

  On the other hand, if \(M\) admits a blocking pair, it must be $\{s,w\}$ as this is the only pair not matched in \(M\).
  In order to form a blocking pair with $d$, there has to be some $U' \subseteq M(w)$ which witnesses this.
  Because of the lower and upper quotas on each type and the fact that \(d\) has every type, the witness~$U'$ has to ``cover'' each type exactly once.
  It is straight-forward to verify that $C=\{S \in \mathcal{S} \mid s\in U'\}$ forms an exact cover.
%
\end{proof}
\fi

Note that Proposition~\ref{prop:checkingMstable:coNP-h} itself does not rule out that \diverse is in NP.
It just suggests the given proof is incorrect.
There could in principle be a different non-deterministic algorithm to place the problem in NP.
We show that this is not the case in our main result (Theorem~\ref{thm:diverse:lq>0:sigma2p-c}), by showing \(\Sigma_2^P\)-hardness.
For this we introduce a crucial gadget which is used in several of our reductions\iflong. \fi\ifshort\ throughout this section. \fi
\iflong
For the sake of recognizability the components of the gadget are marked with gray background color, in particular whenever it occurs in reductions later. 
\fi
\begin{lemma}\label{lem:aux}
  \ifshort Let $T=\{1,2\}$, and
  \else  Let there be two types~$T=\{1,2\}$, and
  \fi
  let $U\uplus \{r_1,r_2,r_3\}$ be a set of students with three distinguished students~$r_1,r_2,r_3$,
  and let $W\uplus\{a,b\}$ be a set of colleges with two distinguished colleges~$a$ and $b$.
  Similar to the format given in \ifshort \cref{sec:prelim},
  \else \cref{ex:intro->prelim}, \fi
  the preference lists and type vectors of students~$r_1,r_2$, and $r_3$,
  and the preference lists, the upper quotas, and the capacities of the colleges are depicted as follows:
  \ifshort\\ \else \begin{center}\fi
  \begin{tabular}
    {l@{}>{\columncolor{white}[0pt][0pt]}l@{}>{\columncolor{white}[0pt][\columnsep]}c|ll@{}>{\columncolor{white}[0pt][\columnsep]}c@{}>{\columncolor{white}[0pt][\tabcolsep]}c}
    \toprule
    S. & Pref. & T. & C. & Pref. & ~~~UQ.~~~ & C. \\
    \midrule 
  \highlight{$r_1\colon$} & \highlight{$b \osucc a$} & \highlight{$10$} &  \highlight{$\forall w\in W\colon$}& \highlight{$[U] \osucc r_2$} & \highlight{$11$} & \highlight{$q_w$}\\
   \highlight{$r_2\colon$} & \highlight{$b \osucc [W] \osucc a$~~~} &\highlight{$11$} &   \highlight{$a\colon$} & \highlight{$r_1 \osucc r_2 \osucc r_3$} & \highlight{$11$} & \highlight{$1$}\\
    \highlight{$r_3\colon$} & \highlight{$a \osucc b$} &\highlight{$01$}&  \highlight{$b\colon$} & \highlight{$r_3 \osucc r_2 \osucc r_1$} & \highlight{$11$} & \highlight{$2$}\\
    \bottomrule
  \end{tabular}
  \iflong \end{center}\fi

  \noindent All students in $U$ have zero types and arbitrary but fixed preferences.
  All lower quotas are zero.
  The following holds for every matching~$M$.
  \begin{compactenum}[(1)]
    \item\label{aux-nec} If $M(a)$$=$$\{r_2\}$, $M(b)$$=$$\{r_1,r_3\}$, and $|M(w)\cap U|$$=$$q_w$ for all $w\in W$, then \emph{no} pair~$\{u,w\}$
      with ``$u$$\in$$\{r_1,r_2,r_3\}$ and $w$$\in$$\{a,b\}$'' or with ``$u=r_2$ and $w\in W$'' is blocking~$M$.
    \item\label{aux-card} If $|M(w)\cap U|$$<$$q_w$ for some~$w$$\in$$W$, then $M$ is not stable.
  \end{compactenum}
\end{lemma}

\iflong
\begin{proof}
  Let $M$ be an arbitrary matching.
  Statement~\eqref{aux-nec}: assume that $M(a)=\{r_2\}$, $M(b)=\{r_1,r_3\}$, and $|M(w')\cap U|=q_{w'}$ for all~$w'\in W$.
  Towards a contradiction, suppose that $M$ admits a blocking pair~$\{u, w\}$ with ``$u \in \{r_1,r_2,r_3\}$ and $w\in \{a,b\}$'' or with ``$u=r_2$ and $w\in W$''.
  Clearly, $u \neq r_1$ because $r_1$ already receives her most preferred college.
  Consequently, $w \neq a$ because~$a$ can only accommodate one student and only prefers $r_1$ to its assigned student~$r_2$ but we have already reasoned that $r_1$ is not involved in any blocking pair.
  Similarly, $u\neq r_3$.
  Moreover, $w\neq b$ because~$b$ already receives two students,
  but cannot replace $r_3$ with $r_2$ because of the upper quota for type~$2$.
  Thus, $u=r_2$ and $w\in W$.
  However, since $|M(w)\cap U|=q_w$ and $w$ does not prefer $u$ to any of its assigned student, it follows that $\{u,w\}$ is not a blocking pair, a contradiction.

  Statement~\eqref{aux-card}: Assume that there exists a college~$w\in W$ with $|M(w)\cap U|< q_w$.
  Towards a contradiction, suppose that $M$ is stable.
  We distinguish between two cases, either $r_2\in M(w)$ or $r_2\notin M(w)$.
  If $r_2 \in M(w)$, then since $r_2$ prefers~$b$ to~$w$ it follows that $r_3 \in M(b)$ as otherwise $\{r_2,b\}$ forms a blocking pair.
  Moreover, $r_1\in M(b)$ because $b$ is $r_1$'s most preferred college and $r_2\notin M(b)$.
  However, this means that $a$ receives no student at all and will form a blocking pair with $r_3$, a contradiction.

  If $r_2\notin M(w)$, then $r_2$ must be assigned to $b$ as otherwise $\{r_2,w\}$ is a blocking pair; note that all students in~$U$ have zero types. 
  Then, due to its upper quotas college~$b$ only receives a student, namely $r_2$.
  However, since $a$ can only accommodate one student, it follows that $r_1$ or $r_3$ is assigned to no college.
  This implies that $\{r_1,a\}$ is a blocking pair (if $M(r_1)=\bot$) or $\{r_3,b\}$ is a blocking pair (if $M(r_3)=\bot$), a contradiction. 
\end{proof}
\fi

\ifshort It is straight-forward 
to verify the correctness of Lemma~\ref{lem:aux} by case analysis.
\fi
With Lemma~\ref{lem:aux} in hand, we can prove that the problem is not NP-complete but instead lies on the second level of the polynomial hierarchy.

\begin{theorem}\label{thm:diverse:lq>0:sigma2p-c}
  \diverseties{} is $\Sigma^{\text{P}}_2$-complete, and remains $\Sigma^{\text{P}}_2$-hard even if feasible matchings always exist, there are no ties, \(m = 4\) and \(\umax = 3\).
  \end{theorem}
\ifshort
\begin{proof}[Proof Sketch]
  To show that \diverseties is in~$\Sigma^{\text{P}}_2$,
  we first observe that checking whether a matching is \emph{not} stable can be done by an NP-oracle (guess an unmatched pair~$\{u,w\}$ and a subset~$S$ of students, and check in polynomial time whether $S$ witnesses that $\{u,w\}$ is blocking $M$).
  Hence, we can guess in polynomial time a matching, ask the NP-oracle whether it is not stable, and return yes if and only if the oracle answers no. Containment follows because $\text{NP}^{\text{NP}}=\Sigma^{\text{P}}_2$. 

To show $\Sigma^{\text{P}}_2$-hardness, we reduce from a problem called \notoneinthree: Given a Boolean 3CNF formula~$\phi(X,Y)$ over two equal-size variable sets $X,Y$ such that each clause contains at least two literals from $Y\cup \overline{Y}$, determine whether there exists a truth assignment of $X$ such that for each truth assignment of $Y$ at least one clause $C_j$ is \myemph{not 1-in-3-satisfied}, i.e., does not have precisely one true literal. \notoneinthree\ can be shown to be hard for $\Sigma^{\text{P}}_2$ via a standard (and complementing) reduction from \threesat, a classic $\Pi^{\text{P}}_2$-hard problem~\cite{Stockmeyer1976phcom}.

The idea of our main reduction is to construct, from a given instance $I$ of \notoneinthree\ with $|X|=|Y|=\enn$ and $\emm$ clauses, an equivalent instance $I'$ of \diverse\ with $2\enn$ ``variable-types'', $\emm$ ``clause-types'' and $2$ auxiliary types (the types are ordered in this sequence).
Instance~$I'$ contains a special student~$d$ that has all the variable-types and all the clause-types, and two distinguished colleges $v,w$ which can both accommodate $d$, but $d$ prefers being in $v$. $I'$ furthermore uses Lemma~\ref{lem:aux} to construct a gadget
which ensures that a matching can only be stable if $d$ is matched to $w$---in particular, this will force a stable and feasible matching to ensure $\{d,v\}$ will not form a blocking pair. 

Moreover,~$I'$ contains one clause-student~$d_j$ for each clause~$C_j$ (let $D$ denote the set of all clause-students) and one student~$\lit$ for each literal in~$X\cup Y \cup \overline{X} \cup \overline{Y}$. Student~$d_j$ only has one type: the clause-type~$2\enn+j$ corresponding to~$C_j$. 
Student~$\lit$ has the variable-type~$i\in [\enn]$ corresponding to its variable as well as all the clause-types~$2\enn+j$ of every clause~$C_j$ containing $\lit$.
All $Y$-literal students only want to go to~$v$; all positive $X$-literal students prefer $v$ to $b$ while all negative $X$-literal students prefer $b$ to $v$.

We can now explain the core of the reduction:
the quotas of $v$ are set up in a way which ensures (assuming $d$ is matched to $w$) that precisely one literal-student for each variable in $X$, both literal-students for each variable in $Y$, and some clause-students must be matched to~$v$. In particular, a clause-student~$d_j$ will be matched to~$v$ if and only if the literal-student missing from $v$ represents a literal in $C_j$. Once set up, we show that $\{d,v\}$ is blocking if and only if there is a witness set of literal-students, and this witness set would represent an assignment which 1-in-3 satisfies $I$. In other words, a feasible and stable matching~$M$ exists if and only if there is an assignment of the $X$-variables (which can be reconstructed from $M$) such that no assignment of the $Y$-variables 1-in-3-satisfies all clauses.

The following describes the preference lists, quotas and capacities of the colleges, together with the preference lists of the students from $\{r_1$, $r_2$, $r_3$, $d\}$.
  \noindent{
  \addtolength{\extrarowheight}{\belowrulesep}
  \aboverulesep=0pt
  \belowrulesep=0pt
  \tabcolsep=0pt
  \begin{tabular}{@{}>{\columncolor{white}[0pt][\tabcolsep]}l@{}>{\columncolor{white}[0pt][\tabcolsep]}l@{}>{\columncolor{white}[0pt][\tabcolsep]}l@{}|@{\,}>{\columncolor{white}[0pt][\tabcolsep]}
  l@{}>{\columncolor{white}[0pt][\tabcolsep]}l@{}>{\columncolor{white}[0pt][\tabcolsep]}l@{}>{\columncolor{white}[0pt][\tabcolsep]}l@{}>{\columncolor{white}[0pt][\tabcolsep]}l@{}>{\columncolor{white}[0pt][\tabcolsep]}c@{}}
    \toprule
    S. & Pref. &  & C. & Pref. & LQ. & UQ. & C.\\\midrule 
   \highlight{$r_1$:} & \highlight{$b \osucc a$} && \highlight{$w$:} & \highlight{$d\osucc r_2$} &    \highlight{${0}^{2\enn+\emm+2}$~}  &  \highlight{${1}^{2\enn+\emm+2}$~} & \highlight{$1$}\\
    \highlight{$r_2$:} & \highlight{$b \osucc w \osucc a$} & & \highlight{$a$:} & \highlight{$r_1 \osucc r_2 \osucc r_3$}  &\highlight{${0}^{2\enn+\emm+2}$}& \highlight{$0^{2\enn+\emm}11$}&\highlight{$1$}\\
    \highlight{$r_3$:} & \highlight{$a \osucc b$} & & \highlight{$b$:} & \highlight{$r_3 \osucc r_2 \osucc r_1 \osucc [X]  \osucc [\overline{X}]$}  &\highlight{$1^{\enn}{0}^{\enn+\emm+2}$}& \highlight{$1^{\enn}0^{\emm}1^{\emm+2}$}& \highlight{$\enn+2$}\\
  $d$: & $vw$ & & $v$:& $[D]\osucc d \osucc [Y] \osucc [\overline{Y}] \osucc [\overline{X}] \osucc [X]$ & $1^\enn2^\enn3^\emm00$~~~& $1^\enn2^\enn3^\emm00$~~~& $3\enn$$+$$\emm$\\
    \bottomrule
  \end{tabular}
  }
  This completes the construction which can be verified to fulfill the restriction stated in the theorem.
  It now suffices to show that (1) there exists an $X$-assignment~$\sigma_X$ such that for each $Y$-assignment~$\sigma_Y$ at least one clause is \emph{not} 1-in-3-satisfied if and only if (2) matching $X'\cup Y \cup \overline{Y} \cup D'$ to~$v$, $(\{r_1,r_3\}\cup X \cup \overline{X}) \setminus X'$ to~$b$, $r_2$ to~$a$, and~$d$ to~$w$ is feasible and stable,
  where $X'$ corresponds to the assignment~$\sigma_X$ and $D'=\{d_j \mid |(X'\cup Y\cup \overline{Y})\cap C_j| = 2\}$.
\end{proof}
\fi

\iflong
\begin{proof}
We establish the $\Sigma^{\text{P}}_2$-hardness of the problem by a polynomial-time reduction from the following quantified satisfiability problem, which can be shown to be $\Sigma^{\text{P}}_2$-complete.
\decprobnormal{\notoneinthree}
{Two equal-size sets~$X$ and $Y$ of Boolean variables; a Boolean formula~$\phi(X,Y)$ over~$X\cup Y$ in 3CNF, i.e., a set of clauses each containing 3 literals. Moreover, each clause contains \emph{at least} two literals from $Y\cup \overline{Y}$.}
{Does there exist a truth assignment of $X$ such that for each truth assignment of $Y$ there is a clause~$C_j$ which is not \myemph{1-in-3-satisfied} (i.e., $C_j$ does not have precisely $1$ true literal)?}

\begin{claim}\label{claim:nooneinthree-sigma2p-h}
  \notoneinthree{} is $\Sigma^{\text{P}}_2$-hard, even if each clause contains at least two literals from~$Y\cup \overline{Y}$. 
\end{claim}

\begin{proof}[Proof of Claim~\ref{claim:nooneinthree-sigma2p-h}]
  \renewcommand{\qedsymbol}{$\diamond$}
  We begin by showing that the following \oneinthree problem is $\Pi^{\text{P}}_2$-hard, using the standard polynomial-time reduction from \ethreesat{} to \textsc{Not-1-in-3-3SAT}.
  \decprobnormal{\oneinthree}
{Two equal-size sets~$X$ and $Y$ of Boolean variables; a Boolean formula~$\phi(X,Y)$ over~$X\cup Y$ in 3CNF, i.e., a set of clauses each containing three literals such that at least two of them come from $Y\cup \overline{Y}$.}
  {Is it true that \emph{for each} truth assignment of $X$ there \emph{exists} a truth assignment of $Y$ such that each clause~$C_j$ is \myemph{1-in-3-satisfied} (i.e., $C_j$ is satisfied by exactly one literal)?}
  
  Let $\phi(X^*,Y^*)=(C^*_1,\ldots,C^*_\emm)$ be a Boolean formula in 3CNF over two equal-size variable sets~$X^*$ and $Y^*$.
  For each clause~$C^*_j$, containing three literals~$\lit_j^1$, $\lit_j^2$, $\lit_j^3$ from variables~$X^*\cup Y^*$,
  we introduce four fresh variables~$a_j, b_j, c_j, d_j$, 
  and construct the following three clauses
  $C_j^1, C_j^2, C_j^3$ such that
  \begin{align*}
  &  C_j^1 \coloneqq (\neg \lit_j^1, a_j, b_j), C_j^2\coloneqq (\lit_j^2, b_j, c_j), \text{ and }\\
  &  C_j^3 \coloneqq (\neg \lit_j^3, c_j, d_j)\text{,}
  \end{align*}
  where $\neg \lit$ denotes $\overline{z}$ if $z\in X^*\cup Y^*$; otherwise $\neg \lit$ denotes $z$.
  It is straightforward to verify that for all truth assignment of $X^*$ there exists a truth assignment of $Y^*$ under which $(C^*_j)_{1\le j \le \emm}$ is satisfied if and only if
  for all truth assignment of $X^*$ there exists a truth assignment of $Y^*\cup \{a_j,b_j,c_j,d_j\mid 1\le j \le \emm\}$ under which each clause~$C_j^z$, $1\le j \le \emm$, $1\le z\le 3$, is \emph{1-in-3-satisfied}. Moreover, the addition of $4$ new variables into $Y^*$ can be mirrored by the addition of $4$ new ``dummy'' variables into $X^*$ which do not occur in any clause. Since the former problem is $\Pi^{\text{P}}_2$-hard~\cite{Stockmeyer1976phcom}, we obtain that \oneinthree{} is also $\Pi^{\text{P}}_2$-hard.
  It is straight-forward to see that the newly constructed clauses each have at least two literals coming from $Y \cup \overline{Y}$.

  Since an instance~$I$ of \oneinthree{} is a yes instance if and only if $I$ is a no instance of \notoneinthree, it follows from the well-known complementarity of these classes that \notoneinthree is $\Sigma^{\text{P}}_2$-hard.
\end{proof}

\paragraph{The idea.}
The idea of our main reduction is to construct, from a given instance~$I$ of \notoneinthree\ with $|X|=|Y|=\enn$ and $\emm$ clauses, an equivalent instance $I'$ of \diverse\ with $2\enn$ \myemph{variable-types}, $\emm$ \myemph{clause-types} and $2$ auxiliary types (the types are ordered in this sequence).
Instance~$I'$ contains a special student~$d$ that has all the variable-types and all the clause-types, and two distinguished colleges $v,w$ which can both accommodate $d$, but $d$ prefers being in $v$. Instance~$I'$ furthermore uses Lemma~\ref{lem:aux} to construct a gadget
which ensures that a matching can only be stable if $d$ is matched to $w$---in particular, this will force a stable and feasible matching to ensure $\{d,v\}$ will not form a blocking pair. 

Moreover,~$I'$ contains one clause-student~$d_j$ for each clause~$C_j$ and one student~$\lit$ for each literal in~$X\cup Y \cup \overline{X} \cup \overline{Y}$. Student~$d_j$ only has one type: the clause-type~$2\enn+j$ corresponding to~$C_j$. 
Student~$\lit$ has the variable-type~$i\in [\enn]$ corresponding to its variable as well as all the clause-types~$2\enn+j$ of every clause~$C_j$ containing $\lit$.
All $Y$-literal students only want to go to~$v$; all positive $X$-literal students prefer $v$ to $b$ while all negative $X$-literal students prefer $b$ to $v$.

We can now explain the core of the reduction:
the quotas of $v$ are set up in a way which ensures (assuming $d$ is matched to $w$) that precisely one literal-student for each variable in $X$, both literal-students for each variable in $Y$, and some clause-students must be matched to~$v$. In particular, a clause-student~$d_j$ will be matched to~$v$ if and only if the literal-student missing from $v$ represents a literal in $C_j$. Once set up, we show that $\{d,v\}$ is blocking if and only if there is a witness set of literal-students, and this witness set would represent an assignment which 1-in-3 satisfies $I$. In other words, a feasible and stable matching~$M$ exists if and only if there is an assignment of the $X$-variables (which can be reconstructed from $M$) such that no assignment of the $Y$-variables 1-in-3-satisfies all clauses.

Next, we formally describe our reduction, showing that \diverse is as hard as \notoneinthree even if there are only four colleges and a feasible matching always exists.

Let $I$$=$$(X$$=$$\{x_1,\ldots, x_{\enn}\},Y$$=$$\{y_{\enn+1},\ldots,y_{2\enn}\},\phi(X,Y))$ be an instance of \notoneinthree{} with $\emm$-many clauses which each have at least two literals from $Y\cup \overline{Y}$.
We construct the following instance of \diverse.

\paragraph{The types.}
There are $t\coloneqq 2\enn+\emm+2$ types: for each variable there is a corresponding \myemph{variable type}, for each clause there is a corresponding \myemph{clause type}, and there are two special types.

\paragraph{The students.}
For each literal~$\lit\in X\cup \overline{X} \cup Y \cup \overline{Y}$, there is a literal student~$\lit$.
The type vector~$\typevec_{\lit}$ of student~$\lit$ is constructed as follows:
\begin{itemize}
  \item For each type~$z\in [2\enn]$ if $\lit \in \{x_z, \overline{x}_z, y_z, \overline{y}_z\}$, then let $\typevec_{\lit}[z]\coloneqq 1$; otherwise let $\typevec_{\lit}[z]\coloneqq 0$.
  \item For each~$j \in [\emm]$ let $\typevec_{\lit}[2n'+j]\coloneqq 1$ if $C_j$ contains~$\lit$; otherwise let $\typevec_{\lit}[2\enn+j]\coloneqq 0$.
  \item Student~$\lit$ has neither type~$2\enn+\emm+1$ nor type~$2\enn+\emm+2$, i.e., $\typevec_{\lit}[2\enn+\emm+1]\coloneqq \typevec_{\lit}[2\enn+\emm+2]\coloneqq 0$.
\end{itemize}

We introduce four special students~\myemph{$d$, $r_1$, $r_2$, and~$r_3$}.
They have the following types:
\begin{align*}
  \typevec_d \coloneqq & \{1\}^{2\enn+\emm}00\text{,} &
  \typevec_{r_1} \coloneqq & \{0\}^{2\enn+\emm}10\text{,}\\
  \typevec_{r_2} \coloneqq & \{0\}^{2\enn+\emm}11\text{,} &  
  \typevec_{r_3} \coloneqq & \{0\}^{2\enn+\emm}01\text{.} 
\end{align*}
We also introduce $\emm$ \myemph{clause-students~$D=\{d_1, \ldots, d_{\emm}\}$} such that for each~$j\in [\emm]$ student~$d_{j}$ has exactly one type, namely:
\begin{align*}
  \forall z \in T\setminus \{2\enn+j\}\colon  \typevec_{d_{j}}[z]\coloneqq 0, \text{~ and ~}
  \typevec_{d_{j}}[2\enn+j]\coloneqq 1, 
\end{align*}

In other words, the type of~$d_j$ corresponds to the clause~$C_j$.

\paragraph{The colleges.}

There are four colleges~$v$, $w$, $a$, and $b$.
Their capacities and type-specific lower and upper quotas are defined as follows: 
\begin{itemize}
  \item Let $q_v\coloneqq 3\enn+\emm$, $q_w\coloneqq 1$, $q_a\coloneqq 1$, and $q_b \coloneqq \enn+2$.
  \item Let $\lowervec_v \coloneqq \uppervec_v \coloneqq \{1\}^{\enn}\{2\}^{\enn}\{3\}^{\emm}00$.
  \item Let $\lowervec_w \coloneqq \{0\}^{2\enn+\emm+2}$ and $\uppervec_w \coloneqq \{1\}^{2\enn+\emm+2}$.
  \item Let $\lowervec_a \coloneqq \{0\}^{2\enn+\emm+2}$ and $\uppervec_a \coloneqq \{0\}^{2\enn+\emm}11$.
  \item Let $\lowervec_b \coloneqq \{1\}^{\enn}\{0\}^{\enn+\emm+2}$ and $\uppervec_b \coloneqq \{1\}^{\enn}\{0\}^{\enn}\{1\}^{\emm+2}$.
\end{itemize}

\paragraph{The preference lists of the students.}
For each variable~$x_i\in X$ the preference lists of the literal students~$x_i$ and $\overline{x}_i$ are as follows:
\begin{align*}
  x_i \colon v \succ_{x_i} b~\text{ and }~ \overline{x}_i \colon b \succ_{\overline{x}_i} v.
\end{align*}
Each clause student \(d_j \in D\), and each literal student~$\lit\in Y\cup \overline{Y}$ contains only $v$ in its preference list.
The preference list of the special students are:
\begin{align*}
  d\colon & v\succ_d w\text{,} &
  r_1\colon& b\succ_{r_1} a\text{,}\\
  r_2\colon &b\succ_{r_2} w \succ_{r_2} a\text{,} & 
  r_3\colon& a\succ_{r_3} b\text{.}
\end{align*}

\paragraph{The preference lists of the colleges.}
The preference lists of the four colleges are as follows:
\begin{align*}
  v\colon & [D]\succ_v d \succ_v [Y] \succ_v [\overline{Y}] \succ_v [\overline{X}] \succ_v [X],\\
  w\colon & d \succ_w  r_2\text{,}\\
  a\colon & r_1 \succ_a  r_2 \succ_a r_3\text{,}\\
  b\colon & r_3 \succ_b  r_2 \succ_b r_1 \succ_b [X] \succ_b [\overline{X}]\text{.}
\end{align*}


\begin{table}
  \noindent{
  \addtolength{\extrarowheight}{\belowrulesep}
  \aboverulesep=0pt
  \belowrulesep=0pt
  \tabcolsep=0pt
  \begin{tabular}{@{}>{\columncolor{white}[0pt][\tabcolsep]}l@{}>{\columncolor{white}[0pt][\tabcolsep]}l@{}>{\columncolor{white}[0pt][\tabcolsep]}l@{}|@{}>{\columncolor{white}[0pt][\tabcolsep]}
  l@{}>{\columncolor{white}[0pt][\tabcolsep]}l@{}>{\columncolor{white}[0pt][\tabcolsep]}l@{}>{\columncolor{white}[0pt][\tabcolsep]}l@{}>{\columncolor{white}[0pt][\tabcolsep]}l@{}>{\columncolor{white}[0pt][\tabcolsep]}c@{}}
    \toprule
    S.~~ & Pref. &  & C.~ & Pref. & LQ. & UQ. & C.\\\midrule 
   \highlight{$r_1$:} & \highlight{$b \osucc a$} && \highlight{$w$:} & \highlight{$d\osucc r_2$} &    \highlight{${0}^{2\enn+\emm+2}$~}  &  \highlight{${1}^{2\enn+\emm+2}$~} & \highlight{$1$}\\
    \highlight{$r_2$:} & \highlight{$b \osucc w \osucc a$} & & \highlight{$a$:} & \highlight{$r_1 \osucc r_2 \osucc r_3$}  &\highlight{${0}^{2\enn+\emm+2}$}& \highlight{$0^{2\enn+\emm}11$}&\highlight{$1$}\\
    \highlight{$r_3$:} & \highlight{$a \osucc b$} & & \highlight{$b$:} & \highlight{$r_3 \osucc r_2 \osucc r_1 \osucc [X]  \osucc [\overline{X}]$}  &\highlight{$1^{\enn}{0}^{\enn+\emm+2}$}& \highlight{$1^{\enn}0^{\emm}1^{\emm+2}$}& \highlight{$\boldsymbol{\enn}+2$}\\
  $d$: & $v\osucc w$ & & $v$:& $[D]\osucc d \osucc [Y] \osucc [\overline{Y}] \osucc [\overline{X}] \osucc [X]$ & $1^\enn2^\enn3^\emm00$~~~& $1^\enn2^\enn3^\emm00$~~~& $3\enn$$+$$\emm$\\
  $x_i$: & $v \osucc b$ & &\\
  $\overline{x}_i$: &   $b\osucc v$ &\\
    \bottomrule
  \end{tabular}
}
\caption{A description of the preference lists, quotas and capacities of the colleges, together with the preference lists of the students from~$\{r_1$, $r_2$, $r_3$, $d\}\cup X\cup \overline{X}$, used in the proof of \cref{thm:diverse:lq>0:sigma2p-c}.}\label{fig:sigma2p-h-reduction}
\end{table}
This completes the construction, which can clearly be done in polynomial time~(also see \cref{fig:sigma2p-h-reduction} for an illustration).
We denote the constructed instance by~$I'$.
We first claim that $I'$ admits a feasible matching. 

\begin{claim}\label[claim]{claim:feasible}
  The constructed instance~$I'$ admits a feasible matching.
\end{claim}

\begin{proof}[Proof of Claim~\ref{claim:feasible}]
  \renewcommand{\qedsymbol}{$\diamond$}
  We prove a slightly stronger claim. Notably, consider an arbitrary matching $M$ constructed as follows.
  Let $X'\subseteq X\cup \overline{X}$ with $|X'|=n$ such that for each~$i\in [n]$ it holds that $|\{x_i,\overline{x}_i\}\cap X'|=1$.
  In other words, $X'$ defines a truth assignment for~$X$.
  Let $M(v)\coloneqq X'\cup Y \cup \overline{Y}\cup D'$, where $D'\coloneqq \{d_j \mid |C_j \cap (X'\cup Y\cup \overline{Y})| = 2\}$.
  Let $M(w)\coloneqq \{d\}$, $M(a)\coloneqq \{r_2\}$, and $M(b)\!\coloneqq\!(\{r_1,r_3\}\cup X\cup \overline{X})\!\setminus X'$.
  
  It is straightforward to verify that matching~$M$ fulfills all capacity and type (diversity) constrains. Indeed, the type constraints for the first $2\enn$ types are met trivially, while the type constraints corresponding to clauses are met by the fact that each clause contains either 3 or 2 literals that were matched to $v$ (since for all \(j \in [m]\), \(|C_j \cap (X'\cup Y\cup \overline{Y})| \geq 2\)), and in the latter case the type constraint is met by the addition of the respective clause student. 
\end{proof}


Before we continue with the correctness proof, we observe the following properties that each feasible and stable matching must fulfill.

\begin{claim}\label{claim:feasible-matching}
  Each feasible and stable matching~$M$ of $I'$ must satisfy the following.
  \begin{enumerate}[(1)]
    \item\label{feasible:X}
    For each~$x_i\in X$ 
    either (i) $M(x_i) = v$ and $M(\overline{x}_i) = b$ or (ii) $M(\overline{x}_i) = v$ and $M(x_i) = b$.
   \item\label{feasible:d-b} $M(d)=w$. 
    \item\label{feasible:Y-a} $Y\cup \overline{Y} \subseteq M(v)$.
  \item\label{feasible:dummy} For each~$C_j$ it holds that $d_{j}\in M(v)$ if and only if $|C_j\cap M(v)| = 2$. 
\end{enumerate}
\end{claim}

\begin{proof}[Proof of
  \cref{claim:feasible-matching}]
  \renewcommand{\qedsymbol}{$\diamond$}
  We first show Statement~\eqref{feasible:d-b} as the other statements are based on this.
  For Statement~\eqref{feasible:d-b}, observe that regarding types~$2\enn+\emm+1$ and $2\enn+\emm+2$, the sets~$\{d\}\uplus \{r_1,r_2,r_3\}$ and $\{w\}\uplus \{a,b\}$ correspond exactly to students~$U\uplus \{r_1,r_2,r_3\}$ and colleges~$W\uplus \{a,b\}$ given in \cref{lem:aux}.
  By \cref{lem:aux}\eqref{aux-card} it follows that  $|M(w)\cap \{d\}|=1$, i.e., $M(d)=w$.

  Statement~\eqref{feasible:X} holds because $d\notin M(v)$ (see Statement~\eqref{feasible:d-b}) and due to the quotas of $v$ and $b$ on the $X$-types.  

  Statement~\eqref{feasible:Y-a} holds because $d\notin M(v)$ (see Statement~\eqref{feasible:d-b}) and due to the $Y$-quotas of $v$.

  It remains to show Statement~\eqref{feasible:dummy}.
  For the ``only if'' part, assume that $d_j\in M(v)$.
  Since $\typevec_{d_j}[2\enn+j] = 1$ and $\lowervec_{v}[2\enn+j]=\uppervec_{v}[2\enn+j]=3$ it follows $M(v)$ contains exactly two literal students~$\lit$ and $\lit'$ with $M(v)\cap C_j = \{\lit, \lit'\}$.
  For the ``if'' part, assume that $|M(v)\cap C_j| = 2$. Similarly, since $d\notin M(v)$ and since $\lowervec_{v}[2\enn+j]=\uppervec_{v}[2\enn+j]=3$ and since $d_j$ is the only student (other than those from $C_j$) which has type~$2\enn+j$,
  it follows that $d_j\in M(v)$.
\end{proof}

Now, we proceed to establish the correctness of the reduction, i.e., $I$ is a yes instance if and only if the constructed \diverse{} instance~$I'$ admits a feasible and stable matching.

For the ``only if'' part, assume that $I$ is a yes instance.
Let $\sigma_X$ be a truth assignment for~$X$ which ``witnesses'' that $I$ is a yes instance.
Let $X'\coloneqq \{x_i\mid \sigma_{X}(x_{i})=\true\} \cup \{\overline{x}_{i}\mid\sigma_{X}(x_i)=\false\}$.
Let $M$ be a matching constructed as follows:
\begin{enumerate}
  \item $M(v)\coloneqq X'\cup Y\cup \overline{Y}\cup D'$, where $D'\coloneqq \{d_j \mid |C_j \cap (X'\cup Y\cup \overline{Y})| = 2\}$,
  \item $M(w)\coloneqq \{d\}$,
  \item $M(a)\coloneqq \{r_2\}$,
  \item $M(b)\coloneqq \{r_1,r_3\}\cup (X\cup \overline{X}\setminus X')$.
\end{enumerate}
We claim that $M$ is feasible and stable.
It is straightforward to verify that $M$ is feasible; note that by \cref{claim:nooneinthree-sigma2p-h} we know that if $d_j \notin M(v)$ then $|C_j \cap (X'\cup Y\cup \overline{Y})|=3$ for all $j\in [\emm]$.
Suppose, for the sake of contradiction, that $M$ is not stable.
That is, it is blocked by some unmatched pair~$\{\alpha, \beta\}$ with $\beta\in \{a,b,v,w\}$.
Now, observe that regarding the special types~$2\enn+\emm+1$ and $2\enn+\emm+2$
students~$\{d\}\cup \{r_1,r_2,r_3\}$ and colleges~$W\cup \{a,b\}$ correspond exactly to the students~$U\cup \{r_1,r_2,r_3\}$ and colleges~$W\cup \{a.b\}$ given in \cref{lem:aux}.
By \cref{lem:aux}\eqref{aux-nec}, we immediately have that no unmatched pair~$\{u',w'\}$ with $u'\in \{r_1,r_2,r_3\}$ and $w'\in \{a,b\}$ is blocking~$M$.
Further, observe that college~$w$ is \emph{not} involved in a blocking pair as $w$ has only capacity one and already receives its most preferred student.
Similarly, neither is any student in $Y\cup \overline{Y}$ involved in a blocking pair.
This means that $\alpha\in X \cup \overline{X} \cup \{d\} \cup D$. 

We now distinguish between four cases; in each case, let $S$ be the witness for the considered blocking pair.
\begin{enumerate}[(1)]
  \item If $\alpha = x_i$ for some $i\in [\enn]$, then $\beta=v$ and the set~$\{\overline{x}_i\}$ is the only subset of students which can be replaced with $x_i$ such that the new matching remains feasible for $v$ regarding the $X$-types.
However, college~$v$ prefers $\overline{x}_i$ to $x_i$, meaning that $\{x_i,v\}$ is not a blocking pair.
\item Similarly, if $\alpha=\overline{x}_i$ for some $i \in [\enn]$, then $\beta=b$ and the set~$\{{x}_i\}$ is the only subset of students which can be replaced with $\alpha$ such that the new matching remains feasible for $\beta$ regarding the $X$-types.
Again, $b$ prefers~${x}_i$ to~$\overline{x}_i$, meaning that $\{\overline{x}_i, b\}$ is not a blocking pair.
\item If $\alpha=d_j$ for some $j\in [\emm]$, then $\beta=v$. However, no subset of students assigned to~$v$ can be replaced so that the resulting matching remains feasible for $v$ regarding the $X$-types and the $Y$-types because $d_j$ has only type~$2\enn+j$.
\item Finally, if $\alpha=d$, then $\beta=v$; this represents the ``core'' of the reduction. In this case, there must be a subset of literal students~$Y'\subseteq Y\cup \overline{Y} \subseteq M(v)$ with $|Y'|=\enn$ such that
for each $i\in \{\enn+1,\ldots,2\enn\}$ it holds that $|\{y_i,\overline{y}_i\}\cap Y'| = 1$ and $S=X'\cup Y'$; recall that $X'= M(a)\cap (X \cup \overline{X})$ and note that $D\cap S =\emptyset$ as college~$v$ prefers each clause student to~$d$.
Now, define a truth assignment for~$Y$ by letting~$\sigma_{Y}(y_i)\coloneqq \true$ if and only if $y_i \in Y'$.
Since~$I$ is a yes instance, it follows that there exists a clause~$C_j=(\lit_j^1, \lit_j^2, \lit_j^3)$ which does \emph{not} have exactly one true literal under $\sigma_x$ and $\sigma_y$.
If $C_j$ has no true literal under~$\sigma_X$ and~$\sigma_Y$, then $|\{u\in M(v)\setminus S \mid \typevec_{u}[2\enn+j]=1\}|=3$, a contradiction to the assumption that replacing $S$ with $d$ is feasible for~$v$ because $\typevec_d[2\enn+j]=1$.
If $C_j$ has \emph{more than one} true literal under~$\sigma_X$ and~$\sigma_Y$, then $|\{u\in M(a)\setminus S \mid \typevec_{u}[2\enn+j]=1\}|\le 1$, which once again represents a contradiction to the assumption that replacing $S$ with $d$ is feasible for $v$.
\end{enumerate}
In any case, we arrive at a contradiction.
Thus, $I$ is a yes instance.

For the ``if'' part, assume that $I'$ is a yes instance, i.e., there exists a feasible and stable matching, called~$M$.
Define $X'\coloneqq M(v)\cap (X\cup \overline{X})$ as the subset of literal students from $X\cup\overline{X}$ which are assigned to $v$.
By \cref{claim:feasible-matching}\eqref{feasible:X}, $X'$ defines a truth assignment of $X$.
Let $\sigma_{X}$ be the truth assignment of $X$ corresponding to~$X'$, i.e,
$\sigma_{X}(x_i)\coloneqq \true$ if and only if $x_i\in X'$.
Consider an arbitrary truth assignment~$\sigma_{Y}$ of $Y$ and let $Y'$ be the subset of literal students that correspond to~$\sigma_Y$, i.e,
$Y'\coloneqq \{y_i \in Y \mid \sigma_Y(y_i)=\true\}\cup \{\overline{y}_i \mid \sigma_Y(y_i)=\false\}$.
Define $A \coloneqq X'\cup Y'$; note that $A\subseteq M(v)$ due to \cref{claim:feasible-matching}\eqref{feasible:Y-a}.
By \cref{claim:feasible-matching}\eqref{feasible:d-b}, $M(d)=w$.
Hence,~$d$ prefers $v$ to~$M(d)$.
Moreover,~$v$ prefers~$d$ to every student from $A$.
Since $M$ is stable, it must hold that $M\setminus \big(\{\{\lit, v\} \mid \lit \in A\} \cup \{d,w\}\big) \cup \{\{d,v\}\}$ is \emph{not} feasible for $v$.
Observe that by \cref{claim:feasible-matching}(\ref{feasible:X},\ref{feasible:Y-a},\ref{feasible:dummy}),
the students~$M(v)$ assigned to~$v$ reach the upper quota of each type for~$v$.
By the type quotas of $v$ and by the type vector of $d$, we have to show that
$\sum_{\lit \in A}\typevec_{\lit} \neq \{1\}^{2\enn+\emm}00$.
Clearly, for each~$z\in [2\enn]$, it holds that $|\{\lit \in A\mid \typevec_{\lit}[z]=1\}|=1$.
Thus, to ensure that $\{d,v\}$ is \emph{not} blocking $M$, there must exist some type~$2n+j$ such that $\sum_{\lit \in A}\typevec_{\lit} \in \{0,2,3\}$ because $\typevec_d[2\enn+j]=1$.
This implies that $C_j$ obtains \emph{no}, \emph{exactly two}, or \emph{exactly three} literals under~$\sigma_X\cup \sigma_Y$.
In other words, $C_j$ is not 1-in-3-satisfied under $\sigma_X\cup \sigma_Y$, implying that $I$ is a yes instance.


It remains to show that \diverseties is contained in the complexity class~$\Sigma^{\text{P}}_2$.
Intuitively, observe that the problem can be stated as ``deciding whether there exists a feasible matching for $I$ such that every unmatched student-college pair~$\{u,w\}$ is \emph{not} blocking~$M$'', i.e., via an existential quantification followed by a universal quantification---which is a well-known characterization of problems in $\Sigma^{\text{P}}_2$~\cite[Theorem~17.8]{Pap94}. 

To formalize, observe that given an instance~$I$ of \diverseties{} and a matching~$M$ of~$I$,
checking~$M$ is \emph{not} stable can be done by a non-deterministic polynomial-time oracle machine: we can simply guess an unmatched student-college pair~$\{u,w\}$ and a subset \(U'\) of students assigned to college~$w$ and check in polynomial time whether \(U'\) witnesses that $\{u,w\}$ is a blocking pair for $M$.
This implies that given an instance of \diverseties{} and a matching for $I$ checking whether $M$ is \emph{not} feasible or \emph{not} stable for~$I$ is in NP; note that checking whether $M$ is not feasible can be even done in polynomial time.
Let $O$ be such an NP oracle.
Then, \diverseties\ belongs to $\text{NP}^{\text{NP}}$ because there exists an non-deterministic polynomial-time Turing machine that on an input~$I$ of \diverseties{} guesses a matching~$M$ and asks oracle~$O$ whether $M$ is not feasible or not stable for input~$I$.
It returns yes if and only if the oracle says no.
By definition, $\text{NP}^{\text{NP}}=\Sigma^{\text{P}}_2$~\cite[Definition 17.2]{Pap94}, implying that \diverseties lies in $\Sigma^{\text{P}}$.
\end{proof}
\fi

To see why the reduction behind Theorem~\ref{thm:diverse:lq>0:sigma2p-c} can be used to directly show $\Sigma^{\text{P}}_2$-hardness for the problem studied by \citeauthor{AzizGaspersSunWalsh2019aamas}~\shortcite{AzizGaspersSunWalsh2019aamas} we observe that in the constructed instance, $\{d,v\}$ is a blocking pair if and only if it is a d-blocking pair.
\iflong
Recall that an unmatched pair~$\{u,w\}$ is a \myemph{d-blocking pair} of a matching~$M$ if there exists a subset~$S\subseteq M(w)$ of assigned students which is a witness that $\{u,w\}$ is blocking~$M$ and $M \cup \{u,w\} \setminus \big(\{u,M(u)\}\cup \{\{u',w\}\mid u'\in S\}\big)$ is feasible for all colleges.
\fi
\iflong
\begin{corollary}\label{co:aziz-sigma2p-c}
  The \textsc{School Choice with Diversity Constraints} problem studied by \citeauthor{AzizGaspersSunWalsh2019aamas}~\shortcite{AzizGaspersSunWalsh2019aamas} is $\Sigma^{\text{P}}_2$-complete.
\end{corollary}

\begin{proof}
  The \textsc{School Choice with Diversity Constraints} problem has as input a \diverse{} instance and asks whether there exists a feasible matching which does not admit any d-blocking pairs.
  To show that this problem is also $\Sigma^{\text{P}}_2$-hard, we use the same construction as given in \cref{thm:diverse:lq>0:sigma2p-c}.
  Let $I'$ be the instance constructed in the $\Sigma^{\text{P}}_2$-hardness proof of \cref{thm:diverse:lq>0:sigma2p-c}.
  To show the correctness,
  it suffices to show that a feasible matching~$M$ is stable for $I'$ if and only if it does not admit any d-blocking pair.
  
  For the ``only if'' part, assume that $M$ is a feasible and stable matching for $I'$. 
  This means that no unmatched pair is blocking~$M$.
  Since by definition each d-blocking pair is a blocking pair, we infer that no unmatched pair is d-blocking~$M$.

  For the ``if'' part, assume that $M$ is a feasible matching which does not admit any d-blocking pair.
  First of all, we observe and show that properties as given in \cref{claim:feasible-matching} can be also shown for $M$.
  We first claim that $d\in M(w)$.
  Suppose, for the sake of contradiction, that $d\notin M(w)$.
  Then, by the preferences of $w$ it must hold that $M(r_2)\succeq_{r_2} w$ as otherwise $\{r_2,w\}$ forms a d-blocking pair; note that both $a$ and $b$ have zero lower quotas regarding types~$2\enn+\emm+1$ and $2\enn+\emm+2$.
  If $M(r_2)=w$, then $r_3 \in M(b)$ as otherwise $\{r_2,b\}$ is a d-blocking pair.
  Then, it must hold that $r_1\in M(b)$ as otherwise $\{r_1,b\}$ is a d-blocking pair.
  However, $a$ will receive no student and will form with $r_3$ a d-blocking pair, a contradiction.

  If $M(r_2)=b$, then neither~$r_3$ nor~$r_1$ can be assigned to $b$ anymore because of the upper quotas regarding types~$2\enn+\emm+1$ and $2\enn+\emm+2$.
  Since $a$ has only capacity one, it follows that $r_1$ or $r_3$ will remain unmatched.
  If $M(r_1)=\bot$, then it will form with $a$ a d-blocking pair.
  If $M(r_3)=\bot$, then it will form with $b$ a d-blocking pair.
  In other words, we infer a contradiction for the case when $M(r_2)=b$.
  Together, we have shown that $d\in M(w)$.
  Consequently, we can infer that $M(a)=\{r_2\}$ and $\{r_1,r_3\}\subseteq M(b)$.
 
  Then, by the quotas of $Y$-types and $X$-types for colleges~$v$ and $w$, it follows that
   \begin{compactitem}
    \item For each~$x_i\in X$ 
    either (i) $M(x_i) = v$ and $M(\overline{x}_i) = b$ or (ii) $M(\overline{x}_i) = v$ and $M(x_i) = b$.
    \item $Y\cup \overline{Y} \subseteq M(v)$.
  \end{compactitem}
  Moreover, for each~$C_j$ it holds that $d_{j}\in M(v)$ if and only if $|C_j\cap M(v)| = 2$.

  Summarizing, there exists a subset~$X'\subseteq X\cup \overline{X}$ with $|X|=\enn$,
  where for each $i\in [\enn]$ it holds that $|\{x_i,\overline{x}_i\}\cap X'|=1$ such that 
  \begin{compactenum}[(1)]
  \item $M(v) = X'\cup Y\cup \overline{Y}\cup D'$, where $D'\coloneqq \{d_j \mid |C_j \cap (X'\cup Y\cup \overline{Y})| = 2\}$.
  \item $M(w)\coloneqq \{d\}$.
  \item $M(a)\coloneqq \{r_2\}$.
  \item $M(b)\coloneqq \{r_1,r_3\}\cup (X\cup \overline{X}\setminus X')$.
\end{compactenum}

  Now, we show that $M$ does not admit any blocking pair.
  Suppose, for the sake of contradiction, that $M$ admits a blocking pair~$\{\alpha, \beta\}$ with $\beta \in \{v,w,a,b\}$.
  Clearly, $\beta \neq w$ as it already obtains its most preferred student.
  Further, $\alpha \notin X\cup \overline{X}$ because of the following reasons:
  If $\alpha=x_i$, then it must hold that $M(x_i)=b$, $M(\overline{x}_i)=v$, and $\alpha = v$.
  However, $v$ prefers $\overline{x}_i$ to $x_i$.
  Similarly, if $\alpha=\overline{x}_i$, then it must hold that $M(\overline{x}_i)=v$, $M({x}_i)=b$, and $\alpha = b$.
  However, $b$ prefers $x_i$ to $\overline{x}_i$.
  Next, we distinguish between three cases for the choices of $\beta$.
  \begin{itemize}
    \item If $\beta = a$, then $\alpha = r_1$.
    This would imply that $\{r_1,a\}$ is also a d-blocking pair because $M(r_1)=b$ has zero lower quotas regarding the types~$2\enn+\emm+1$ and $2\enn+\emm+2$, a contradiction.
    \item If $\beta = b$, then $\alpha \in \{x_i,\overline{x}_i\}$ for some~$i\in [\enn]$.
    However, we have just reasoned that $\alpha\notin X\cup \overline{X}$.
    \item If $\beta = v$, then $\alpha \in \{d_j, d\}$ for some~$j\in [\emm]$ and some~$i\in [\enn]$.

    On the one hand, if $\{d_j, v\}$ would form a blocking pair, then $M(d_j)=\bot$ and thus, $\{d_j,v\}$ would also form a d-blocking pair, a contradiction.
    On the other hand, if $\{d, v\}$ would form a blocking pair, then since $M(d)=w$ has zero lower quotas for all types, it follows that $\{d,v\}$ would also form a d-blocking pair, a contradiction.
  \end{itemize}
  All together, we achieve our claim that $M$ is stable.
  This completes the proof for showing hardness.

  Finally, let us turn the $\Sigma^{\text{P}}_2$-containment.
  Similarly to the containment proof for \cref{thm:diverse:lq>0:sigma2p-c},
  it suffices to show that given a matching~$M$ we can check whether $M$ is not feasible or admits a d-blocking pair using a non-deterministic polynomial-time oracle machine:
  Check whether $M$ is not feasible in polynomial time,
  and guess an unmatched student-college pair~$\{u,w\}$ and a subset \(U'\) of students assigned to college~$w$ and check in polynomial time whether \(U'\) ``witnesses'' that $\{u,w\}$ is a d-blocking pair for $M$.
\end{proof}
\fi
The proof of Theorem~\ref{thm:diverse:lq>0:sigma2p-c} can also be adapted to correct an erroneous theorem pertaining to a related problem called \CSM (\textsc{CSM})~\cite[Theorem~3.1]{Huang2010classifiedSM}. In particular, that theorem claims that \textsc{CSM} is NP-complete, but it is in effect also $\Sigma^{\text{P}}_2$-hard.
The idea for the adaption is to construct dummy variable students with zero-types and introduce additional types to ensure that $\{v,d\}$ is a blocking pair in the proof of \cref{thm:diverse:lq>0:sigma2p-c} if and only if $v$ forms with $d$ and the dummy variable students a blocking coalition.

\iflong
\begin{proposition}\label{prop:CSM-sigma2p-c}
  \CSM is $\Sigma^{\text{NP}}_2$-complete; the hardness holds even for only four colleges.
\end{proposition}

\begin{proof}
  Before we show the statement,
  we repeat the definition of \CSM~\cite{Huang2010classifiedSM} (also see \cite[Chapter~5.2.5]{Manlove2013}) for the sake of completeness.
  \CSM instances are the same as \diverse instance; in particular we assume preference lists without ties.
  The task is to decide whether there exists a matching which is feasible and
  \myemph{c-stable}; we use a different name to distinguish it from our stability definition:
  A matching for an input~$I=(U,W,T,(\succ_u,\typevec_u)_{u\in U}, (\succ_w,\lowervec_w,\uppervec_w,q_w)_{w\in W})$ is \myemph{c-stable} if it does not admit a \myemph{blocking coalition}, where a \myemph{blocking coalition} of $M$ comprises a college~$w_j\in W$ and a set~$U'=\{u'_{i_1}, \ldots, u'_{i_k}\}\subseteq \acset(w_j)$ of $k\ge |M(w_j)|$ students with $u'_{i_s} \succ_{w_j} u'_{i_{s+1}}$, $s\in [k-1]$,
  such that
  \begin{enumerate}[(1)]
    \item\label{bc:feasible} $\{\{u,w_j\} \mid u \in U'\}$ is feasible for $w_j$;
    \item\label{bc:students} each student in $u\in U' \setminus M(w_j)$ prefers $w_j$ to $M(u)$;
    \item\label{bc:college} either college~$w_j$ prefers~$u'_{i_s}$ to $u_{j_s}$ or $u'_{i_s}=u_{j_s}$, $s \in [|M(w_j)|]$,
    where $(u_{j_1}, \ldots, u_{j_{|M(w_j)|}})$ denotes the sequence of students assigned to~$w_j$ by $M$ in decreasing order of preferences of $w_j$;
    \item\label{bc:size} $k > |M(w_j)|$ or there exists some~$s$ such that $w_j$ strictly prefers $u'_{i_s}$ to $u_{i_s}$.
  \end{enumerate}
  By the above definition, a blocking coalition involving a college~$w$ must involve at least the same number of students as are assigned to~$w$~(Property~\eqref{bc:size}),
  such that~$w$ weakly prefers the coalition to its assigned students (Property~\eqref{bc:college}), and if the numbers are the same, then~$w$ must strictly prefer the coalition to its assigned students.
  Due to this, a blocking pair corresponds to a blocking coalition if and only if the corresponding witness consists of at most one student.
  However, in our reduction for \cref{thm:diverse:lq>0:sigma2p-c}, the witness for a blocking pair is very large.
  To adapt our construction for this different c-stability notion, we need to introduce dummy students (with none of the constructed types) that together with the blocking pair form a blocking coalition.
  To make sure that no unintended blocking coalitions pop up, we also need to duplicate variable students, clause students, and the special student,
  and introduce additional types to control the size of the blocking coalitions.
  
  Formally, let $I$$=$$(X$$=$$\{x_1,\ldots, x_{\enn}\},Y$$=$$\{y_{\enn+1},\ldots,y_{2\enn}\}$, $\phi(X,Y))$ be an instance of \notoneinthree{} with $\emm$-many clauses, which each have at least two literals from $Y\cup \overline{Y}$.
  Let $X\cup \overline{X}\cup Y \cup \overline{Y} \cup D\cup \{d\} \cup \{r_1,r_2,r_3\}$ be the students,
  and let $\{v,w,a,b\}$ be the colleges constructed in the reduction for \cref{thm:diverse:lq>0:sigma2p-c}; recall that $D=\{d_j \mid j\in [\emm]\}$.
  We introduce $2\enn+2\emm+3$ types, where the first $2\enn+\emm+2$~types are the same as the ones introduced in the reduction for \cref{thm:diverse:lq>0:sigma2p-c}.
  Now, for each $i\in [2\enn]$, we also introduce a \myemph{dummy variable student}, called~$f_{i}$,
  and for each~$j\in [\emm]$, we introduce a \myemph{c-clause student}, called~$e_j$.
  Let $F$ and $E$ denote the set of all dummy variable students and c-clause students, respectively.
  Finally, we introduce a copy of the special student~$d$, called~$d'$.
  The role of $d'$ is to be replaced with $d$ in a blocking coalition, while the role of $F$ is to form a blocking coalition with $\{d,v\}$.
  The set of colleges remains unchanged.

  \paragraph{The types of the students.}
  Each literal student~$\lit\in X\cup \overline{X} \cup Y\cup \overline{Y}$ has the same types as that in the preceding reduction.
  The same holds for the three special student~$r_1,r_2,r_3$.
  More precisely, the type vector~$\typevec_{\lit}$ of student~$\lit$ is constructed as follows.
  \begin{itemize}
    \item For each type~$z\in [2\enn]$ if $\lit \in \{x_z, \overline{x}_z, y_z, \overline{y}_z\}$, then let $\typevec_{\lit}[z]\coloneqq 1$; otherwise let $\typevec_{\lit}[z]\coloneqq 0$.
    \item For each~$j \in [\emm]$ let $\typevec_{\lit}[2\enn+j]\coloneqq 1$ if $C_j$ contains~$\lit$; otherwise let $\typevec_{\lit}[2\enn+j]\coloneqq 0$.
    \item Student~$\lit$ has no types in $\{2\enn+\emm+1,  \ldots, 2\enn+2\emm+3\}$.
  \end{itemize}  
  The type vector of the special students are:
  \begin{align*}
    \typevec_{r_1} \coloneqq & \{0\}^{2\enn+\emm}10\{0\}^{\emm+1}\text{,}&
    \typevec_{r_2} \coloneqq & \{0\}^{2\enn+\emm}11\{0\}^{\emm+1}\text{,} \\  
    \typevec_{r_3} \coloneqq & \{0\}^{2\enn+\emm}01\{0\}^{\emm+1}\text{.} 
  \end{align*}
  No dummy variable student~$f_i$ ($i \in [\enn]$) possesses any type.
  Each clause student~$d_j\in D$ has two types~$2\enn+j$ and $2\enn+\emm+2+j$.
  Each c-clause student~$e_j\in D$ has only type~$2\enn+\emm+2+j$.
  The types of the special students and her copy are:
   \begin{align*}
    \typevec_{d} \coloneqq & \{1\}^{2\enn+\emm}00\{0\}^{\emm}{\color{winered}1}\text{,} &
    \typevec_{d'} \coloneqq & \{0\}^{2\enn+2\emm+2}{\color{winered}1}\text{.}
  \end{align*}

\paragraph{The preference lists of the students.}
The preferences of the students from $X\cup \overline{X} \cup Y\cup \overline{Y} \cup \{r_1,r_2,r_3\}$ remain the same as in the other reduction.
Formally, for each variable~$x_i\in X$ the preference lists of the literal students~$x_i$,
$\overline{x}_i$ and the dummy variable student~$f_i$ are as follows:
\begin{align*}
  x_i \colon v \succ_{x_i} b~\text{, }~ \overline{x}_i \colon b \succ_{\overline{x}_i} v,
  ~\text{ and }~ f_i\colon v \succ_{f_i} w.
\end{align*}
Each literal student~$\lit\in Y\cup \overline{Y}$ has only $v$ in her preferences.
The preference list of the special students are:
\begin{align*}
  d\colon & v\succ_d w\text{,} & d'\colon &v,\\
  r_1\colon& b\succ_{r_1} a\text{,} &  r_2\colon & b\succ_{r_2} w \succ_{r_2} a\text{,} &r_3\colon& a\succ_{r_3} b\text{.}
\end{align*}
The preferences of each clause student~$d_j$ and c-clause student~$e_j$, $j\in [\emm]$ are:
\begin{align*}
  d_j \colon v \succ_{d_j} b~\text{ and }~ e_j \colon b \succ_{e_j} v.
\end{align*}

\paragraph{The preference lists of the colleges.}
The preference lists of $v$ and $b$ are changed while the others remain the same:
\begin{align*}
  v\colon & e_1 \succ_v d_1 \succ_v \ldots \succ_v e_\emm \succ_v d_{\emm} \succ_v d \succ_v d' \succ_v\\
  &  [\{f_{\enn+1}, \ldots, \succ_v f_{2\enn}\}]  \succ_v [Y] \succ_v [\overline{Y}] \succ_v\\
  & f_1 \succ_v \overline{x}_1 \succ_v x_1 \succ_v \ldots \succ_v f_\enn \succ_v \overline{x}_\enn \succ_v x_\enn,\\
  w\colon & d \succ [F] \succ_w  r_2\text{,}\\
  a\colon & r_1 \succ_a  r_2 \succ_a r_3\text{,}\\
  b\colon & r_3 \succ_b  r_2 \succ_b r_1 \succ_b [X] \succ_b [\overline{X}] \succ_b [D] \succ [E] \text{.}
\end{align*}

This completes the construction, which can clearly be done in polynomial time.
In total, there are $5\enn+5+2\emm$~students and four colleges.
We denote the constructed instance by~$I'$.
We claim that $I$ is a yes instance if and only if the constructed \CSM{} instance~$I'$ admits a feasible and c-stable matching.

For the ``only if'' part, assume that $I$ is a yes instance.
Let~$\sigma_X$ be a truth assignment of $X$ which  ``witnesses'' that $I$ is a yes instance.
Let $X'\coloneqq \{x_i\mid \sigma_{X}(x_{i})=\true\} \cup \{\overline{x}_{i}\mid\sigma_{X}(x_i)=\false\}$.
Let $M$ be a matching constructed as follows:
\begin{enumerate}
  \item $M(v)\coloneqq X'\cup Y\cup \overline{Y}\cup D' \cup {\color{winered}E'}$, where $D'\coloneqq \{d_j \mid |C_j \cap (X'\cup Y\cup \overline{Y})| = 2\}$ and $E'=\{e_j \mid |C_j \cap (X'\cup Y\cup \overline{Y})| \neq 2\}$.
  \item $M(w)\coloneqq \{d\}\cup {\color{winered}F}$.
  \item $M(a)\coloneqq \{r_2\}$.
  \item $M(b)\coloneqq \{r_1,r_3\}\cup (X\cup \overline{X}\setminus X') \cup (D\cup E \setminus (D'\cup E'))$.
\end{enumerate}

\begin{table*}
  \noindent{
  \addtolength{\extrarowheight}{\belowrulesep}
  \aboverulesep=0pt
  \belowrulesep=0pt
  \tabcolsep=0pt
  \begin{tabular}{@{\,}>{\columncolor{white}[1pt][\tabcolsep]}p{.6cm}@{}>{\columncolor{white}[0pt][\tabcolsep]}p{1cm}|p{.6cm}p{1cm}|p{.6cm}p{1cm}|p{.6cm}p{1cm}|p{.6cm}p{1cm}@{}@{}>{\columncolor{white}[0pt][\tabcolsep]}l@{}>{\columncolor{white}[2.5pt][\tabcolsep]}
  l@{}>{\columncolor{white}[0pt][\tabcolsep]}l@{}>{\columncolor{white}[0pt][\tabcolsep]}l@{}>{\columncolor{white}[0pt][\tabcolsep]}l@{}>{\columncolor{white}[0pt][\tabcolsep]}l@{}>{\columncolor{white}[0pt][\tabcolsep]}c@{}}
    \toprule
    S.~~ & Pref.~~~ &   C.~ & \multicolumn{7}{l}{Pref.} & LQ. & UQ. & C.\\\midrule 
   \highlight{$r_1$:} & \highlight{$b \osucc a$} & \highlight{$w$:} & \multicolumn{7}{l}{\highlight{$d\osucc [F] \osucc r_2$}} &    \highlight{${0}^{2\enn+2\emm+3}$~~}  &  \highlight{${1}^{2\enn+\emm+2}{0}^{\emm+1}$~} & \highlight{${\color{winered}2\enn+1}$}\\
    \highlight{$r_2$:} & \highlight{$b \osucc w \osucc a$} & \highlight{$a$:} & \multicolumn{7}{l}{\highlight{$r_1 \osucc r_2 \osucc r_3$}}  &\highlight{${0}^{2\enn+2\emm+3}$}& \highlight{$0^{2\enn+\emm}110^{\emm+1}$}&\highlight{$1$}\\
    \highlight{$r_3$:} & \highlight{$a \osucc b$} & \highlight{$b$:} & \multicolumn{7}{l}{\highlight{$r_3 \osucc r_2 \osucc r_1 \osucc [X]  \osucc [\overline{X}] \osucc [D] \osucc [E]$}}  &\highlight{$1^{\enn}{0}^{\enn+2\emm+3}$~}& \highlight{$1^{\enn}0^{\enn}1^{2\emm+2}0$}& \highlight{$\enn+\emm+2$}\\
    $d$: & $v\osucc w$ &  $v$:& \multicolumn{7}{l}{$e_1\osucc d_1 \osucc \ldots \osucc e_\emm \osucc d_\emm \osucc d \osucc d'
                                 \osucc [\{f_{\enn+1},\ldots, f_{2\enn}\}] 
                                \osucc [Y] \osucc [\overline{Y}] \osucc f_1 \osucc \overline{x}_1 \osucc x_1 \osucc \ldots \osucc f_\enn \osucc \overline{x}_i \osucc x_i$~~} & $1^\enn2^\enn3^\emm00{\color{winered}1^{\emm+1}}$~~~& $1^\enn2^\enn3^\emm00{\color{winered}1^{\emm+1}}$~~~& $3\enn$$+$$\emm+1$\\
    \midrule
    S.~~ & Pref.~~~ & S.~~ & Pref.~~~ & S.~~ & Pref.~~~ & S.~~ & Pref.~~~ & S.~~ & Pref.~~~\\\midrule
  $x_i$:~ & $v \osucc b$ &   $\overline{x}_i$:~ &   $b\osucc v$ & $f_i$:~ &   $v\osucc w$ &  $d_j$:~ & $v \osucc b$ &    $e_j$:~ & $b \osucc v$ &\\\hline
    \bottomrule
  \end{tabular}
}
\caption{A description of the preference lists, quotas and capacities of the colleges, together with the preference lists of the students from~$\{r_1$, $r_2$, $r_3$, $d\}\cup X\cup \overline{X}\cup F\cup D \cup E$, used in the proof of \cref{prop:CSM-sigma2p-c}.}\label{fig:CSM:sigma2p-h-reduction}
\end{table*}

We claim that $M$ is feasible and c-stable.
It is straight-forward to verify that $M$ is feasible; recall that in instance~$I$ each clause contains at least two literals from $Y\cup \overline{Y}$ so that the quotas of the clause types for~$v$ are indeed fulfilled.

Suppose, for the sake of contradiction, that $M$ is not c-stable.
That is, it is blocked by some coalition~$(\beta, U'=\{u'_{i_1}, u'_{i_2},\ldots,u'_{i_k'}\})$ with $\beta \in \{w,v,a,b\}$.
For the sake of notation, let $U''\coloneqq U'\setminus M(\beta)$.
Similar properties as stated in \cref{lem:aux} also hold for blocking coalitions.
To be complete, we consider every college explicitly.
  \paragraph{Case~1: $\beta=a$.} Then, since $q_a=1$, it must hold that $k'=1$.
  However, similarly to the proof for \cref{lem:aux} we infer that no blocking coalition involves college~$a$.
  \paragraph{Case~2: $\beta=b$.} Then, $k'=\enn+\emm+2$ because $k'\ge |M(b)|=\enn+\emm+2=q_b$.
  If $r_2\in U'$ then $r_3\notin U'$ because of type~$2\enn+\emm+2$. However, college~$b$ strictly prefers $r_3$ to $r_2$, a contradiction to Property~\eqref{bc:college}.
  Thus, $r_2\notin U'$, and consequently $r_3,r_1\in U'$.
  In other words, $U''\cap \{r_1,r_2,r_3\} =\emptyset$.
  Since $U''$ is a blocking coalition and $k'=|M(b)|$,
  there must be some~$i\in [\enn]$ such that $x_i \in U''$ or
  there must be some~$j\in [\emm]$ such that $d_j\in U''$.
  Then, by the definition of $M$, it must hold that $x_i\in M(v)$ or $d_j\in M(v)$.
  However, both $x_i$ and $d_j$ strictly prefer her assigned college~$v$ to~$b$, a contradiction to Property~\eqref{bc:students} in the definition of blocking coalition.
  Thus, $U''=\emptyset$, implying that $U'=M(b)$ and no blocking coalitions involve college~$b$.
 \paragraph{Case~3: $\beta=w$.} This is not possible because college~$w$ already receives its $2\enn+2$ most preferred students.
 \paragraph{Case~4: $\beta=v$.} Then, $k'=2\enn+\emm+1$ because $|M(v)|=2\enn+\emm+2=q_v$.
 In this case, we first show that the following.

  \begin{claim} \label{CSM:noDE}
    It holds that $M(v) \cap (D\cup E) = U'\cap (D\cup E)$.
  \end{claim}
  \begin{proof}
  \renewcommand{\qedsymbol}{(of \cref{CSM:noDE}~$\diamond$)}
  Suppose, for the sake contradiction, that $M(v) \cap (D\cup E) \neq U'\cap (D\cup E)$.
  Observe that due to the quotas for the types~$\{2\enn+\emm+3,\ldots, 2\enn+2\emm\}$ for each $j\in [\emm]$ either $e_j$ or $d_j$ belongs to $M(v)$, and the same holds for $U'$.
  Hence, we infer that $|M(v)\cap (D\cup E)| = |U'\cap (D\cup E)|=\emm$.
  This means in particular that when we compare the students in $U'$ with the students in $M(v)$,
  for each $j' \in [\emm]$ we have to compare $U'\cap \{e_{j'},d_{j'}\}$ to $M(v)\cap \{e_{j'},d_{j'}\}$.
  If $e_{j'}\in U''$, then by the definition of $M$ it must hold that $e_{j'}\in M(b)$.
  However, $e_{j'}$ strictly prefers her assigned college~$b$ to $v$, a contradiction to Property~\eqref{bc:students}.
  If $d_{j'}\in U''$, then by the definition of $M$, it follows that $e_{j'}\in M(v)$.
  However, $v$ prefers $e_{j'}$ to $d_{j'}$, a contradiction to Property~\eqref{bc:college}.
  \end{proof}
  
  By the quotas of $Y$-types for $v$, it must hold that for each $i\in \{\enn+1,\ldots, 2\enn\}$
  at least one of $\{y_i,\overline{y}_i\}$ belongs to $U'$.
  We distinguish between two cases.
  \begin{compactenum}
    \item[Case~4.1: $|\{y_i,\overline{y}_i\}\cap U'|=2$.]
    Then, by the quotas of the $Y$-types, it must hold that $d\notin U'$.
    Consequently, it holds that $Y\cup \overline{Y}\subseteq U'$.
    Moreover, for each $i\in [\enn]$ either $x_i$ or $\overline{x}_i$ belongs to $U'$.
    Let $A\coloneqq U'\cap (X\cup \overline{X})$.
    Since $3\enn+\emm+1=q_v \ge k' \ge |M(v)|=3\enn+\emm+1$, 
    it follows that $\{d'\}\cup Y\cup \overline{Y} \subseteq U'$.
    Moreover, by \cref{CSM:noDE} when we compare the students in $U'$ and $M(v)$ for Property~\eqref{bc:college},
    we have to compare the students in $A$ with the students in $X'$ because all other students remain the same in $M(v)$ as well as in $U'$.
    For each literal student~$\lit \in A$, if $\lit \notin X'$, then it must be the case that $\overline{x}_i\in A$ and $x_i \in X'$ (see Property~\eqref{bc:college} in the definition of blocking coalitions).
    However, $\overline{x}_i$ strictly prefers her assigned college~$b$ to~$v$, a contradiction to Property~\eqref{bc:students}.
    Hence, $A = X'$, i.e.,
    \begin{align}\label{CSM:noX-replaced}
      U''\cap (X\cup X') = \emptyset.
    \end{align}
    
    Summarizing, we have just reasoned that $U'=M(v)$, a contradiction to $\{v,U'\}$ be a blocking coalition.
    \item[Case~4.2: $|\{y_i, \overline{y}_i\}\cap U'|=1$.]
    Then, by the quotas of the $Y$-types and the type~$2\enn+2\emm+3$, it must hold that $d\in U'$
    and  $d'\notin U'$.
    By the quotas of the $X$-types, it must also hold that
    \begin{align}\label{CSM:noX-replaced-2}
      U'\cap (X\cup \overline{X})=\emptyset.
    \end{align}    
    Moreover, there exist a size-$\enn$ subset~$Y'\subseteq U'$ of literal students such that for each $i\in \{\enn+1,\ldots, 2\enn\}$ either $y_i$ or $\overline{y}_i$ belongs to $Y'$.
    Let $Y''\coloneqq (Y\cup \overline{Y}) \setminus Y'$.
    Let $\sigma_Y$ be the truth assignment corresponding to~$Y''$.
    Since $I$ is a yes instance, there exists a clause~$C_j$ which is not 1-in-3-satisfied under $\sigma_X$ and $\sigma_Y$, i.e.,
    \begin{align}\label{CSM:not1in3C}
      |C_j\cap (X'\cup Y'')|\neq 1.
    \end{align}
    If $d_j\in M(v)$, then $|C_j\cap (X'\cup Y \cup \overline{Y})| = 2$.
    Moreover, $d_j \in U'$ because of \cref{CSM:noDE}.
    Since $d\in U'$ and $U'$ is ``feasible'' for $v$ regarding type~$2\enn+j$,
    it follows that $|C_j\cap Y'|=1$.
    Since $Y'$ and $Y''$ partition $Y\cup \overline{Y}$ it follows that $|C_j \cap (X'\cup Y'')| = 1$,
    a contradiction to~\eqref{CSM:not1in3C}.

    If $d_j \notin M(v)$, then $d_j \notin U'$ due to \cref{CSM:noDE}.
    Since $d\in U'$,  to make $U'$ feasible for $v$ regarding type~$2\enn+j$ it must hold that $|C_j\cap Y'| = 2$. 
    By \eqref{CSM:not1in3C}, it follows that $|C_j\cap (X'\cup Y \cup \overline{Y})| \neq 3$,
    a contradiction to the definition of $M$ that $d_j\notin M(v)$.
  \end{compactenum}

We have just shown the ``only if'' part.
For the ``if'' part, assume that $I'$ admits a feasible and stable matching, called~$M$.
Suppose, for the sake of contradiction, that $I$ is a no instance.
First of all, we show that $M(w)=d\cup F$.
Suppose, for the sake of contradiction, that $M(w) \neq d\cup F$.
Then, it must hold that $r_2 \in M(w)$ or $r_2 \in M(b)$ as otherwise $w$ forms a blocking coalition with $M(w)\cup \{r_2\}$.
If $r_2 \in M(w)$, then $r_3 \in M(b)$ as otherwise $b$ forms a blocking coalition with $(M(b)\setminus \{r_1\}) \cup \{r_2\}$.
However, then $r_1\in M(b)$ as otherwise $b$ forms a blocking coalition with $M(b)\cup \{r_1\}$.
Then, $a$ does not receive any student and will form a blocking coalition with $\{r_3\}$, a contradiction.
If $r_2\in M(b)$, then neither $r_1$ nor $r_3$ can be matched to $b$ any more because of the upper quotas regarding types~$2\enn+\emm+1$ and $2\enn+\emm+2$.
However, since $q_a=1$, it follows that $r_1$ or $r_3$ will be unmatched.
In the former case, $a$ and $\{r_1\}$ form a blocking coalition.
In the latter case, $b$ forms a blocking coalition with $(M(b)\setminus \{r_2\}) \cup \{r_3\}$.
This implies that $M(w)=d\cup F$.
In particular, $d\notin M(v)$.

Due to the quotas on the $X$-types, there exists a subset~$X'$ with $|X'|=\enn$
such that for each $i\in [\enn]$ it holds that $|M(v)\cap \{x_i,\overline{x}_i\}|=1$
such that $X' \subseteq M(v)$.
In other words, $X'$ corresponds to a truth assignment.
Let $\sigma_X$ be such a truth assignment.
Due to the quotas on the clause types~$2\enn+1$ to $2\enn+\emm$,
there exists a subset~$D'\subseteq M(v)$ such that $D'=\{d_j \in D \mid |M(v)\cap (X'\cup Y\cup \overline{Y})| = 2\}$.
Moreover, by the quotas of the types~$2\enn+\emm+3,\ldots, 2\enn+2\emm+2$
there exists a subset~$E'$ with $E'= \{e_j \in E \mid  e_j\notin D'\}$.

Since $I$ was a no instance, for the truth assignment~$\sigma_X$ there must be a truth assignment~$\sigma_Y$ for $Y$ such that each clause~$C_j$ is 1-in-3-satisfied.
Let $Y'\coloneqq \{y_i \mid \sigma_Y(y_i)=\textcolor{winered}{\false}\}\cup \{\overline{y}_i\mid \sigma_Y(y_i)=\textcolor{winered}{\true}\}$.
We claim that $v$ forms a blocking coalition with $U'=D'\cup E' \cup \{d\} \cup Y' \cup F$.
First of all, we verify that $|D'\cup E'\cup \{d\}\cup Y'\cup F|=3\enn+\emm+1$.
Moreover, $Y\cup \overline{Y}\subseteq M(v)$ due to the quotas of the $Y$-types.
Thus, each student~$u\in U'\setminus M(v) = \{d\}\cup F$ prefers $v$ to her assigned college~$w$.
Further, college~$v$ prefers~$d$ to~$d'$.
Moreover, if $y_i \in Y'$, then we can use the fact that $v$ prefers~$f_i$ to~$\overline{y}_i$;
if $\overline{y}_i \in Y'$ then we can use the fact that $v$ prefers~$f_i$ to~$y_i$.
College~$w$ also prefers each dummy variable student~$f_i$ to~$x_i$ and~$\overline{x}_i$.
Hence, in order to prevent $\{v,U'\}$ from forming a blocking coalition,
there must exists a type~$z\in [2\enn+2\emm+3]$ so that one of the quota is not satisfied.
Clearly, $z\notin [2\enn]$ and $z \notin \{2\enn+\emm+1,\ldots,2\enn+2\emm+3\}$.
Thus, $z$ corresponds to a clause type from $\{2\enn+1,\ldots, 2\enn+\emm\}$.
Let $z=2\enn+j$.
Then, $|C_j \cap (Y'\cup D' \cup \{d\})| \neq 3$.
Since $d$ has type at $z$ it follows that
\begin{align}\label{CSM:d}
  |C_j \cap (Y'\cup D')|\neq 2.
\end{align}
If $d_j \in D'$, then by the definition of $M$ it must hold that $|C_j \cap (X'\cup Y \cup \overline{Y})| = 2$.
Moreover, by \eqref{CSM:d} we also know that $|C_j \cap Y'|\neq 1$.
Together, it follows that $C_j \cap (X'\cup (Y\cup \overline{Y})\setminus Y')| \neq 1$, a contradiction to the assumption that $C_j$ is 1-in-3-satisfied under $\sigma_X\cup\sigma_Y$.

If $d_j \notin D'$, then by the definition of $M$ it must hold that $|C_j \cap (X'\cup Y \cup \overline{Y})| = 3$.
Together with \eqref{CSM:d} it follows that $|C_j \cap (X'\cup (Y \cup \overline{Y}) \setminus Y')| \neq 1$, again a contradiction to the assumption that $C_j$ is 1-in-3-satisfied under $\sigma_X\cup\sigma_Y$.

It remains to the show the $\Sigma^{\text{P}}_2$-containment.
Similarly to the one shown in \cref{thm:diverse:lq>0:sigma2p-c}.
We first show that checking whether a given matching~$M$ is not feasible or not c-stable can be done in polynomial time: Guess a college~$w$ and a subset~$U'$ of students, and check in polynomial time whether $M$ is not feasible or $w$ forms a blocking coalition with $U'$.
This means that there exists an NP-oracle for checking whether a matching is not feasible or not c-stable.
Then, we can use exactly the same reasoning as the one for \cref{thm:diverse:lq>0:sigma2p-c} to show that \CSM is in $\Sigma^{\text{P}}_2$: Guess a matching and ask the NP-oracle whether matching is not feasible or not c-stable; return yes if and only if the oracle answer no.
\end{proof}
\fi

\ifshort \paragraph{The Impact of Diversity.}\else
\subsection{The Impact of Diversity}\fi
The fact that \diverse lies in a higher complexity class than \feasible\ can be attributed to the stability constraints, which are not taken into account for feasible matchings.
	On the other end of the spectrum, a stable matching without diversity constraints always exists and can be found in polynomial time~\cite[Chapter 3]{Manlove2013}; the corresponding problem is known as the \textsc{Hospitals/Residents problem with Ties (HRT)}.
We can pinpoint the cause of this jump in complexity more precisely to the existence of lower quotas, which in some sense implement affirmative action in the \diverse model.
\ifshort Specifically, we show that if the lower quotas are all zero, then \diverseties{} becomes NP-complete (\cref{thm:diverse:NP-c:t=2}). Lemma~\ref{lem:lq=0:stable-equiv-form} will be crucial for showing NP-containment.
\fi
\iflong Specifically, we show that if the lower quotas are all zero, then \diverseties{} is in NP (\cref{thm:diverse:NP-c:t=2}) by showing a statement which is complementary to the coNP-hardness result in Proposition~\ref{prop:checkingMstable:coNP-h}.
\fi
\begin{lemma}\label{lem:lq=0:stable-equiv-form}
  If $\lmax=0$, then checking whether a matching is stable can be done in $\mathcal{O}(n\cdot m \cdot t)$~time.
\end{lemma}

\iflong
\begin{proof}
	Let $I=(U,W,T, (\succeq_x)_{x\in U\cup W}, (\typevec_{u})_{u\in U}, (q_w,\lowervec_w=\boldsymbol{0},\uppervec_w)_{w\in W})$ be an instance of \diverseties{} where $\lowervec_w=\{0\}^{|T|}$ and let $M$ be an arbitrary matching for~$I$.
	To check whether $M$ is stable for~$I$ in polynomial time, we show that $M$ is stable if and only if for each unmatched pair~$\{u,w\}\notin M$ with $u\in U$
	and $w\in W$ such that $u$ prefers $w$ to $M(u)$ it holds that
	\begin{align}
	\label{bp-cond} \typevec_{u}+\sum_{u'\in M(w)\text{ with } u'\succeq_w u}\typevec_{u'}\not\le \uppervec_{w}\text{; or }\\
	\label{bp-q} |M(w)|=q_w \wedge \forall u' \in M(w) \ u' \succeq_w u.
	\end{align}
	
	For the ``only if'' part, assume that $M$ is stable.
	Suppose, for the sake of contradiction, that there exists an unmatched pair~$\{u,w\}\notin M$ with $u\in U$ and $w\in W$ such that $u$ prefers $w$ to $M(u)$ \emph{neither} Condition~\eqref{bp-cond} \emph{nor} Condition~\eqref{bp-q} in the lemma holds.
	Since Condition~\eqref{bp-cond} does not hold, it follows that
	\begin{align}
	\lowervec_w = \boldsymbol{0} \le \typevec_u + \sum_{u'\in M(w) \text{ with } u' \succeq_w u} \typevec_{u'}\le \uppervec_w.\label{bp-cond2}
	\end{align}
	Define $S'\coloneqq \{v\in M(w)\mid u\succ_w v\}$.
	We distinguish between two cases.
	If $S'\neq \emptyset$, then $\{u,w\}$ is blocking $M$ since \eqref{bp-cond2} holds and replacing $S'$ to $w$ results in a feasible matching for $w$, a contradiction.
	If $S'=\emptyset$, then since \eqref{bp-q} does not hold, it follows that $|M(w)| < q_w$.
	In other words,  $\{u,w\}$ is blocking $M$ since \eqref{bp-cond2} holds, $|M(w)\setminus S'| < q_w$ and replacing $S'$ to $w$ results in a feasible matching for $w$, a contradiction.
	
	For the ``if'' part, assume that Condition~\eqref{bp-cond} or Condition~\eqref{bp-q} holds for each unmatched pair~$\{u,w\}\notin M$ with $u$ preferring $w$ to $M(u)$.
	Suppose, for the sake of contradiction that, $M$ is not stable.
	Let $\{u,w\}$ be a blocking pair of $M$.
	That is, $\{u,w\}\notin M$,
	$u$ prefers $w$ to $M(u)$,
	and there exists a (possibly empty) subset of students~$S\subseteq M(w)$ assigned to~$w$ such that
	\begin{enumerate}[(i)]
		\item\label{it:prefer} $w$ prefers $u$ strictly to every student in~$S$,
		\item\label{it:S} $\lowervec_{w}\le \typevec_{u}+\sum_{u'\in M(w)\setminus S}\typevec_{u'}\le \uppervec_{w}$, and 
		\item\label{it:q} $|M(w)\setminus S|+1 \le q_w$.
	\end{enumerate}
	Let $S'\coloneqq \{u'\in M(w)\mid u'\succeq_w u\}$. Then, by Condition~\eqref{it:prefer}, it follows that
	$S'\subseteq M(w)\setminus S$.
	By Condition~\eqref{it:S}, it follows that 
	\begin{align*}
	\typevec_{u}+\sum_{u'\in S'}\typevec_{u'}\le \typevec_{u}+\sum_{u'\in M(w)\setminus S}\typevec_{u'}\le \uppervec_{w},
	\end{align*} implying that Condition~\eqref{bp-cond} does not hold.
	By our assumption, Condition~\eqref{bp-q} must hold.
	In other words, $|M(w)|=q_w$ and $S=\emptyset$, a contradiction to Condition~\eqref{it:q}.
	
	Using the above equivalent formulation, one can go through each student~$u\in U$ and each college~$w\in W$ with $w\succ_u M(u)$ 
	and check in $\mathcal{O}(t)$~time whether Condition~\eqref{bp-cond} or Condition~\eqref{bp-q} holds.
	This gives the promised running time.
\end{proof}
\fi
\ifshort
\begin{proof}[Proof Sketch]
  It suffices to show that a matching~$M$ is stable for an instance~$I=(U,W,T, (\typevec_u,\succeq_u)_{u\in U}, (\succeq_w,\lowervec_w=\boldsymbol{0},\uppervec_w,q_w)_{w\in W})$ if and only if the following (polynomially verifiable) condition is met: for each unmatched student-college pair~$\{u,w\}\notin M$ with $u$ preferring~$w$ to~$M(u)$
      either ``$\typevec_{u}+\sum_{u'\in M(w)\colon u'\succeq_w u}\typevec_{u'}\not\le \uppervec_{w}$'' or ``$|M(w)|$$=$$q_w$ and $M(w) \succeq_w u$'', or both holds.
\end{proof}
\fi

\iflong From \cref{lem:lq=0:stable-equiv-form}, we immediately obtain:
\begin{proposition}\label{prop:lq=0:inNP}
  \diverseties{} and \diverse\ restricted to instances where $\lmax=0$ are in NP.
\end{proposition}
\fi
\iflong
\begin{proof}
	To show NP membership, one only needs to guess in polynomial time a matching~$M$,
	and check whether $M$ is feasible in polynomial time, and whether it is stable in polynomial time, using \cref{lem:lq=0:stable-equiv-form}.
\end{proof}
\fi

Even though with zero lower quotas \diverseties is in NP, and thus can be considered significantly easier than \diverseties in general, it is actually hard within NP.
Hardness for this case was claimed in \cite[Proposition~5.3]{AzizGaspersSunWalsh2019aamas}.
%
However their reduction contains a technical flaw. 
\ifshort
Indeed, the instance constructed in that proof is always a yes instance,
independent of the original \textsc{3-SAT} instance.
To see this, define the matching~\(M\) for their produced instance as follows (notations taken from that proof):
First, let \(M_{\text{F}} \coloneqq\) $\bigcup_{i\in [k]}X_F^i \setminus \{ (t_1^i, o(t_1^i)), (t_2^i,o(t_2^i)) \mid i\in [k]\}$.
For each \(j \in [l]\), let \(S_j\) be the set consisting of the first two (if there are fewer than two, then all) students of the form~$t_{z}^{i}$ ($i\in [k]$, $z\in [2]$) appearing in the preferences of $o_j$.
Then \(M \coloneqq M_{\text{F}} \cup \{(s, o_j) \mid j\)\(\in\)\([l] \land s\)\(\in\)\(S_j\}\) is feasible and stable.

Below, we use Lemma~\ref{lem:aux} to provide a new and simpler NP-hardness proof for the case with zero lower quotas. 
\fi
\iflong
Interested readers can refer to \cref{rem:Amaas-wrong-construction} for details.

\begin{observation}\label{rem:Amaas-wrong-construction}
  \cite[Lemma 5.7]{AzizGaspersSunWalsh2019aamas} which is used to prove \cite[Proposition 5.3]{AzizGaspersSunWalsh2019aamas} is incorrect.
  In fact, even the construction in the proof of \cite[Proposition 5.3]{AzizGaspersSunWalsh2019aamas} is incorrect.
  In particular, the  instance constructed in that proof is always a yes instance, independent of the original \textsc{3-SAT} instance.
\end{observation}
\begin{proof}
  To see this, define the matching~\(M\) for their produced instance as follows (notations taken from that proof):
  First, let Let \(M_{\text{F}} \coloneqq \{(s^i_1,c^i_{t_1}),(s^i_2,c^i_2),(s^i_3,c^i_1),(s^i_4,c^i_1),(s^i_5,c^i_{t_2}),(s^i_6,c^i_2),(f^i_1,$ $c^i_{f_1}),(f^i_2,c^i_{f_2}) \mid i \in [k]\}\).
  Note that $M_{\text{F}} = \bigcup_{i\in [k]}X_F^i \setminus \{ (t_1^i, o(t_1^i)), (t_2^i,o(t_2^i)) \mid i\in [k]\}$.
  For each \(j \in [l]\), let \(S_j\) be the set consisting of the first two (if there are fewer than two, then all) students of the form~$t_{z}^{i}$ ($i\in [k]$, $z\in [2]$) appearing in the preferences of $o_j$. 
  For instance, if $C_j =(\overline{x}_1\wedge x_2 \wedge x_3)$, where $C_j$ is the second clause where $x_2$ appears as a true literal and the first clause where $x_3$ appears as a true literal,
  then $S_j = \{t_2^2, t_1^3\}$.
  One can verify that \(M \coloneqq M_{\text{F}} \cup \{(s, o_j) \mid j\)\(\in\)\([l] \land s\)\(\in\)\(S_j\}\) is feasible and stable for the constructed instance in that proof.
\end{proof}

Below, we use Lemma~\ref{lem:aux} to provide a new NP-hardness proof for the case when there are no lower quotas. 

\fi 

\begin{theorem}\label{thm:diverse:NP-c:t=2}
  For $\lmax=0$, \diverse{} is \textnormal{NP}-complete; it remains \textnormal{NP}-hard even if $\lmax=0$, \(\umax=1\) and \(t=\qmax=2\).
\end{theorem}

\iflong
\begin{proof}
  The containment for the case with $\lmax=0$ follows from \cref{prop:lq=0:inNP}.
  To show the NP-hardness, we reduce from (2,2)-\ethreesat, an NP-complete variant~\cite{BerKarSco2003223SAT} of 3SAT where each literal $x\in X\cup \overline{X}$ appears precisely two times in the set $\phi(X)=(C_1,C_2,\ldots,C_{\emm})$ of clauses.

  Given an instance~$I=(X=\{x_1,\ldots,x_{\enn}\},\phi(X)=\{C_1,\ldots,C_{\emm}\})$ of (2,2)-\ethreesat, we construct an instance of \diverse{} as follows.
  For each clause~$C_j \in \phi(X)$, introduce a clause college~$c_j$.
  For each variable~$x_i\in X$, introduce
  two \myemph{variable students}, called~$x_i$ and~$y_i$,
  four \myemph{literal students}, called~$u_i^1$, $u_i^2$, $v_i^1$, and $v_i^2$,
  and two \myemph{variable colleges}, called~$w_i$ and $p_i$.
  Finally, introduce three special students, called~$r_1,r_2,r_3$,
  and two special colleges, called~$a$ and $b$.
  Let there be two types; i.e., $T=\{1,2\}$.
  
  For ease of description of the preference lists we use the following notation:
  let \myemph{$c[u_i^z]$} and \myemph{$c[v_i^z]$}, ($z\in [2]$) be the clause student~\(c_j\) such that the clause~\(C_j\) contains the $z^{\text{th}}$ occurrence of literal~$x_i$ and $\overline{x}_i$, respectively.
  Further, let \myemph{$s^{z'}[c_j]$ ($z'\in [3]$)} denote the literal student that corresponds to the ${z'}^{\text{th}}$ literal appearing in clause~$C_j$.
  For instance, if $C_j=(\overline{x}_2, x_3,\overline{x}_5)$ and the occurrence of \(\overline{x}_2\) in \(C_j\) is its second one,
  then $s^1[c_j] = v_2^2$.
  \begin{table}[t!]\centering
  \begin{tabular}{@{}>{\columncolor{white}[1pt][\tabcolsep]}l@{}>{\columncolor{white}[1pt][\tabcolsep]}lc@{}l|l@{}lc@{\;\;}c}
    \toprule
    S. & Pref. &T. && C. &  Pref. & UQ. & C.\\\midrule 
    $x_i\colon$ &  $p_i\osucc w_i$ & $10$ & & $w_i\colon$ & $x_i \osucc u_i^1 \osucc y_i \osucc  u_i^2$ & $11$ & $2$ \\
    $y_i\colon$ & $w_i \osucc p_i$ & $01$ && $p_i\colon$ & $y_i \osucc v_i^1 \osucc x_i \osucc  v_i^2$ & $11$ & $2$\\
    $u_i^1\colon $ & $w_i \osucc c[u_i^1]$ & $11$ & &  \\
    $u_i^2\colon $ & $w_i \osucc c[u_i^2]$ & $00$ & &  \\
    $v_i^1\colon $ & $p_i \osucc c[v_i^1]$ & $11$ & &  \\
    $v_i^2\colon $ & $p_i \osucc c[v_i^2]$ & $00$ & &  \\
    \highlight{$r_1\colon$} & \highlight{$b \osucc a$} & \highlight{$10$}  && \highlight{$c_j\colon$} & \highlight{$s^1[c_j] \osucc s^2[c_j] \osucc s^3[c_j] \osucc r_2$} & \highlight{$11$} & \highlight{$1$}\\
    \highlight{$r_2\colon$} & \highlight{$b \osucc [C] \osucc a$} &\highlight{$11$}&& \highlight{$a\colon$} & \highlight{$r_1 \osucc r_2 \osucc r_3$} & \highlight{$11$} & \highlight{$1$}\\
    \highlight{$r_3\colon$} & \highlight{$a \osucc b$} &\highlight{$01$}& & \highlight{$b\colon$} & \highlight{$r_3 \osucc r_2 \osucc r_1$} & \highlight{$11$} & \highlight{$2$}\\\bottomrule
   \end{tabular}
   \caption{The types and the preference lists of the students and the colleges are as follows where we omit ``$\succ$'' to save space, let $i\in [\enn]$ and $j\in [\emm]$, and let $[C]\coloneqq c_1\succ\cdots \succ c_{\emm}$, ``S.'', ``Pref.'', ``T.'', ``C.'', ``UQ.'', and ``C.'' stand for ``student'', ``preference list'', ``type vector'', ``college'', ``upper quota'', and capacity, respectively.
}\label{tab:description-for-NPh-red-lmax=0-t=2}
 \end{table}
  All lower quotas are defined as zero.
  This completes the construction of the instance for \diverse.
  One can verify the restrictions stated in the theorem.
  
  It remains to show that $(X,\phi(X))$ admits a satisfying truth assignment if and only if the constructed instance admits a feasible and stable matching.

  For the ``only if'' part, let the truth assignment~$\sigma_X$ satisfy $\phi(X)$.
  We claim that the following matching~$M$ is feasible and stable for the constructed \diverse instance.
  \begin{itemize}
    \item For each variable~$x_i\in X$, if $\sigma_{X}(x_i)=\true$ let $M(w_i)\coloneqq \{x_i,y_i\}$ and $M(p_i)\coloneqq \{v_i^1, v_i^2\}$,
    and otherwise let $M(w_i)\coloneqq \{u^1_i,u^2_i\}$ and $M(p_i)\coloneqq\{x_i,y_i\}$.
    \item For each clause~$C_j\in \phi(X)$, let \(M(c_j)\coloneqq \{s^z(c_j)\}\), where \(z \in [3]\) is the smallest index such that the \(z^\text{th}\) literal in \(C_j\) is set to \(\true\) under \(\sigma_X\); note that there exists at least one such literal since $\sigma_{X}$ is a satisfying assignment.
    \item Let $M(a)\coloneqq \{r_2\}$ and $M(b)\coloneqq \{r_1,r_3\}$.
  \end{itemize}

  It is straightforward to verify that $M$ is feasible.
  To show stability, assume for contradiction that there is some blocking pair \(\{\alpha, \beta\}\) for \(M\) with $\beta \in \{w_i,p_i\mid i \in [\enn]\}\cup \{c_j\in j\in [\emm]\}\cup \{a,b\}$.
  Let this be witnessed by \(S' \subseteq M(\beta)\).
  Observe that $\{u_i^z,v_i^z\mid i\in [\enn], z\in [2]\}\uplus \{r_1,r_2,r_3\}$
  and $\{c_j\mid j\in [\emm]\}\cup \{a,b\}$ correspond to exactly the situation in \cref{lem:aux}, 
  apart from the fact that the students \(u^1_i\) and \(v^1_i\) have both special types.
  Note however that the fact that students in \(U \setminus \{r_1, r_2, r_3\}\) do not have the two types considered in Lemma~\ref{lem:aux}, is not used in the proof of Lemma~\ref{lem:aux}\eqref{aux-nec} at all.
  By \cref{lem:aux}\eqref{aux-nec} and by the definition of~$M$, it follows that $\alpha \in \{x_i,y_i,u_i^1,u_i^2,v_i^1,v_i^2\}$ for some~$i\in [\enn]$.
  If $\alpha=x_i$, then, as \(M(x_i) \neq \bot\), it follows that $\beta=p_i$ and $M(p_i)=\{v_i^1,v_i^2\}$.
  Since $p_i$ already receives two students and prefers $x_i$ to only $v_i^2$ it follows that $S'=\{v_i^2\}$. However, replacing $\{v_i^2\}$ with $x_i$ exceeds the quota regarding type~$1$, a contradiction.
  Similarly, we infer that $\alpha\neq y_i$.
  Next, if $\alpha=u_i^1$, then \(\beta \in \{w_i,c[u_i^1]\}\).
  We distinguish between two cases:
  \begin{compactenum}[(i)]
    \item If $\beta=w_i$, then $M(w_i)=\{x_i,y_i\}$.
    Since $w_i$ already receives two students and prefers~$u_i^1$ to only~$y_i$ (among all students assigned to~$w_i$)
    it follows that $S'=\{y_i\}$.
    However, replacing $\{y_i\}$ with $u_i^1$ exceeds the quota regarding type~$1$, a contradiction.
    \item If $\beta=c[u_i^1]$, then there exists some clause~$C_j$ such that $c[u_i^1]=c_j$.
    By the definition of $M$ there exists another literal student~$s^z[c_j]$ who is assigned to $c_j$ such that $c_j$ prefers~$s^z[c_j]$ to $u_i^1$, a contradiction.
  \end{compactenum}
  Similarly, we infer contradiction for the case that \(\alpha \in \{u_i^2,v_i^1,v_i^2\}\).
  As all cases lead to a contradiction, there is no blocking pair \(\{\alpha, \beta\}\), thus \(M\) is stable.

  For the ``if'' part, let $M$ be a feasible and stable matching for the constructed \diverse instance.
  Define the following assignment~$\sigma_X$ with
  $\sigma_X(x_i)\coloneqq \true$ if there exists a clause college~$c_j$ such that $u_i^1$ or $u_i^2$ is assigned to $c_j$; let $\sigma_X(x_i)\coloneqq \false$ if there exists a clause college~$c_j$ such that $v_i^1$ or $v_i^2$ is assigned to $c_j$.
  If no such clause college exists, then the truth value of $x_i$ can be arbitrary; for instance, let $\sigma_{X}(x_i)=\true$.
  
  We show that \(\sigma_X\) is well-defined, i.e., for each variable~$x_i$ there exist no two clause colleges~$c_j$ and $c_{j'}$ with $M(c_j)\cap \{u_i^1,u_i^1\}\neq \emptyset$ and $M(c_{j'})\cap \{v_i^1,v_i^2\}\neq \emptyset$.
  Suppose, for the sake of contradiction, that $c_j$ and $c_{j'}$ are two clause colleges such that
  $M(c_j)=\{u_i^z\}$ and $M(c_{j'})=\{v_i^{z'}\}$ for some $i \in [\enn]$ and $z,z'\in [2]$.
  We distinguish two cases:
  If $x_i \in M(w_i)$, then $y_i \in M(p_i)$ as otherwise $v_i^{z'}$ forms a blocking pair with \(p_i\).
  By the upper quotas of $p_i$ it follows that $v_i^1\notin M(p_i)$.
  However then $x_i$ forms a blocking pair with $p_i$, a contradiction.
  If $x_i\notin M(w_i)$, then $u_i^1\in M(w_i)$ because otherwise $u_i^1$ forms a blocking pair with \(w_i\).
  By the upper quotas of $w_i$ it follows that $y_i\notin M(w_i)$.
  Then, it must hold that $u_i^2\in M(w_i)$ because otherwise $u_i^2$ forms a blocking pair with $w_i$.
  This is a contradiction to $M(c_j)=u_i^z$ with $z\in [2]$.
  
  Finally, we show that \(\sigma_{X}\) satisfies \((X, \phi(X))\).
  As observed earlier, $\{u_i^z,v_i^z\mid i\in [\enn], z\in [2]\}\uplus \{r_1,r_2,r_3\}$
  and $\{c_j\mid j\in [\emm]\}\uplus \{a,b\}$ correspond to exactly the situation in \cref{lem:aux},
  except for the types of the students~\(u^1_i\) and \(v^1_i\).
  Note that the fact that all students in \(U \setminus \{r_1, r_2, r_3\}\) have no types is only used in the proof of Lemma~\ref{lem:aux}\eqref{aux-card} to show that if \(|M(w) \cap (U \setminus \{r_1, r_2, r_3\})| < q_w\) and \(r_2 \notin M(w)\) for any \(w \in W \setminus \{a, b\}\), \(\{r_2, w\}\) is blocking.
  As in this case for all \(w \in \{c_j\mid j\in [\emm]\}\), \(q_w = 1\), this is obviously still the case, independently of the types of the \(u^1_i\) and \(v^1_i\) students.
  By \cref{lem:aux}\eqref{aux-card} and by the capacities of the clause colleges, it follows that for each clause~$C_j$ there exists a literal~$\lit_j^z\in C_j$ such that $M(c_j)=\{s^z[c_j]\}$, and thus under \(\sigma_{X}\) there is a true literal for each clause.
\end{proof}
\fi

\ifshort
\begin{proof}[Proof Sketch]
To show NP-containment, we guess in polynomial time a matching~$M$, and check whether $M$ is feasible and stable in polynomial time, using \cref{lem:lq=0:stable-equiv-form}.

To establish NP-hardness, we reduce from (2,2)-\ethreesat, an NP-complete variant~\cite{BerKarSco2003223SAT} of 3SAT where each literal $\lit\in X\cup \overline{X}$ appears precisely two times in the set~$\phi(X)$ of clauses.
  Given an instance~$I=(X=\{x_1,\ldots,x_{\enn}\},\phi(X)=\{C_1,\ldots,C_{\emm}\})$ of (2,2)-\ethreesat, construct an instance of \diverse{} as follows.
  For each clause~$C_j \in \phi(X)$, introduce a \myemph{clause college~$c_j$}.
  For each variable~$x_i\in X$, introduce
  two \myemph{variable students}~$x_i$ and~$y_i$,
  four \myemph{literal students}~$u_i^1$, $u_i^2$, $v_i^1$, and $v_i^2$,
  and two \myemph{variable colleges}~$w_i$ and $p_i$.
  Introduce three special students~$r_1,r_2,r_3$,
  and two special colleges~$a$ and $b$.
  Let $T=\{1,2\}$.
  
  For ease of description of the preference lists we use the following notation:
  let \myemph{$c[u_i^z]$}, and \myemph{$c[v_i^z]$}, ($z\in [2]$) be the clause student~\(c_j\) such that the clause~\(C_j\) contains the $z^{\text{th}}$ occurrence of literal~$x_i$, and $\overline{x}_i$ respectively.
  Further, let \myemph{$s^z[c_j]$ ($z\in [3]$)} denote the literal student that corresponds to
  the $z^{\text{th}}$~literal appearing in clause~$C_j$.
  For instance, if $C_j=(\overline{x}_2, x_3,\overline{x}_5)$ and the occurrence of \(\overline{x}_2\) in \(c_j\) is its second one,
  then $s^1[c_j] = v_2^2$.
  Types and preference lists are given below, in a format similar to the one used in Lemma~\ref{lem:aux}.
  
 \noindent{
 	\tabcolsep=0pt
 	\aboverulesep=0pt
 	\belowrulesep=0pt \begin{tabular}{@{}>{\columncolor{white}[0pt][\tabcolsep]}l@{}>{\columncolor{white}[0pt][\tabcolsep]}l@{}>{\columncolor{white}[0pt][\tabcolsep]}c@{}|@{}>{\columncolor{white}[0pt][\tabcolsep]}l@{}>{\columncolor{white}[0pt][\tabcolsep]}l@{}>{\columncolor{white}[0pt][\tabcolsep]}c@{}|@{\,}>{\columncolor{white}[0pt][\tabcolsep]}l@{}>{\columncolor{white}[0pt][\tabcolsep]}l@{}>{\columncolor{white}[0pt][\tabcolsep]}l@{}>{\columncolor{white}[0pt][\tabcolsep]}c@{}>{\columncolor{white}[0pt][\tabcolsep]}c@{}}
    \toprule
    S. & Pref. &T. &  S. & Pref. &T. &~& C. &  Pref. & UQ. & C.\\\midrule 
    $u_i^1\colon $ & $w_i \osucc c[u_i^1]$ & $11$ & \highlight{$r_1\colon$} & \highlight{$b \osucc a$} & \highlight{$10$}  && \highlight{$c_j\colon$ }& \highlight{$s^1[c_j] \osucc s^2[c_j] \osucc s^3[c_j] \osucc r_2$} & \highlight{$11$} & \highlight{$1$}\\
   $u_i^2\colon $ & $w_i \osucc c[u_i^2]$ & $00$ & \highlight{$r_2\colon$} & \highlight{$b \osucc [C] \osucc a$} & \highlight{$11$}&& \highlight{$a\colon$} & \highlight{$r_1 \osucc r_2 \osucc r_3$} & \highlight{$11$} & \highlight{$1$}\\
    $v_i^1\colon $ & $p_i \osucc c[v_i^1]$ & $11$ & \highlight{$r_3\colon$} & \highlight{$a \osucc b$} & \highlight{$01$} & & \highlight{$b\colon$} & \highlight{$r_3 \osucc r_2 \osucc r_1$} & \highlight{$11$} & \highlight{$\enn+2$}\\
    $v_i^2\colon $ & $p_i \osucc c[v_i^2]$ & $00$ &      $x_i\colon$ &  $p_i\osucc w_i$ & $10$ & & $w_i\colon$ & $x_i \osucc u_i^1 \osucc y_i \osucc  u_i^2$ & $11$ & $2$ \\
&&&    $y_i\colon$ & $w_i \osucc p_i$ & $01$ && $p_i\colon$ & $y_i \osucc v_i^1 \osucc x_i \osucc  v_i^2$ & $11$ & $2$\\
   \bottomrule
  \end{tabular}}

  All lower quotas are zero.
  This completes the construction of the instance for \diverse.
  One can verify the restrictions stated in the theorem.  
  We show that $(X,\phi(X))$ is satisfiable if and only if the constructed instance admits a feasible and stable matching.

  \noindent  The ``only if''~part: let the truth assignment~$\sigma_X$ satisfy $\phi(X)$.
  It can be verified that the following matching~$M$ is feasible.
  \begin{inparaitem}
    \item For each~$x_i\in X$, if $\sigma_{X}(x_i)=\true$,
    then let $M(w_i)\coloneqq \{x_i,y_i\}$ and $M(p_i)\coloneqq \{v_i^1, v_i^2\}$;
    otherwise let $M(w_i)\coloneqq \{u^1_i,u^2_i\}$ and $M(p_i)\coloneqq\{x_i,y_i\}$.
    \item For each clause~$C_j\in \phi(X)$, let \(M(c_j)\coloneqq \{s^z(c_j)\}\), where \(z \in [3]\) is minimal such that the \(z^\text{th}\) literal in \(C_j\) is set to \(\true\) under \(\sigma_X\); note that there exists at least one such literal since $\sigma_{X}$ is a satisfying assignment.
    \item Let $M(a)\coloneqq \{r_2\}$ and $M(b)\coloneqq \{r_1,r_3\}$.
  \end{inparaitem}

  To show that $M$ is stable, consider a blocking student-college pair~$\{\alpha, \beta\}$ for \(M\) witnessed by \(S' \subseteq M(\beta)\). By Lemma~\ref{lem:aux}\eqref{aux-nec}, we infer that $\alpha$ must lie in $\{x_i,y_i,u_i^1,u_i^2,v_i^1,v_i^2\}$ for some~$i\in [\enn]$. It then suffices to do a case distinction that rules out $\alpha$ being one of the former $2$ students, and also being one of the latter $4$ students.

  For the ``if'' part, let $M$ be a feasible and stable matching for the constructed \diverse instance.
  Define the following assignment~$\sigma_X$ with
  $\sigma_X(x_i)\coloneqq \true$ if there exists a clause college~$c_j$ such that $u_i^1$ or $u_i^2$ is assigned to $c_j$; let $\sigma_X(x_i)\coloneqq \false$ if there exists a clause college~$c_j$ such that $v_i^1$ or $v_i^z$ is assigned to $c_j$.
  If no such clause college exists, then the truth value of $x_i$ can be arbitrary; e.g., let $M(x_i)=\true$. The constructed assignment satisfies all clauses because of \cref{lem:aux}\eqref{aux-card}. Thus, to complete the proof, it suffices to show that $\sigma_X$ is a valid truth assignment. 
\end{proof}
\fi

\noindent We note that the reduction behind \cref{thm:diverse:NP-c:t=2} can be adapted to show
NP-hardness for \feasible, even with three types.

\begin{proposition}\label{prop:feasible-t+umax+qmax-NP-h}
  \feasible{} remains \textnormal{NP}-hard even if $t=3$, $\umax=1$, $\qmax=2$.
\end{proposition}

\iflong
\begin{proof}
  We show how to modify the \diverse instance constructed in the proof of \cref{thm:diverse:NP-c:t=2}:
  We delete the special students~$r_1,r_2,r_3$ and the special colleges~$a$ and $b$ as they are only relevant if stability is required.
  Instead, we introduce one more type, to enforce that each clause college is matched to at least one student.
  All students have this new third type in addition to the ones described in the construction of the proof of \cref{thm:diverse:NP-c:t=2}.
  All clause colleges have lower quotas~$(0,0,1)$ and upper quotas~$(1,1,1)$.
  All variable colleges have the same lower quotas as upper quotas, namely~$(1,1,2)$.
  The capacities and acceptable pairs remain the same as in the previous construction (except for those involving the deleted special students or special colleges).
  
  It is straightforward to verify that the matching constructed from a satisfying truth assignment in the proof of \cref{thm:diverse:NP-c:t=2}, is also feasible for the modified \diverse instance.
  For the reverse direction, if $M$ is a feasible matching in the modified instance, we claim that the truth assignment as in the proof of \cref{thm:diverse:NP-c:t=2} is a valid and satisfying assignment.
  To see this, we first observe that, similarly as in the proof of \cref{thm:diverse:NP-c:t=2},
  each clause college~$c_j$ receives exactly one literal student.
  This means that if $\sigma_{X}$ is a valid truth assignment, then it must satisfy each clause~$C_j$.
  It remains to show that $\sigma_{X}$ is a valid assignment.
  Suppose, for the sake of contradiction, that 
  that $c_j$ and $c_{j'}$ are two clause colleges such that
  $M(c_j)=\{u_i^z\}$ and $M(c_{j'})=\{v_i^{z'}\}$ for some $i \in [\enn]$ and $z,z'\in [2]$.
  We distinguish two cases:
  If $x_i \in M(w_i)$, then \(y_i \in M(w_i)\) as otherwise the lower quota for college~\(w_i\) and type \(2\) cannot be fulfilled without exceeding \(w_i\)'s upper quota for type~\(1\).
  This means \( M(p_i) = \{v^1_i, v^2_i\}\) as otherwise \(p_i\)'s quotas are not fulfilled.
  This is a contradiction to $M(c_{j'})=v_i^{z'}$ with $z'\in [2]$.
  If $x_i\notin M(w_i)$, then $M(w_i)=\{u_i^1,u_i^2\}$ because of the lower and upper quotas for $p_i$ regarding types~$1$ and $2$.
  This is a contradiction to $M(c_{j})=u_i^{z}$ with $z\in [2]$.
\end{proof}
\fi

\ifshort
As a final remark on the impact of diversity, we note that if there are only few types~$t$ or the maximum capacity~$\qmax$ is a constant, then \diverseties{} is in NP. The reason for this is that the size of a witness set is upper-bounded by $\min\{t,\qmax\}$. 
\fi
\iflong Note that all our hardness proofs from \cref{thm:diverse:lq>0:sigma2p-c} and cases where \(\lmax = 0\) are tight in the sense that we show completeness in \(\Sigma_2^P\) and NP respectively for these cases (\cref{thm:diverse:lq>0:sigma2p-c} and \cref{prop:lq=0:inNP} respectively).
Similarly containment in NP is known for \feasible.
Thus these are the strongest hardness results one can achieve in terms of levels of the polynomial hierarchy.
We are also able to show NP-completeness for the NP-hard fragments of \diverse, which are not \(\Sigma_2^P\)-hard and for which it does not necessarily hold that \(\lmax = 0\).
Specifically for \diverse where the number of types, the maximum upper quota and/or the maximum capacity are bounded, the following observation closes the gap between hardness and completeness.
\fi

\begin{observation}\label{obs:t-or-qmax-NP}
  If $t$ or $\qmax$ is a constant, then \diverseties is in \textnormal{NP}.
\end{observation}

\iflong
\begin{proof}
  Notice that the size of a minimum witness for a blocking pair never exceeds \(t\) or \(\qmax\):
  \begin{itemize}
    \item each student has at most \(t\) types and hence at most \(t\) students have to be removed to make space in terms of upper quotas; and
    \item a witnessing set is a subset matched to a single college, which is by definition never larger than \(\qmax\).
  \end{itemize}
  Hence, given a matching $M$, one can check in polynomial time whether it admits no blocking pair---one can enumerate all unmatched pairs and all potential witnesses and check if they describe a blocking pair.
  This suffices to show NP-containment.
\end{proof}
\fi

\ifshort\paragraph{The Case with Few Colleges.}\else
\subsection{The Case with Few Colleges}\fi
The NP-hardness reduction behind \cref{thm:diverse:NP-c:t=2} produces a college gadget for each variable in order to maintain as few types as possible.
This leads to the question of whether the problem remains NP-hard for few colleges.
The following theorem answers the question affirmatively.
The idea is to reduce from the \textsc{Independent Set} problem~\cite{GJ79} and introduce types corresponding to the vertices and the edges in an input graph,
and students corresponding to the vertices such that the students assigned to a special college~$w$ must correspond to an independent set. As before, we use \cref{lem:aux} to enforce that $w$ receives at least some given number of students.

\begin{theorem}\label{thm:lq=0:m=4:diverse-NP-c}
  \diverse is \textnormal{NP}-hard even if $m=4$, $\lmax=0$ and $\umax=2$.
\end{theorem}

\iflong
\begin{proof}
	We provide a polynomial-time reduction from the following NP-complete \textsc{Independent Set} problem~\cite{GJ79}.
	\decprob{Independent Set}
	{An undirected graph~$G=(V,E)$ with $V$ being the vertex set and $E$ being the edge set,
		and an integer~$k\ge 0$.}
	{Does $G$ admit a size-at-least-$k$ \emph{independent set}~$V'\subseteq V$ (i.e., for each edge~$e\in E$ it holds that $|V'\cap e|\le 1$)?}
	
	Let $(G=(V,E), k)$ be an instance of \textsc{Independent Set}.
	Let $V\coloneqq \{v_1,\ldots, v_{n'}\}$ and $E\coloneqq \{e_1,\ldots, e_{m'}\}$.
	We construct the following instance for \diverse{}.
	
	\paragraph{The types.}
	We consider $2n'+m'+2$ types: For each vertex~$v_i\in V$, there are two corresponding \myemph{vertex types}.
	For each edge~$e_j$, there is a corresponding \myemph{edge type}.
	Besides this, there are \emph{two} special types; which will be the ones used to incorporate the gadget from \cref{lem:aux}.
	
        \allowdisplaybreaks
	\paragraph{The students.}
	We consider $4n'+3$ students:
	For each vertex~$v_i \in V$, there are two \myemph{vertex students}, called~$v_i,u_i$,
	and two corresponding \myemph{dummy students}, called $x_i,y_i$.
	The types of these students are as follows: 
	\begin{align*}
	\forall j \in [2n']\colon &
	\typevec_{v_i}[j]\coloneqq 
	\begin{cases}
	1\text{,} & \text{ if } j=2i-1,\\
	0\text{,} & \text{ otherwise,}
	\end{cases}\\
	&        \typevec_{u_i}[j]\coloneqq  
	\begin{cases}
	1\text{,} & \text{ if } j=2i,\\
	0\text{,} & \text{ otherwise,}
	\end{cases}\\
	&   \typevec_{x_i}[j]\coloneqq  
	\typevec_{y_i}[j]\coloneqq  
	\begin{cases}
	1\text{,} & \text{ if } j \in \{2i-1,2i\},\\
	0\text{,} & \text{ otherwise.}
	\end{cases}\\
	\forall j \in [m+2] \colon & 
	\typevec_{v_i}[2n+j]\coloneqq 
	\begin{cases}
	1\text{,} & \text{ if } e_j \in E\text{ and } v_i \in e_j,\\
	0\text{,} & \text{ otherwise.}
	\end{cases}\\
	&   \typevec_{u_i}[2n'+j] \coloneqq \typevec_{x_i}[2n'+j]\coloneqq \typevec_{y_i}[2n'+j] \coloneqq 0.
	\end{align*}
	For an illustration, assume that $G$ contains four vertices and five edges, with vertex~$v_1$ being incident to edge~$e_2$ and $e_4$.
	Then the types of the four students~$v_1$, $u_1$, $x_1$, and $y_1$ corresponding to vertex~$v_1$ are:
	\begin{align*}
	\typevec_{v_1} \coloneqq & 10\,00\,00\,00\,01010\,00,\\
	\typevec_{u_1} \coloneqq & 01\,00\,00\,00\,00000\,00,\\
	\typevec_{x_1} \coloneqq & 11\,00\,00\,00\,00000\,00,\\
	\typevec_{y_1} \coloneqq & 11\,00\,00\,00\,00000\,00.
	\end{align*}
	
	There are three \myemph{special students}, \myemph{$r_1,r_2,r_3$}.
	\begin{align*}
	\typevec_{r_1} \coloneqq \{0\}^{2n'+m'}10,~ \typevec_{r_2}\coloneqq \{0\}^{2n'+m'}11,~ \typevec_{r_3}\coloneqq \{0\}^{2n'+m'}01.
	\end{align*}
	
	\paragraph{The colleges.}
	We consider four colleges~$w, p,  a, b$.
	The capacities and the upper quotas of the colleges are as follows while the lower quotas are set to zero.
	\begin{itemize}
		\item $q_{w}\coloneqq n'+k$, $q_p\coloneqq 2n'-k$, $q_a\coloneqq 1$, and $q_b\coloneqq 2$.
		
		\item $\uppervec_w \coloneqq \{1\}^{2n'+m'+2}$, $\uppervec_p\coloneqq \{1\}^{2n'}\{2\}^{m'}00$,
		$\uppervec_a\coloneqq \{0\}^{2n'+m'}11$, and $\uppervec_b\coloneqq \{0\}^{2n'+m'}11$.
	\end{itemize}
	
	\paragraph{The preferences of the students and colleges.}
	The preference lists are defined as follows.
        
        \begin{tabular}{l@{}l|l@{}ll@{\;\;}l@{}}
          \toprule
          S. & Pref. & C. & Pref. \\
          \midrule
	\rowcolor{lightgray}$r_1\colon$& $b \osucc a$ &   $a\colon$ & $r_1 \osucc r_2 \osucc r_3$ \\ 
	\rowcolor{lightgray}$r_2 \colon$& $b \osucc w \osucc a$ &   $b\colon$& $r_3 \osucc r_2 \osucc r_1$ \\ 
	\highlight{$r_3 \colon$}& \highlight{$a \osucc b$}& \\
	$v_i\colon$&  $w \osucc p$, & \highlight{$w\colon$} & \highlight{$u_1\osucc x_1\osucc v_1\osucc \cdots \osucc u_n\osucc x_n\osucc v_n \osucc r_2$}  \\ 
	$u_i\colon$&  $p \osucc w$, & $p\colon$ & $v_1\osucc y_1\osucc u_1\osucc \cdots \osucc v_n\osucc y_n\osucc u_n$ \\ 
	$x_i\colon$ & $w$, &  \\
	$y_i\colon$ & $p$,& \\\bottomrule
	\end{tabular}

	This completes the construction, which clearly can be carried out in polynomial time.
	It is straightforward to verify that $\lmax=0$, $\umax=2$, and $m=4$.
	
	We claim that graph~$G=(V,E)$ admits an independent set of size at least~$k$ if and only if the constructed \diverse instance admits a feasible and stable matching.
	
	For the ``only if'' part, assume that $(G=(V,E), k)$ is a yes instance and let $V'\subseteq V$ be an independent set of size exactly~$k$.
	We show that the following matching~$M$ with
	\begin{align*}
	M(w)\coloneqq & \{v_i, u_i \mid v_i \in V'\}\cup \{x_i \mid v_i \in V\setminus V'\},\\
	M(p)\coloneqq & \{v_i, u_i \mid v_i \in V\setminus V'\}\cup \{y_i \mid v_i \in V'\},\\
	M(f)\coloneqq & \{y\}, \text{ and } M(g)\coloneqq \{x,z\}
	\end{align*}
	is feasible and stable.
	It is straightforward to check that $M$ is feasible for colleges~$a$ and $b$.
	$M$ is also feasible for college~$p$ since each edge type is ``incident'' to exactly two vertex students from~$V$.
	Similarly, $M$ is also feasible for college~$w$ since $V'$ is an independent set, meaning that for each edge~$e_j$, at most one vertex student from $V'$ has edge type~$2n'+j$.
	It remains to show that $M$ is stable.
	Assume for contradiction, that $M$ is not stable, i.e., there is a blocking pair~$\{\alpha, \beta\}$ for~\(M\).
	Let $S' \subseteq M(\beta)$ witness this.
	First, we observe that regarding the types~$2n'+m'+1$ and $2n'+m'+2$, the students~$\{u_i,x_i,v_i\mid i\in [n']\}\uplus \{r_1,r_2,r_3\}$
	and the colleges~$\{w\}\uplus\{a,b\}$ correspond exactly to the students and colleges discussed in \cref{lem:aux}.
	By \cref{lem:aux}\eqref{aux-nec}, it follows that $\alpha\notin \{r_1,r_2,r_3\}$ and $\beta\notin \{a,b\}$.
	Hence, we infer that $\alpha\in \{v_i,u_i,x_i,y_i\mid i\in [n']\}$ and $\beta\in \{w,p\}$. We distinguish four cases:
	\begin{itemize}
		\item If $\alpha = x_i$ for some~$i \in [n']$, then $\beta=w$.
		By the construction of $M$ it follows that $\{v_i,u_i\}\subseteq M(w)$.
		Since $\typevec_{x_i}[2i-1]=\typevec_{x_i}[2i]=1$ and since besides~$v_i$ and~$u_i$, \emph{no} student assigned to $w$ has type~$2i-1$ or $2i$,
		it follows that $\{v_i,u_i\}\subseteq S'$.
		However, $w$ prefers $u_i$ to $x_i$, a contradiction to $\{x_i,w\}$ being a blocking pair.
		\item Similarly, we obtain a contradiction for the case when $\alpha=y_i$ with~$i \in [n']$.
		\item If $\alpha=v_i$ for some $i\in [n']$, then $\beta = w$.
		By the construction of $M$, it follows that $x_i \in M(w)$ and $u_i\notin M(w)$.
		Since $\typevec_{v_i}[2i-1]=1$ and since besides~$x_i$, no student assigned to~$w$ has type~$2i-1$,
		it follows that $x_i \in S'$.
		However, $w$ prefers $x_i$ to $v_i$, a contradiction to $\{v_i,w\}$ being a blocking pair.
		\item Similarly, we obtain a contradiction for the case when $\alpha=u_i$ with~$i\in [n']$.
	\end{itemize}
	As every case leads to a contradiction, there is no blocking pair \(\{\alpha, \beta\}\).
	Thus, $M$ is indeed stable.
	
	Now, we turn the ``if'' part and assume that the constructed \diverse instance admits a feasible and stable matching~$M$.
	Let $V'\coloneqq \{v_i \mid i \in [n'] \wedge v_i \in M(w)\}$, $U'\coloneqq \{u_i \mid i \in [n']\wedge u_i \in M(w)\}$, $X'\coloneqq \{x_i \mid i\in [n'] \wedge x_i \in M(w)\}$.
	We aim to show that $V'$ is an independent set of size at least~$k$.
	Clearly, $V'$ is an independent set because of $w$'s upper quotas for the edge types.
	It remains to show that $V'$ has size at least~$k$.  
	Before we show this, we first observe the following for the sizes of $V'$, $U'$, and $X'$.
	\begin{claim}\label{claim:IS-M}
		\begin{enumerate}[(1)]
			\item\label{size-bound} It holds that $|V'|+|U'|+|X'| = n+k$.
			\item\label{U'V'} It holds that $|V'| \ge k$.
		\end{enumerate}
	\end{claim}
	\begin{proof}
		\renewcommand{\qedsymbol}{(of
			\cref{claim:IS-M})~$\diamond$}
		Statement~\eqref{size-bound} follows directly from \cref{lem:aux}\eqref{aux-card} applied to
		\(\{u_i,x_i,v_i\mid i\in [n']\}\uplus \{r_1,r_2,r_3\}\)
		and $\{w\}\cup\{a,b\}$.
		
		Statement~\eqref{U'V'}: First, observe that for each $i\in [n']$ it must hold that $x_i \in X'$ if and only if $u_i \notin U'$ because of the upper quota of $w$ regarding type~$2i$.
		This implies that $|X'|+|U'|\le n'$. By Statement~\eqref{claim:IS-M} it follows that $|V'|\ge k$.
	\end{proof}
	Now, it is straightforward to see that $|V'|=k$ follows from \cref{claim:IS-M}\eqref{U'V'}.
\end{proof}
\fi

\section{Algorithmic Results}\label{sec:algos}
This section provides the algorithmic results that together allow us to complete Table~\ref{tab:overview}. \ifshort The first result deals with the case where the number of students is bounded by a constant. \fi

\iflong
\subsection{\diverseties\ with few students}
\fi
%
\iflong
In this section, we deal with the case where there are few number of students.
For simplicity of notation, given a type~$z\in T$, we let \myemph{$U_z^{\typevec}\coloneqq \{u_i \in U \mid \typevec_{u_i}[z]=1\}$} be the set of all students who have type~$z$.
We omit the superscript~$\typevec$ if the type vectors~$\typevec$ are clear from the context.
\fi

\ifshort
\begin{theorem} \label{thm:fptn}
	\diverseties{} is solvable in \(\mathcal{O}(n^{2n + 5} \cdot 2^{2n})\).
\end{theorem}
\begin{proof}[Proof Sketch]
  We show how to preprocess an \diverseties instance~\(I=(U, W, T, (\succeq_u, \typevec_u)_{u\in U}, (q_w, \lowervec_w,\uppervec_w)_{w\in W}\) to obtain an instance \((U, W', T', (\succeq'_u\), \(\typevec'_u)_{u\in U}, (q'_w,\lowervec'_w,\uppervec'_w)_{w\in W})\) with \(n\) students, \(n^2 + n\) colleges and \(2^n\) types which is equivalent in terms of the existence of a feasible and stable matching. This then suffices to solve the instance in the claimed running time via an exhaustive brute-force procedure.
	
	\noindent \textbf{\(2^n\) types.}
	Observe that types \(z,z' \in T\) which describe the same subset of students, i.e., \(\{u\)\(\in\)\(U\)\(\mid \typevec_u[z]\)\(=\)\(1\} = \{u \in U \mid \typevec_u[z']\)\(=\)\(1\}\), can be merged into a single type \(\zeta\). 
	For each feasible matching of \(I\), the students assigned to a college \(w\)\(\in\)\(W\) of two types \(z, z' \in T\) merged in this way always adhere to the stricter of the upper and lower quotas of the merged types, i.e., \(\max\{\lowervec_w[z], \lowervec_w[z']\}\) and \(\min\{\uppervec_w[z], \uppervec_w[z']\}\).
\iflong	Thus setting \(\lowervec_w(\zeta) \coloneqq \max\{\lowervec_w[z], \lowervec_w[z']\}\) and \(\uppervec_w(\zeta) \coloneqq \min\{\uppervec_w[z], \uppervec_w[z'])\}\) maintains equivalence of the original types and their quotas, and the merged types and their quotas in terms of feasible and stable matchings.
\fi
	After exhaustive merging we are left with modified types \(T'\) with~\(|T'|\)\(\leq\)\(2^n\).
	
	\noindent \textbf{\(n^2 + n\) colleges.}
First, note that an instance with more than \(n\) colleges with non-zero lower quotas can be immediately rejected.
	To upper-bound the number of colleges with zero lower-quotas (hereinafter denoted $W_0$), note that in a stable and feasible matching $M$ every student (say, $u$) matched to a college from $W_0$ may only be matched to one of her\todoH{Student == she, college == it}\ $n$ most preferred colleges in $W_0$ which has enough upper-quotas to accommodate her. Otherwise there would exist an empty zero lower-quota college in $W_0$ which $u$ prefers to $M(u)$, forming a blocking pair.

We employ this observation by defining the following marking procedure. Let us begin by setting $W'=\emptyset$. Now, for each student $u\in U$, we mark the $n$ most preferred colleges in $W_0\setminus W'$, resolving ties arbitrarily. Clearly, at the end we obtain a set $W'$ of size at most $n^2$, and it is easy to show via a replacement argument and the above observation that the colleges in $W_0\setminus W'$ may be deleted without changing the existence of a stable and feasible matching.
\end{proof}
\fi

\iflong
\begin{theorem} \label{thm:fptn}
  \diverseties{} admits a problem kernel with $n$~students,~$n^2+n$~colleges, and~$2^n$~types,
  \todoH{ I think the running time should be the following.}
  and can be solved in  $\mathcal{O}(n\cdot m \cdot t + 2^n\cdot (2n+1)^n\cdot n^2\cdot t)$~time.
\end{theorem}
\begin{proof}
  We first show how to preprocess an \diverseties instance~\(I=(U, W, T, (\succeq_u, \typevec_u)_{u\in U}, (q_w, \lowervec_w,\uppervec_w)_{w\in W}\) to obtain an instance \((U, W', T', (\succeq'_u\), \(\typevec'_u)_{u\in U}, (q'_w,\lowervec'_w,\uppervec'_w)_{w\in W})\) with \(n\) students, \(n^2 + n\) colleges and \(2^n\) types which is equivalent in terms of the existence of a feasible and stable matching. This then suffices to solve the instance in the claimed running time via an exhaustive brute-force procedure.
	


  \paragraph{Upper-bounding the number of types.}
  We first show that the number of types can be reduced to at most~$2^n$.
  To achieve this, we observe that if there are two types~$z, z'\in T$ which are possessed by exactly the same subset of students, i.e., $U^{\typevec}_z = U^{\typevec}_{z'}$,
  then each student~$u$ has type~$z$ if and only if she has type~$z'$.
  This implies that the lower-quota (resp.\ the upper-quota) on~$z$ must also be a lower-quota (resp.\ an upper-quota) on~$z'$ and vice versa.

  By the above observation, we can group types~$z$ together which belong to exactly the same subset~$U_z$ and update the lower-quotas and upper-quotas for each college accordingly.
  Formally, let $\mathcal{F}\coloneqq \{U^{\typevec}_z \mid z \in T\}$ denote the family of subsets of students defined according to the original types and let $f\colon |\mathcal{F}| \to \mathcal{F}$ be an arbitrary but fixed bijection.
  The ``reduced'' set of types is defined as~$T^*\coloneqq |\mathcal{F}|$.
  The new type vector~$\typevec^*_{u}$ of each student~$u\in U$ is defined as follows:
  \begin{align*}
    \forall z\in T^* \colon   \typevec^*_{u}[z] \coloneqq 1 \text{ if and only if } u \in f(z).
  \end{align*}
  Clearly, for each $z\in T^*$ it holds that $U^{\typevec}_{z}=U^{\typevec^*}_z$.
  
  The new lower-quotas and upper-quotas of each college~$w \in W$ are defined as follows:
  \begin{align*}
    \forall z\in T^*\colon & \lowervec^*_{w}[z] \coloneqq \max_{z'\in T\text{ with } U^{\typevec}_{z'}=f(z)}\lowervec_{w}[z'] \text{~ and ~}\\
    \uppervec^*_{w}[z] \coloneqq & \min_{z'\in T\text{ with } U^{\typevec}_{z'}=f(z)}\lowervec_{w}[z']\text{.}
  \end{align*}
  Now, we show that $I$ and the reduced instance~$I^*=(U, W, (\succeq_{x})_{x\in U\cup W}, (\typevec^*_{u_i})_{u_i\in U}, (q_{w_j}, \lowervec^*_{w_j}, \uppervec^*_{w_j})_{w_j\in W})$ are equivalent by showing that
  every matching is feasible and stable for $I$ if and only if it is also feasible and stable for $I^*$.

  For the ``only if'' part, let $M$ be a feasible and stable matching of $I$.
  Suppose, for the sake of contradiction, that $M$ is not feasible or not stable for~$I^*$.
  If $M$ is not feasible for $I^*$, then since the capacity bounds are not changed, there must be a college~$w \in W$ and a type~$z\in T^*$ such that
  either $|M(w)\cap U^{\typevec^*}_z| < \lowervec^*_{w}[z]$ or $|M(w)\cap U^{\typevec^*}_z| > \uppervec^*_{w}[z]$.
  We only consider the case when $|M(w)\cap U^{\typevec^*}_z|< \lowervec^*_{w}[z]$; the other case where the upper-quota is violated is analogous.
  Now, if $|M(w)\cap U^{\typevec^*}_z|< \lowervec^*_{w}[z]$,
  then since $U^{\typevec}_z=U^{\typevec^*}_z$, by the definition of $\lowervec^*_{w}[z]$, 
  there is an original type~$z'\in T$ with $U^{\typevec}_{z'} = U^{\typevec^*}_z$
  such that $\lowervec^*_{w}[z] = \lowervec_{w}[z]$.
  Together, we derive that  $|M(w)\cap U^{\typevec}_{z'}| = |M(w)\cap U^{\typevec^*}_{z}| < \lowervec^*_{w}[z]  = \lowervec_{w}[z]$, a contradiction to $M$ being feasible for $I$.

  Now, suppose, for the sake of contradiction, that $M$ is not stable for~$I^*$.
  This means that there exists an unmatched pair~$\{u, w\} \notin M$ with $u \in U$ and $w\in W$ such that
  \begin{enumerate}[(1)]
    \item $u$ prefers~$w$ to~$M(u)$ and $w$ prefers~$u$ to~$M(w)$, and
    \item there exists a subset~$U'\subseteq M(w)$ of students assigned to~$w$, where~$w$ prefers~$u$ to~$U'$, and $M\cup \{\{u, w\}\}\setminus (\{\{u, M(u)\} \cup \{u', w\}_{u'\in U'}\})$ remains feasible for~$w$, regarding~$\typevec^*$,~$\lowervec^*_{w}$, and~$\uppervec^*_{w}$.
  \end{enumerate}
  Since $M$ is stable for~$I$, it must hold that $M\cup \{\{u, w\}\}\setminus (\{\{u, M(u)\} \cup \{u', w\}_{u'\in U'}\})$ is not feasible for $I$ as otherwise~$\{u, w\}$ is also blocking~$M$ in $I$.
  By a similar reasoning as above, we can show that $M\cup \{\{u, w\}\}\setminus (\{\{u, M(u)\}\} \cup \{u', w\}_{u'\in U'}\})$ is also \emph{not} feasible for $I'$, a contradiction.

  \paragraph{Upper-bounding the length of the preference list of each student.}
  We have just shown that $I^*$ is equivalent to $I$ and has at most~$2^n$~types.
  Next, we show that the number of colleges in each student's preference list is upper-bounded by a function in~$(n,t)$. 
  We consider colleges with at least one non-zero lower-quota and with all-zero lower-quotas separately.
  Let $W_{>0}$ be the set of all colleges with at least one non-zero lower-quota for some type.
  Observe that each college in $W_{>0}$ must be assigned at least one student.
  Then, for any `yes'-instance it must hold that $|W_{>0}|\le n$ as otherwise no feasible matching exists because at least one college from $W_{>0}$ is not assigned any student.
  Thus, we obtain the following.
  \begin{claim}\label{claim:non-zero-lower-quotas-colleges}
    If $I$ is a yes instance, then $|W_{>0}|\le n$.
  \end{claim}
  We also observe that no feasible matching will assign a student~$u$ to a college~$w$ which cannot accommodate~$u$:
  \begin{claim}\label{claim:zero-lower-quotas-colleges}
    For each student~$u$ and each acceptable college~$w\in \acset(u)$ if $\uppervec_w \not\ge \typevec_u$, then no feasible matching can assign~$w$ to~$u$.
  \end{claim}
  
  Next, we turn to the set of colleges with all-zero lower-quotas, denoted as~$W_{=0}$.
  Consider an arbitrary student~$u\in U$.
  Although the length of~$u$'s preference list~$\succeq_{u}$ can be unbounded, we show in the following that we only need to consider at most $n$~colleges from $W_{=0}$.
  We introduce one more notation.
  Let $W_{u} \coloneqq \{w \in W_{=0} \cap \acset(u) \mid \uppervec_{w}\ge \typevec_{u}\}$ denote a set consisting of all acceptable colleges from $W_{=0}$ whose upper-quotas are large enough to accommodate student~$u$ alone.
  Now, observe that if there exists a stable matching~$M$ where some college~$w \in W_{u}$ has $M(w)=\emptyset$, then~$M$ must assign to~$u$ some college~$w'$ which is either preferred to or tied with~$w$ by student~$u$ as otherwise~$u$ and~$w$ will form a blocking pair.
  In other words, we only need to go through the colleges from $W_{u}$ in order~$\succ^*_u$ of the preference list~$\succeq_u$ of $u$ (we break ties arbitrarily) and select the first $n$~ones.
  Formally, let
  \begin{align*}
    W'_{u}\coloneqq \{w \in W_{u} \mid |\{w' \in W_{u} \mid w' \succ^*_{u} w\}| < n\}.
  \end{align*}
  We claim that it suffices to consider the colleges from $W'_{{u}}$ by showing the following.
  \begin{claim}\label{claim:W'}
    If there exists a feasible and stable matching which assigns to student~$u$  a college from~$W_{u}\setminus W'_{u}$,
    then there exists a feasible and stable matching which assigns to student~$u$ a college from~$W'_{u}$.
  \end{claim}
  \begin{proof}[Proof of \cref{claim:W'}] \renewcommand{\qedsymbol}{$\diamond$} 
    Let $M$ be a feasible and stable matching.
    Assume that there exists a college~$w \in W_{u}\setminus W'_{u}$ with $M(u)=w$.
    This implies that $|W'_u| = n$ and there exists a college~$w'\in W'_{u}$ which does not receive any student.
    Since $M$ is stable, it must hold that $w$ and $w'$ are tied by~$u$ as otherwise $u$ will form a blocking pair with $w'$, because $M(w')=\emptyset$ and $w'$ has enough upper-quotas to accommodate student~$u$ alone, a contradiction to $M$ being stable.
    Then, it is straightforward to check that the modified matching~$M'$
    with $M'\coloneqq M\setminus \{\{u, w\}\} \cup \{\{u, w'\}\}$ is a feasible and stable matching; note that $M'$ is also feasible for $w$ as $w$ has all-zero lower-quotas.
  \end{proof}
  By the above observations, we can restrict the preference list of each student~$u$ to only the colleges from $\big(\acset(u_i)\cap W_{=0}\big)\cup W'_{u}$; note that we only ``ignore'' colleges with all-zero lower-quotas and colleges which cannot ``accommodate'' student~$u$.
  We update the preference lists of the colleges accordingly.
  If some college obtains empty preference list after the update, then we can simply delete it as it will never be assigned a student by any feasible and stable matching.
  Summarizing, we obtain a new equivalent instance with $n$~students, at most~$2^n$ types, and at most $n^2+n$~colleges.

  \paragraph{The running time for the preprocessing.}
  It remains to analyze the running time.
  The modification of the types and the type vectors of all students can be done in~$\mathcal{O}(n\cdot t)$~time.
  The modification of the lower-quotas and upper-quotas of all colleges can be done in $\mathcal{O}(m\cdot t)$~time.
  The modification of the preference list of each student can be done in linear time by using an integer counter.
  For each student~$u$, we go through the colleges from~$\acset(u)$ in order of the preferences of $u$ (we break ties arbitrarily).
  Let $w$ be the college currently considered.
  If $\uppervec_w \not\ge \typevec_{u} $, then we delete the acceptable pair~$\{u,w\}$.
  Here, deleting a pair~$\{x,y\}$ means deleting $x$ from the preference list of $y$ and~$y$ from the preference list of $x$.
  Otherwise, if there is no integer counter~$c_{u,\succeq_w}$ for $\succeq_{w}$, then keep college~$w$ in the list and create a counter~$c_{u}$ and set it to one.
  Otherwise, if $c_{u}<n$, then also keep college~$w$ and increment counter~$c_{u,\succeq_w}$ by one; otherwise, delete the pair~$\{u,w\}$.
  In this way, we can modify the preference lists of the students in~$\mathcal{O}(n\cdot m\cdot t)$~time.%

  \paragraph{Solving the original instance.}
  First, We perform the described preprocessing in $\mathcal{O}(n\cdot m\cdot t)$~time.
  Then, we branch on the at most $(2n+1)^n$ possible ways to assign colleges to the students for the instance obtained from preprocessing; note that after preprocessing, each student~$u$ has at most $2n$ colleges ($n$ from $W_{>0}$ and $n$ from $W'_u$) in her preference list.
  For each of the $O((2n+1)^n)$ possible matchings, we check feasibility and stability.
  Since each matching contains at most $n$~pairs, feasibility can be checked in $O(n\cdot t)$~time.
  Recall that each student has at most $2n$ relevant colleges in her preference lists.
  Thus, to check stability, we go through all~$O(2n^2)$ unmatched pairs, say~$\{u,w\}$, and all possible subsets of students~$2^n$ assigned to~$w$ and all types~$z\in [t]$. 
  This can be done in $\mathcal{O}(2n^2\cdot 2^n\cdot t)$~time.
  The overall running time is $\mathcal{O}(n\cdot m \cdot t + 2^n\cdot (2n+1)^n\cdot n^2\cdot t)$.
\end{proof}

\fi

\iflong
\subsection{\diverseties with few colleges and small maximum capacity}
\fi
Next, we show that \diverseties{} can be solved in polynomial time if the number~$m$ of colleges and the maximum capacity~$\qmax$ of all colleges are constants, using a simple brute-forcing algorithm based on the following observation.


\begin{observation}\label{obs:m+qmax}
 Every feasible matching can assign colleges to at most $m\cdot \qmax$ students.
\end{observation}

By the above observation, we only need to guess a subset of at most $m\cdot \qmax$ students which are assigned to colleges, and branch for each student in the guessed set on the choice of one out of $m$ possible colleges.
For each branch, we check feasibility and stability in $\mathcal{O}(2^{\qmax}\cdot n \cdot m \cdot t)$~time since each college obtains at most~$\qmax$ students (see Observation~\ref{obs:t-or-qmax-NP}).

\begin{proposition}\label{prop:diverseties-m+qmax:XP}
  \diverseties{} and \feasible{} can be solved in $\mathcal{O}(n^{m\cdot \qmax}\cdot (m \cdot \qmax)^m \cdot 2^{\qmax}\cdot n\cdot m\cdot t)$~time.
\end{proposition}
\iflong
\begin{proof}
  We only consider \diverseties{} as the algorithm for \feasible works the same except in the checking phase we only need to check feasibility instead of both feasibility and stability. 
  Let $I=(U,W,T, (\succeq_u, \typevec_u)_{u\in U}$, $(q_w,\lowervec_w,\uppervec_w)_{w\in W})$ be an instance of \diverseties{}.
  The algorithm behind \diverseties{} works as follows.
  For each subset~$U'\subseteq U$ of at most $m\cdot \qmax$ and for each possible matching~$M$
  where each \(u \in U'\) is assigned to some \(w \in W\), and \(M(u)=\bot\) for all \(u \in U \setminus U'\),
  we check whether the feasibility conditions and the stability condition holds:
  \begin{compactenum}[(1)]
    \item\label{cond:qmax} $|M(w)|\le \qmax$ for all $w\in W$,
    \item\label{cond:diverse} $\lowervec_w\le \sum_{u\in M(w)}\typevec_{u}\le \uppervec_w$ for all~$w\in W$, and
    \item\label{cond:stable} for each unmatched student-college pair~$\{u,w\}$ where $u \notin M(w)$ and $u$ prefers $w$ to $M(u)$
    and for each (possibly empty) subset~$S\subseteq M(w)$ such that $w$ to prefers~$u$ to every student in~$S$ there exists a type~$z\in T$ with $\typevec_{u}[z]+\sum_{u' \in M(w)\setminus S}\typevec_{u'}[z]<\lowervec_w[z]$
    or with $\typevec_{u}[z]+\sum_{u' \in M(w)\setminus S}\typevec_{u'}[z] > \uppervec_w[z]$.
  \end{compactenum}
  We return ``yes'' by accepting $M$ as a feasible and stable matching for $I$ as soon as we found a matching fulfilling the above conditions, whereas we return ``no'' if no such matching is found.
  To see the correctness, clearly, if our algorithm returns ``yes'' by accepting a specific matching~$M$,
  then $M$ stable and feasible.
  Now, if $I$ is a yes instance and admits a feasible and stable matching~$M$, then let $U'$ denote the subset of students which are assigned colleges under~$M$, i.e., $U'\coloneqq \{u\in U\mid M(u)\in W\}$.
  By \cref{obs:m+qmax}, $|U'|\le m\cdot \qmax$.
  Thus, our algorithm must have considered the subset~$U'$ and the matching~$M$.
  Since $M$ is feasible and stable it satisfies Conditions~\eqref{cond:qmax}--\eqref{cond:stable}.
  Thus, when our algorithm returns ``yes'', latest by accepting~$M$.

  It remains to analyze the running time.
  First of all, since $|U'|\le m\cdot \qmax$, there are $O(n^{m\cdot \qmax}\cdot (m\cdot \qmax)^m)$ matchings to test for Conditions~\eqref{cond:qmax}--\eqref{cond:stable}.
  Next, for each considered matching~$M$, testing the feasibility Conditions~\eqref{cond:qmax}--\eqref{cond:diverse} can obviously be done in $\mathcal{O}(m\cdot n\cdot t)$~time.
  Finally, testing the stability conditions can be done in $\mathcal{O}(2^{\qmax}\cdot m\cdot n\cdot t)$~time since there are at most~$n\cdot m$ unmatched student-college pairs and each college in the pair has at most $\qmax$~students. 
  
  Note that for \feasible{}, we only need to check Conditions~\eqref{cond:qmax}--\eqref{cond:stable}.
\end{proof}
\fi

\iflong
\subsection{\diverseties with few colleges and types}
\fi

\ifshort
Finally, we turn our attention to instances with a small number of colleges and types, and show that in this case \diverseties\ also admits a polynomial-time algorithm. We note that while under such restrictions one can use the \emph{bounded-variable ILP Encoding} technique~\cite{BredereckKN17,FellowsLMRS08} to show that \feasible\ becomes polynomial-time solvable, the same technique is unlikely to work for \diverseties. That is because two students, even with the same type vectors and the same preference lists,
may be preferred differently by a college.
\fi

\iflong
In this section, we show that for few number of colleges and types, both \feasible and \diverseties can be solved in polynomial time.
For \feasible, we observe that among all students with the same type vector and the same acceptable set of colleges, it does not matter which student of them is matched to a college as long as the lower quotas and upper quotas are fulfilled.
In other words, all students with the same type vector and acceptable set can be grouped together and treated as the same.
This allows us to express \feasible as an ILP of small size.
\begin{lemma}
	\label{lemma:ILP}
	\feasible can be expressed as an integer linear program
	with \(\mathcal{O}(m\cdot 2^{m}\cdot 2^t)\)~variables
	which take values of at most \(n\), and \(\mathcal{O}(n\cdot m+t\cdot m)\) inequalities.
\end{lemma}
\ifshort
\begin{proof}[Proof Sketch]
	Introduce the following notation.
	For each type vector~$\typevec \in \{0,1\}^t$ and each subset~$A\subseteq W$ such there there exists a student~$u$ with
	$\typevec_u=\typevec$ and $\acset(u)=A$, let $S_{\typevec, A}$ denote the set of students with type vector~$\typevec$ and acceptable set~$A$:
	\begin{align*}
	S_{\typevec,A} \coloneqq \{u\in U\mid \typevec_u = \typevec \wedge \acset(u)=A\}.
	\end{align*}
	Note that the sets~$S_{\typevec,A}$ partition the whole set~$U$ of students and for each $w\in W$ it holds that either $S_{\typevec,A} \subseteq \acset(w)$ or $S_{\typevec,A} \cap \acset(w)=\emptyset$ because of the symmetry of the acceptability.
	Further, there are at most $\min(n, 2^t\cdot 2^m)$ such sets,
	and they can be pre-computed in polynomial time.
	Let $\mathcal{S}$ be the family consisting of the computed sets~$S_{\typevec, A}$.  
	
	For each college~$w$ and each non-empty set~$S_{\typevec,A}$ with $A\subseteq \acset(w)$, let $x_{w,\typevec,A}$ be an integer variable.
	The value of $x_{w,\typevec,A}$ will encode the number of students with type vector~$\typevec$ and acceptable set~$A$ that shall be assigned to~$w$ in a feasible matching.
	The ILP is given by the following constraints:
	\begin{alignat}{3}
	x_{w,\typevec,A}  & \ge 0, &~~ & 
	\forall w \in W,  \forall S_{\typevec, A} \colon
	S_{\typevec,A} \subseteq \acset(w)
	, \label{ilp:integer} \\
	\sum_{w\in W}x_{w,\typevec,A} & \le |S_{\typevec, A}|, && \forall S_{\typevec, A} \colon
	S_{\typevec,A}\subseteq \acset(w),\label{ilp:students}\\
	\sum_{\mathclap{
			S_{\typevec,A}\subseteq \acset(w)}} x_{w,\typevec, A}  & \le q_{w}, && \forall w \in W, \label{ilp:colleges}\\
	\sum_{\mathclap{\substack{
				S_{\typevec,A}\subseteq \acset(w)\colon \\
				\typevec[z]=1}}}x_{w,\typevec,A} & \ge \lowervec_w[z], & & \forall w\in W, \forall z\in [t], \label{ilp:college-lower}\\
	\sum_{\mathclap{\substack{
				S_{\typevec,A}\subseteq \acset(w)\colon \\
				\typevec[z]=1}}}x_{w,\typevec,A} & \le  \uppervec_{w}[z], && \forall w\in W, \forall z\in [t].\label{ilp:college-upper}%
	\end{alignat}
	It is straightforward to verify that the ILP formulation has $\mathcal{O}(m\cdot 2^t)$ integer variables, each with value at most~$n$,
	and $\mathcal{O}(n\cdot m+ t\cdot m)$ inequalities, and that it encodes feasibility with the described interpretation of the variables.
\end{proof}
\fi
\iflong
\begin{proof}
	Before we describe the ILP formulation, we introduce some notation.
	For each type vector~$\typevec \in \{0,1\}^t$ and each subset~$A\subseteq W$ such there there exists a student~$u$ with
	$\typevec_u=\typevec$ and $\acset(u)=A$, let $S_{\typevec, A}$ denote the set of students with type vector~$\typevec$ and acceptable set~$A$:
	\begin{align*}
	S_{\typevec,A} \coloneqq \{u\in U\mid \typevec_u = \typevec \wedge \acset(u)=A\}.
	\end{align*}
	Note that the sets~$S_{\typevec,A}$ partition the whole set~$U$ of students and for each $w\in W$ it holds that either $S_{\typevec,A} \subseteq \acset(w)$ or $S_{\typevec,A} \cap \acset(w)=\emptyset$ because of the symmetry of the acceptability.
	Further, there are at most $\min(n, 2^t\cdot 2^m)$ such sets,
	and they can be pre-computed in polynomial time: go through each student~$u$ add her to the corresponding set~$S_{\typevec_u, \acset(A)}$ if such set already exists; otherwise let $S_{\typevec_u, \acset(A)}\coloneqq \{u\}$.
	
	Next, for each college~$w$ and each non-empty set~$S_{\typevec,A}$ with $A\subseteq \acset(w)$, introduce an integer variable~$x_{w,\typevec,A}$; the value~$x_{w,\typevec,A}$ will be exactly the number of students with type vector~$\typevec$ and acceptable set~$A$ that shall be assigned to~$w$ in a feasible matching.
	Now we are ready to state the ILP formulation, where the expression~``$S_{\typevec,A}\subseteq \acset(w)$'' means taking all non-empty sets~$S_{\typevec,A}$ (as described above) with $S_{\typevec,A}\subseteq \acset(w)$:
	\begin{alignat}{3}
	x_{w,\typevec,A}  & \ge 0, &~~ & 
	\forall w \in W,  \forall S_{\typevec, A} \colon
	S_{\typevec,A} \subseteq \acset(w)
	, \label{ilp:integer} \\
	\sum_{w\in W}x_{w,\typevec,A} & \le |S_{\typevec, A}|, && \forall S_{\typevec, A} \colon
	S_{\typevec,A}\subseteq \acset(w),\label{ilp:students}\\
	\sum_{\mathclap{
			S_{\typevec,A}\subseteq \acset(w)}} x_{w,\typevec, A}  & \le q_{w}, && \forall w \in W, \label{ilp:colleges}\\
	\sum_{\mathclap{\substack{
				S_{\typevec,A}\subseteq \acset(w)\colon \\
				\typevec[z]=1}}}x_{w,\typevec,A} & \ge \lowervec_w[z], & & \forall w\in W, \forall z\in [t], \label{ilp:college-lower}\\
	\sum_{\mathclap{\substack{
				S_{\typevec,A}\subseteq \acset(w)\colon \\
				\typevec[z]=1}}}x_{w,\typevec,A} & \le  \uppervec_{w}[z], && \forall w\in W, \forall z\in [t].\label{ilp:college-upper}%
	\end{alignat}
	It is straightforward to verify that the ILP formulation has $\mathcal{O}(m\cdot 2^t)$ integer variables, each with value at most~$n$,
	and $\mathcal{O}(n\cdot m+ t\cdot m)$ inequalities; recall that we have reasoned that there are at most $n\cdot m$ non-empty sets~$S_{\typevec, A}$.
	To show the correctness of the ILP, note that Inequality~\eqref{ilp:integer} ensures that the introduced variables are non-negative.
	Inequality~\eqref{ilp:students} ensures in total there are enough students with type~$\typevec$ and acceptable set~$A$ to be matched to the colleges.
	Inequality~\eqref{ilp:colleges} ensures that each college's capacity is not exceeded while Inequalities~\eqref{ilp:college-lower}--\eqref{ilp:college-upper} ensure that the lower quotas and upper quotas of each college are fulfilled.
	
	For one direction, assume that we have values for \(\{x_{w, \typevec, A} \mid w \in W \land  \emptyset \neq S_{\typevec, A}\subseteq \acset(w)\}\) which form a solution for the ILP.
	We claim that going through each college~\(w \in W\) and each computed set~$S_{w,\typevec, A}$ with
	$\emptyset \neq S_{\typevec, A}\subseteq \acset(w)$
	and assigning \(x_{w, \typevec, A}\)~arbitrary not-yet assigned students from $S_{w,\typevec, A}$ is always possible because of Inequality~\eqref{ilp:students}.
	Moreover, the constructed matching~$M$ is a feasible matching because of Inequalities~\eqref{ilp:colleges}--\eqref{ilp:college-upper}.  
	
	For the converse direction assume that we have a feasible matching~\(M\).
	It can be easily verified that setting
	\[x_{w, \typevec, A} \coloneqq |\{u \in U \mid M(u) = w \land \typevec_u =\typevec
	\land \acset(u) = A\}|\] 
	for each $w\in W$ and each set~$S_{\typevec, A}$ with $\emptyset\neq S_{\typevec, A}\subseteq \acset(w)$
	yields a solution to the ILP.
	Inequalities~\eqref{ilp:integer}--\eqref{ilp:students} hold because $M$ is a matching.
	Inequalities~\eqref{ilp:colleges}--\eqref{ilp:college-upper} hold because $M$ is feasible.
\end{proof}
\fi

\begin{corollary}\label{cor:feasible-m+t-fpt}
	\feasible can be solved in $(m\cdot 2^m \cdot 2^t)^{\mathcal{O}(m\cdot 2^m \cdot 2^t)}\cdot (n\cdot m^2 + t\cdot m)$~time.
\end{corollary}
\begin{proof}
	This follows immediately from Lemma~\ref{lemma:ILP} and the fact that solving an ILP with $\rho$~variables and $L$~input bits can be solved in~$\mathcal{O}(\rho^{2.5\rho+o(\rho)}L)$~time~(\cite{Lenstra83} and~\cite{Kan87}).
\end{proof}

We have shown that to solve \feasible it suffices to know how many students of the same type vector and the same acceptable set are assigned to a college.
For \diverseties, we need more information to check stability, i.e., to tackle the preferences of the colleges.
This is because two students, even with the same type vectors and the same preference lists,
may be preferred differently by a college, which is relevant for stability.
\fi

\ifshort
However, in a matching only students which are least preferred by the college they are assigned to among those with the same type vector need to be considered to witness blocking pairs.
\fi
\iflong
Nevertheless, we show in the following that we do not need to store all students to check stability but only those that are least preferred by each college and for each type vector.  
\fi
We introduce two notations to formally describe such students.
Given a matching~$M$, a college~$v$, and a type vector~$\typevec\in\{0,1\}^t$,
\ifshort
let $\St(M,v,\typevec)\coloneqq \{u\in M(v)\mid \typevec_u=\typevec\}$ denote the set of students with type vector~$\typevec$ that are assigned to $v$, 
and let $\worst(M,v,\typevec)$ denote the set of students in~$M(v)$ with type vector~$\typevec$ that $v$ prefers least:

\smallskip
{\centering
  $\worst(M,v,\typevec) \coloneqq  \{u \in \St(M,v,\typevec) \mid  \St(M,v,\typevec) \succeq_v u \}$.\par
}

\fi
\iflong
let $\St(M,v,\typevec)\coloneqq \{u\in M(v)\mid \typevec_u=\typevec\}$ denote the set of students with type vector~$\typevec$ that are assigned to $v$: 
\begin{align*}
\St(M,v,\typevec)\coloneqq \{u\in M(v)\mid \typevec_u=\typevec\}\text{, and}
\end{align*}
let $\worst(M,v,\typevec)$ denote the set of students in~$M(v)$ with type vector~$\typevec$ that $v$ prefers least:
\begin{align*}
\worst(M,v,\typevec) \coloneqq & \{u \in \St(M,v,\typevec) \mid  \St(M,v,\typevec) \succeq_v u \}.
\end{align*}
\fi
\begin{proposition}\label{obs:critical-worst-pairs}
	Let $M$ be a feasible matching in an \diverseties{} instance.
	Then, an unmatched student-college pair~$\{u,w\}$ with \(w \succ_u M(u)\)
	is blocking in $M$ if and only if
	there is a subset of $k$~students~$U'\coloneqq \{u_{i_1}, \ldots, u_{i_k} \} \subseteq M(w)$ ($0\le k\le |M(w)|$) assigned to~$w$ such that
	\begin{compactenum}[(i)]
		\item\label{cond:atmostone} no two students from $U'$ have the same type vector,
		\item\label{cond:difftypes} each student $u'\in U'$ belongs to $\worst(M, w, \typevec_{u'})$,
		\item\label{cond:prefer} $w$ strictly prefers~$u$ to each student in $U'$, and
		\item\label{cond:feasible} $M \cup \{\{u,w\}\} \setminus (\{\{u,M(u)\}\} \cup \{\{u',w\}\mid u'\in U'\}$ is feasible for~$w$.
	\end{compactenum}
\end{proposition}
\ifshort
\begin{proof}[Proof Sketch]
	\eqref{cond:prefer} and \eqref{cond:feasible} necessarily hold for every set which witnesses that \(\{u,w\}\) is a blocking pair.
	\eqref{cond:atmostone} can be seen to hold for every minimal such set because at most one student from each type vector has to be removed from \(M(w)\) to make the addition of \(u\) to \(M(w)\) feasible for~\(w\). 
	\eqref{cond:difftypes} can be achieved by swapping each student \(u \notin \worst(M, w, \typevec_u)\) from a minimal witnessing set for a student in \(\worst(M, w, \typevec_u)\).
	This modification maintains all previously checked conditions and the fact that the set witnesses that \(\{u, w\}\) is a blocking~pair.
\end{proof}
\fi
	
\iflong	
\begin{proof}
	For the ``only if'' part, assume that $\{u,w\}$ is blocking $M$,
	and let $U^*\subseteq M(w)$ be a subset of students that witnesses this, i.e.,
	$w$ prefers~$u$ to every student in $U^*$ and $M\setminus (\{\{u,M(u)\}\}\cup \{\{u,w\} \mid u\in U^*\})\cup \{\{u,w\}\}$ is feasible for~$w$.
	We show how to construct a subset~$U'$ fulfilling the conditions given in the statement, starting with $U'\coloneqq \emptyset$.
	For each type vector~$\typevec\in \{0,1\}^{t}$ with $U'\cap \St(M,v,\typevec)\neq \emptyset$, pick an arbitrary student from $\worst(M, w, \typevec)$ and add it to~$U'$; note that such student must exist because $U'\cap \St(M,v,\typevec)\neq \emptyset$.
	Before we show that $U'$ fulfills the conditions given in the statement, we observe that
	by the definitions of $\worst$ and $U'$,
	the set~$U'$ characterizes the same set of type vectors as $U^*$.
	Formally,
	\begin{align}\label{eq:nice-U'}
	&\forall \typevec \in \{0,1\}^{t}:& \nonumber \\
	&    \exists~u' \in U'\colon \typevec_{u'} = \typevec \text{ iff. } \exists~u^*\in U^*\colon \typevec_{u^*}=\typevec.
	\end{align}
	Now, we are ready to show that $U'$ satisfies Conditions~\eqref{cond:atmostone}--\eqref{cond:feasible}.
	First of all, it is straightforward to check that Conditions~\eqref{cond:atmostone}--\eqref{cond:difftypes} are fulfilled.
	Condition~\eqref{cond:prefer} is also satisfied because of Property~\eqref{eq:nice-U'} and because $U^*$ witnesses that $\{u,w\}$ is a blocking pair.
	To see why Condition~\eqref{cond:feasible}, we observe that
	\begin{align}\label{eq:size}
	|U'|\le |U^*|
	\end{align}
	since Condition~\eqref{cond:atmostone} and Property~\eqref{eq:nice-U'} hold.
	This implies that $|\acset(w)\setminus U'|\le |\acset(w)\setminus U'| < q_w$; the last inequality holds because $U'$ is a witness.
	Moreover, by Properties~\eqref{eq:nice-U'}--\eqref{eq:size}, it follows that
	\begin{align*}
	\lowervec_w \le \sum_{u' \in M(w)\setminus U^* \cup \{u\}}\typevec_{u'}  \stackrel{\eqref{eq:nice-U'}\eqref{eq:size}}{\le} \sum_{u' \in M(w)\setminus U' \cup \{u\}} \typevec_{u'}.
	\end{align*}
	It remains to consider the upper quotas, i.e., to show that $\sum_{u' \in M(w)\setminus U' \cup \{u\}} \typevec_{u'}\le \uppervec_w$ holds.
	Suppose, for the sake of contradiction, that there exists a type~$z\in [t]$ such that
	\begin{align}\label{eq:upper}
	|\{u' \in (A(w) \setminus U') \cup \{u\} \mid \typevec_{u'}[z]=1\}| > \uppervec_{w}[z].
	\end{align}
	Since $\sum_{u'\in A(w)}\typevec_{u'}[z]\le \uppervec_w[z]$ (because $M$ is feasible) and $\typevec_{u}[z]\in \{0,1\}$,
	Property~\eqref{eq:upper} implies that
	$\typevec_{u}[z]=1$, 
	$\sum_{u'\in A(w)}\typevec_{u'}[z] = \uppervec_w[z]$, and
	no student in $U'$ has type~$z$.
	By Property~\eqref{eq:nice-U'} we infer that no student in $U^*$ has type~$z$.
	Together with~\eqref{eq:upper}, it follows that
	\begin{align*}
	& |\{u' \in (A(w) \setminus U^*) \cup \{u\} \mid \typevec_{u'}[z]=1\}|\\
	= & |\{u' \in A(w) \cup \{u\} \mid \typevec_{u'}[z]=1\}| \stackrel{\eqref{eq:upper}}{>} \uppervec_{w}[z],
	\end{align*}
	a contradiction to $U^*$ being a witness.
	
	The ``if'' part is straightforward as $U'$ witnesses that $\{u,w\}$ is a blocking pair.
\end{proof}
\fi

\iflong
Using \cref{obs:critical-worst-pairs}, we can give an algorithm to solve \diverseties in polynomial time when \(m + t\) is bounded.
\fi
\iflong \cref{obs:critical-worst-pairs} implies that if two matchings have the same information in terms of least preferred students for each type vector and the number of students for each type, then either both are feasible and stable or neither is.
Using this insight, we can use Dynamic Programming (DP) to solve \diverseties{}.
\fi

\begin{theorem}\label{thm:diverse-m+t-XP}
	\diverseties is solvable in $ \mathcal{O}(n^{m\cdot 2^t+(2m+1)\cdot (t+1)}\cdot m^2\cdot (n^t\cdot t+m))$~time.\todoH{Someone needs to check again.\newline R: Why not $n^{\bigoh(m\cdot 2^t)}$? Similarly also for Prop3...\newline H: But it is a difference: using bigoh in the exponent makes the running time looks huge because you could have $1000$ times $m\cdot 2^t$ in the exponent.} 
\end{theorem}
\iflong\begin{proof}\fi
\ifshort\begin{proof}[Proof Sketch]\fi
	Let $I=(U,W,T=[t],(\typevec_u,\succeq_u)_{u\in U}$, $(\succeq_w, q_w$, $\lowervec_w$, $\uppervec_w)_{w\in W})$ be an instance of \diverseties{}.
	We introduce an extra type possessed by each student,
    and require each college~$w\in W$ to have no more than $q_w$ students for this extra type to encode capacities by types.
    \iflong
	For convenience we still denote $|T|=t$ but keep the additional type in mind for the running time analysis.
	\fi 
	
	Motivated by \cref{obs:critical-worst-pairs}, we will exhaustively branch, for each college and each type vector~$\typevec\in\{0,1\}^t$, on the choice of a student~$\Wst(w_j,\typevec)\in \acset(w_j)\cup \{\top\}$ who will be in~$\worst(M, w_j, \typevec)$ for a hypothetical feasible and stable matching \(M\).
	Here $\Wst(w_j,\typevec) = \top$ is interpreted as $\worst(M, w_j, \typevec)=\emptyset$.
	Moreover we branch to determine the number~$\#(w_j,z)\in \{\lowervec_{w_j}[z], \ldots, \uppervec_{w_j}[z]\}$ of students of each type~$z\in [t]$ that each college~$w_j\in W$ receives by $M$.
	
	For each such branch we iteratively try to extend \(M_0 = \{\{\Wst(w_j, \typevec), w_j\} \mid w_j \in W, \typevec \in \{0,1\}^t\}\) to a feasible and stable matching which conforms to the guesses in the branch, one not yet matched student at a time.
	
	\iflong
	Assume that $U=\{u_1,\ldots,u_n\}$ and $W=\{w_1$, $\ldots,w_m,w_{m+1}\}$, where we use~$w_{m+1}$ to receive all students that are unmatched. 
	Further, let $U_0=\emptyset$ and for each~$i\in [n]$, let $U_i=U_{i-1}\cup \{i\}$.
	\fi
	
        
	More specifically, we only add a student-college pair to the matching if doing so maintains the status that each guessed \(\Wst(w_j, \typevec)\)-student is least preferred among the students assigned to \(w_j\) with type vector \(\tau\), the guessed number of students for each college and type is not exceeded, and there is no induced blocking pair involving the added student and some guessed \(\Wst(w_j, \typevec)\)-students (as witness).
	To check these conditions and more importantly to upper-bound the number of considered matchings we keep a \emph{record} in addition to each constructed (partial) matching, guessed least preferred students~\(\Wst(w_j, \typevec)\), $\typevec \in \{0,1\}^{t}$ and the guessed numbers~$\#(w_j,z)$ of students, $z\in [t]$, $w_j\in W$.
	\ifshort
	A record for a set~$U_i$ of students is an $(m+1)\times (t + 1)$-dimensional integer matrix~$\capR \in \{0,\ldots,n\}^{(m+1)\times (t + 1)}$
	storing the type-specific number of students assigned to a college, and the number of students assigned to it in total.
	Two ``partial'' matchings in a branch can be argued to be equivalent in terms of existence of feasible and stable extensions whenever they have the same record, which is why in each branch we only need to consider at most \(n^{(m+1) \cdot (t + 1)}\) matchings.
	\fi
	
    \iflong
	\paragraph{Records and the corresponding matchings.}
	A record for a set~$U_i$ of students is an $(m+1)\times (t + 1)$-dimensional integer matrix~$\capR \in \{0,\ldots,n\}^{(m+1)\times (t + 1)}$
	storing the type-specific number of students assigned to a college, and the number of students assigned in total.
	Formally, for each college~$w_j\in W$ and each type~$z\in [t]$,  
	let the number of students from $U_i$ with type~$z$ that are assigned to college~$w_j$ be stored in~$\capR[j][z]$, and the number of students from \(U_i\) that are assigned to college~\(w_j\) in total be stored in \(\capR[j][t + 1]\).  
        We say that two matchings~$M_1$ and $M_2$ have \myemph{the same record~$\capR$} if for each college~$w \in W$ the following hold:
        $\sum_{u\in M_1(w)}\typevec_{w} = \typevec_{u'\in M_2(w)}\typevec_{u'} = \capR[j]$.
        We also say that $M_1$ \myemph{realizes} record~$\capR$.
        
	\paragraph{Initialization.}
	The initial record~$\capR_0$ stores the information for $M_0$, i.e.,
	for each~$w_j\in W$,
	let $\capR_{0}[j] \coloneqq \sum_{u\in M_0(w_j)} \typevec_u$.
	
	\paragraph{Update.}
	For the update, we assume that we have all possible records for the student~$u_i$,
	and, for each~$\capR$ of the records, a possible matching such that the number of students for each specific type is store in $\capR$.
	For each record~$\capR$ and a corresponding matching~$M$ ``realizing'' the record~$\capR$,
	we consider assigning student~$u_{i+1}$ to each possible college~$w \in \acset(u_{i+1})\cup \{w_{m+1}\}$ in order to build a new record and its corresponding matching which includes $\{u_{i+1}, w_j\}$.
        We consider the current assignment of matching~$u_{i+1}$ to college~$w_j$ 
        if both of the following conditions are met:
        \begin{compactenum}[(a)]
          \item \label{dp:update:least}
          $u_{i+1} \succeq_{w_j} \Wst(w_j, \typevec_{u_{i+1}})$\text{, and }
          \item\label{dp:update:bp} for each other college~$w \in W\setminus \{w_j,w_{m+1}\}$ and
           each subset~$S\subseteq \{\Wst(w, \typevec)\neq \top \mid \typevec\in \{0,1\}^{t}\}$ of students with 
          \begin{inparaitem}
            \item[(b1)] $w \succ_{u_{i+1}}w_j$, and 
            \item[(b2)] $\forall s \in S \ u_{i+1}\succ_{w} S$, 
          \end{inparaitem} there must be a type~$z\in [t]$
          such that
          \begin{itemize}
            \item[(b3)] $\#(w,z)+\typevec_{u_{i+1}}[z]-\sum_{u'\in S}\typevec_{u'}[z] < \lowervec_{w}[z]$ or $\#(w,z)+\typevec_{u_{i+1}}[z]-\sum_{u'\in S}\typevec_{u'}[z] > \uppervec_{w}[z]$. 
          \end{itemize}
        \end{compactenum}
	otherwise we skip to next possible assignment of matching~$u_{i+1}$ to some college.
	The record~$\capR'$ for the new matching~$M'\coloneqq M\cup \{\{u_{i+1}, w_j\}\}$ is constructed as follows:
	\begin{align*}
	\forall k \in [m]\colon \capR'[k] \coloneqq
	\begin{cases}
	\capR[k]+\typevec_{u_{i+1}}, & \text{ if } j=k,\\
	\capR[k], & \text{ otherwise.}\\
	\end{cases}
	\end{align*}
	If $\capR'$ was already constructed in some previous consideration, even if the corresponding matchings differ, 
	then we also abandon the current consideration and go to next possible assignment.
	
	\paragraph{Checking the numbers.}
	\fi
	After we have considered the last student~$u_n$, \iflong and built all records and their corresponding matchings for all students in~$U_n$\fi
	we check whether there exists a record~$\capR$ with a matching~$M$ that corresponds to the information in~$\#(w_j,z)$, i.e., for each college~$w_j\in W$ and each type~$z\in [t]$ whether
        \iflong
        \begin{align}\label{dp:check:numbers}
          \capR[j][z] = \#(w_j,z) \text{ holds.}
        \end{align}
        \fi
      \ifshort  $\capR[j][z] = \#(w_j,z)$ holds. \fi
	We return~$M$ once we found a matching fulfilling the above condition.
	If no such matching is found, we return that we have a ``no''-instance.
		\ifshort
	Correctness can be argued using \cref{obs:critical-worst-pairs} and the fact that, in each branch, two partial matchings with the same record are equivalent in terms of existence of feasible and stable extensions.
	\fi
	\iflong
	\paragraph{Correctness.}
	We claim that $I$ has a feasible and stable matching~$M$ if and only if
	our DP procedure returns a matching.
	
	For the ``only if'' part, assume that $I$ admits a feasible and stable matching~$N^*$ and let $M^*=N^*\cup \{\{u, w_{m+1}\} \mid N^*(u) = \bot\}$.
	Recall that we conduct a DP for each possible branching of the least preferred students and the number of students of each type. 
	Thus, let us consider the branching where 
        \begin{compactitem}
          \item for each $w_j\in W\setminus \{w_{m+1}\}$ and for each type vector~$\typevec\in \{0,1\}^{t}$, it holds that $\Wst(w_j,\typevec) \in \worst(M^*, w_j,\typevec)$ if $\worst(M^*,w_j,\typevec)\neq \emptyset$; $\Wst(w_j,\typevec)\coloneqq \top$ otherwise, and
          \item for each $w_j \in W$ and for each type~$z\in [t]$ it holds that $\#(w_j,z)=\sum_{u'\in M^*(w_j)}\typevec_{u'}[z]$.
        \end{compactitem}
	Let $M_0$ be the initial matching containing all the pairs consisting of a college and the guessed least preferred student.
	Re-enumerate the remaining unmatched students as $u_1,\ldots,u_n$.
        We claim that we will find a matching~$M$ which has the same record and the same least preferred students as $M^*$ such that for each $\{u_{i'},M(u_{i'})\}$ with $M(u_{i'})\neq w_{m+1}$ both Condition~\eqref{dp:update:least} and Condition~\eqref{dp:update:bp} are met (setting $u_{i+1}\coloneqq u_{i'}$ and $w_{j}\coloneqq M(u_{i'})$).
        Suppose, for the sake of contradiction, that our DP returns no such matching.
        Let $k \in [n]$ be the largest index for which the matching~$M'\coloneqq \{\{u_k, w\}\in M^* \mid k'\le k\}\cup M_0$ is still considered in the DP.	
        Let $\capR'$ be the record of $M'$.
        If $k < n$, then consider the branching where we consider adding~$\{u_{k+1}, M^*(u_{k+1})\}$ to $M'$; let $M''\coloneqq M'\cup \{\{u_{k+1},M^*(u_{k+1})\}\}$.
        Since $M^*$ satisfies Condition~\eqref{dp:update:least} (setting $u_{i+1}\coloneqq u_{k+1}$ and $w_j\coloneqq M^*(u_{k+1})$) matching~$M''$ must also satisfy Condition~\eqref{dp:update:least}.
        Now, observe that $M^*$ also satisfies Condition~\eqref{dp:update:bp} because otherwise there exist a college~$w\in W\setminus \{w_j,w_{m+1}\}$ and a subset~$S\subseteq \{\Wst(w,\typevec) \neq \top \mid \typevec \in \{0,1\}^{t})\}$
        which satisfy (b1) and (b2) such that
        $\lowervec_{w}[z] \le \#(w,z)+\typevec_{u_{k+1}}[z]-\sum_{u'\in S}\typevec_{u'}[z] \le \uppervec_{w}[z]$ for each type~$z\in [t]$.
        Since $\#(w,z)$ is also the number of students with type~$z$ which are assigned to~$w$ in $M^*$,
        the subset~$S$ witnesses that $(u_{k+1},w_j)$ is a blocking pair, a contradiction.
        Thus,~$M''$ also satisfies Condition~\eqref{dp:update:bp}.
        Hence, the only reason we abandoned $M''$ for the later branching is that there exists another matching~$M_{k+1}$ with the same record as $M''$ such that Conditions~\eqref{dp:update:least}--\eqref{dp:update:bp} are met for $(u_{k+1}, M_{k+1}(u_{k+1}))$.
        Analogously, we can infer that after we have considered student~$u_n$,
        there must be a matching~$M_n$ which has the same record and the same least preferred students as $M^*$ such that for each $\{u_{i'},M_n(u_{i'})\}$ with $M_n(u_{i'})\neq w_{m+1}$ both Condition~\eqref{dp:update:least} and Condition~\eqref{dp:update:bp} are met (setting $u_{i+1}\coloneqq u_{i'}$ and $w_{j}\coloneqq M_n(u_{i'})$).
        Let $\capR_n$ be the corresponding record.
        Then, $\capR_n$ must also satisfy the checking given in~\eqref{dp:check:numbers}.

        We have just shown that our DP returns a matching~$M_n$ that has the same ``crucial information'' as $M^*$.
        Then, $M_n$ must be a feasible and stable matching.
        Obviously, $M_n$ is feasible because $M^*$ and $M_n$ have the same record~$\capR_n$.
        To see that $M_n$ is stable, let us consider an arbitrary unmatched pair~$\{u,w\}\notin M_n$ and an arbitrary subset~$S \subseteq M_n(w)$ of students assigned to $w$
        such that student~$u$ strictly prefers~$w$ to~$M_n(u)$ and college~$w$ strictly prefers $u$ to every student in $S$.
        Assume for contradiction, that $S$ is a witness for \(\{u,w\}\) being a blocking pair of $M_n$.
        Then, it must hold that
        $M_n \cup \{\{u,w\}\} \setminus (\{\{u,M_n(u)\}\}\cup \{\{u',w\}\mid u'\in S\})$
        is feasible for~$w$.
        By \cref{obs:critical-worst-pairs},
        there exists a subset~$U'\coloneqq \{\Wst(w, \typevec)\neq \top \mid \typevec \in \{0,1\}^t\}$ of students such that
        $w$ strictly prefers $u$ to $U'$ and
        \begin{align*}
          \lowervec_u \le \typevec_{u}+\sum_{u' \in M_n(w)\setminus U'}\typevec_{u'}\le \uppervec_{w}.
        \end{align*}
        This means that when we consider branching on adding the pair~$\{u,w\}$ Condition~\eqref{dp:update:bp} in the update step does not hold for~$u_{i+1}\coloneqq u$ and $w_j \coloneqq w$, a contradiction.
        Thus,~$M_n$ is also stable.

	For the ``if'' part, assume that for some choices of least preferred students~$\Wst(w_j,\typevec)$ and ``feasible'' numbers~$\#(w_j,z)$, the DP returns a matching~$N$.
	We claim that $M\coloneqq N\setminus \{\{u,w_{m+1}\}\in N\}$ is a feasible and stable matching for $I$.
	It is straight-forward to see that $M$ is feasible because we only branch on ``feasible'' numbers and $N$ satisfies the final check~\eqref{dp:check:numbers}. 
	For the stability, suppose, for the sake of contradiction, that an unmatched pair~$\{u,w\}\notin M$ with $u$ preferring $w$ to $M(u)$ is blocking~$M$.
	Then, by \cref{obs:critical-worst-pairs}, there exists a subset~$U'\coloneqq \{u_{i_1}, \ldots, u_{i_k}\} \subseteq M(w)$ of $k$ students for which Conditions~\eqref{cond:atmostone}--\eqref{cond:feasible} hold.
	Since Condition~\eqref{cond:difftypes} holds for $U'$, for each $u_{i_s}\in U'$ ($s\in [k]$)
	there exists exactly one guessed least preferred student~$\Wst(w_j,\typevec_{u_{i_1}})$ with $\Wst(w_j,\typevec_{u_{i_1}})\neq \top$.
	Then, $U''\coloneqq \{\Wst(w_j, \typevec_u) \mid u\in U'\}$ must also be a subset satisfying Conditions~\eqref{cond:prefer}--\eqref{cond:feasible}.
	In other words, $U''$ does not fulfill Condition~(b3) when we consider adding the pair~$\{u,M(u)\}$, a contradiction to $M$ being returned.
	
	\paragraph{Running time.}
	For each of the $\mathcal{O}(n^{m\cdot 2^t}\cdot n^{m\cdot (t+1)})$~possible starting branches, we iteratively match each student to some college.
	As, we only keep a matching if the newly constructed record is distinct from existing ones and
	there are at most $n^{m\cdot (t + 1)}$ different records, there are $\mathcal{O}(n^{m\cdot (t + 1)})$ matchings to maintain.
	For each considered student~$u_i$ and each existing record for the students~$U_{i-1}$,
	we have $\mathcal{O}(m)$ possibilities to assign $u_i$.
	To check \eqref{dp:update:least} we require \(2^t\) many steps.
	Note that for \eqref{dp:update:bp} we only need to consider subsets \(S\) of size at most \(t\), which then leads to a checking time of \(\mathcal{O}(m \cdot n^{t} \cdot t)\).
        For each considered student, checking whether some record already exists takes in total $\mathcal{O}(m^2)$ time because there are at most $m$ records that are relevant for the comparisons.
	After having considered student~$u_n$, we check whether at least one of the constructed matchings fulfills Property~\eqref{dp:check:numbers} in $\mathcal{O}(t)$~time.
	In total, the running time lies in $\mathcal{O}\big(n^{m\cdot 2^t+m\cdot(t+1)}\cdot n \cdot n^{m\cdot (t+1)}\cdot (m \cdot n^{t}\cdot (t + 2^t +  m \cdot n^t \cdot t + m^2) + t) \big)=
        \mathcal{O}(n^{m\cdot 2^t+(2m+1)\cdot (t+1)}\cdot m^2\cdot (n^t\cdot t+m))$.
	\fi
      \end{proof}

\section{Conclusion}

\iflong
We identified and studied a natural, albeit highly intractable, stable matching problem enhanced with diversity constraints~(\diverseties).
We located the source of intractability 
by considering both relaxations (by dropping the lower quotas or the stability constraints)
and restrictions (such as  upper-bounding the number~$n$ of students, the number~$t$ of types, the number~$m$ of colleges, and/or the maximum upper quota~$\umax$, and the maximum capacity~$\qmax$).
While most of the cases are still at least NP-hard, we also identified special cases for which we provided polynomial-time algorithms.
This implies that the respective problems lie in the class XP in terms of parameterized complexity~\cite{DF13,FG06,Nie06,CyFoKoLoMaPiPiSa2015}.
A natural question is whether the XP results can be improved to fixed-parameterized tractability.
For instance, with respect to the number~$n$ of students, we indeed obtain an exponential-size problem kernel, which implies that \diverseties{} parameterized by~$n$ is fixed-parameter tractable, and we can show that polynomial-size kernels are unlikely to exist.
However, for the number~$m$ of colleges combined with the maximum capacity~$\qmax$, our algorithms (except for the case when $\lmax=0$ and no ties are present) are essentially optimal because we can show that both \diverse{} and \feasible{} are W[2]-hard; this refutes any fixed-parameter algorithms for $(m+\qmax)$ unless FPT{}$=${}W[2].
For the combined parameter~$m+t$ (recall that $t$ denotes the number of types), we can encode \feasible via an integer linear program formulation with $f(m,t)$ variables and a polynomial number of inequalities, and show that it is fixed-parameter tractable for $m+t$.
Summarizing, the only considered fragment left open in terms of fixed parameter tractability is \diverseties{} parameterized by $(m+t)$.

Continuing with parameterized complexity research we can also explore other parameters associated with the underlying acceptability graph and the type graph (a student is connected to a type via an edge if she has this type), such as tree width~\cite{GupSauZeh2017,BreHeeKnoNie2019}.

Another interesting future research direction is to investigate the trade-off between stability and diversity by allowing few blocking pairs~\cite{AbBiMa2005,CHSYicalp-par-stable2018} or few unsatisfied diversity constraints.

Last but not least, it would also be interesting to know whether \diverseties{} becomes polynomial-time solvable when the input preferences are for instance single-peaked or single-crossing~\cite{BreCheFinNie2017}.

\fi

\iflong
\paragraph{Acknowledgments.}
Jiehua Chen is supported by the WWTF research project~(VRG18-012). Robert Ganian and Thekla Hamm acknowledge support from the Austrian Science Foundation (FWF, project P31336).
\fi

\ifshort \clearpage
\bibliographystyle{named}
\bibliography{bib}
\fi
\iflong

\clearpage
\appendix

\section{Results on Parameterized Complexity}

We complement the algorithms from Proposition~\ref{prop:diverseties-m+qmax:XP} with the following W[2]-hardness results.

\begin{proposition} \label{prop:feasible-m+qmax:W2-h}
	\feasible and \diverse parameterized by \(m + \qmax\) are \textnormal{W[2]}-hard.
\end{proposition}
\begin{proof}
	We show this by providing a parameterized reduction from the W[2]-complete \textsc{Set Cover} problem, parameterized by the solution size~\cite{DF13}.
	\decprobnormal{Set Cover}
	{An \(n^*\)-element universe \(\mathcal{U} = \{x_1, \dotsc, x_{n^*}\}\), a collection~\(\mathcal{S} = \{S_1, \dotsc, S_{m^*}\}\) of $m^*$~sets,
          where each \(S_i \subseteq \mathcal{U}\) and \(\mathcal{U} \subseteq \bigcup_{i \in [m]} S_i\) and \(k \in \mathbb{N}\) with \(k \leq m\).}
	{Is there a size-at-most-$k$ \myemph{set cover}, i.e., a subcollection~$C\subseteq \mathcal{S}$ with $|C|\le k$
          such that \(\mathcal{U} \subseteq \bigcup_{S \in C} S\)?}
	
	Let \(I=(\mathcal{U}, \mathcal{S}=\{S_1, \dotsc, S_{m^*}\}, k)\) be a \textsc{Set Cover} instance; without loss of generality assume that $m^* > k \ge 1$.
	We construct an \feasible/\diverse instance as follows.
	Let \(U = \{s_1, \dotsc, s_{m^*}, d_1, \dotsc, d_{m^*}\}\), \(W = \{w\}\) and \(t = n^* + m^* + 1\).
	We also refer to \(S \coloneqq \{s_1, \dotsc s_{m^*}\}\) as the \emph{set-students}, and to \(D \coloneqq \{d_1, \dotsc d_{m^*}\}\) as the \emph{dummy-students}.
	Furthermore, let \(\acset(w) = U\), i.e., we have unrestricted acceptability.
	This uniquely defines the preference lists of all students.
	The preference list of \(w\) is given by \(s_1 \succ_w d_1 \succ_w \cdots \succ_w s_{m^*} \succ_w d_{m^*}\).
	
	There are $n^*+m^*+1$ types: the first $n^*$ ones correspond to the elements, the middle $m^*$ ones correspond to the sets, while the last one is dedicated to the dummy students which will allow only $k$~dummy students to be assigned to $w$.
        For each set-student~\(s_i\), set \(\typevec_{s_i}[z] = 1\) if ``\(z \in [n^*]\) and \(x_z \in S_i\)'', or ``\(z = n^* + i\)''.
	That is, the set-student corresponding to~\(S_i\) has type \(n^* + i\), as well as types corresponding to \(S_i\)'s elements.
	For each dummy student~\(d_i\), set \(\typevec_{d_i}[z] = 1\) if ``\(z \in [n^* + 1, n^* + m^*]\) and \(z \neq n^* + i\)'', or ``\(z = n^* + m^* + 1\)''.
	That is, all dummy students have type \(n^* + m^* + 1\), and a dummy student corresponding to \(S_i\) has the types corresponding to all other sets given in the \textsc{Set Cover} instance.
        In other words, $s_i$ and $d_i$ together cover all the set-types.
	
	Set \(\lowervec_w = \{1\}^{n^*}\{k\}^{m^* + 1}\), \(\uppervec_{w} = \{k\}^{n^* + m^* + 1}\) and \(q_w = 2k\).
	This concludes the construction of the \feasible/\diverse instance.
	
	We claim that a solution for the original instance of \textsc{Set Cover} infers a solution for the constructed instance of \feasible/\diverse, and vice versa.
	
	For the first direction, let \(C\) be a size-$k$ set cover for $I$; we add arbitrary sets to $C$ to make sure that $C$ has indeed exactly $k$~sets.
	Start with \(M \coloneqq \{\{s_i, w\}, \{d_i, w\} \mid S_i \in C\}\).
	Feasibility is straightforward to check.
	Assume for contradiction, that \(M\) is not stable, i.e., there is some blocking pair \(\{u, w\}\) for \(M\).
	Let this be witnessed by \(U' \subseteq M(w)\).
        There are two cases for $u$, either $u=s_i$ or $u=d_i$ for some $i\in [m^*]$ such that $S_i \notin C$.
        First of all, since college~$w$ already receives~$q_w$ students, it follows that $|U'|\ge 1$.
        \begin{itemize}
          \item If $u=d_i$, then since the lower quota and upper quota of $w$ for the type~$n^*+m^*+1$ are equal to $k$
          and only the dummy students have such type, it follows that $U'\cap D=\{d_j\}$ with $j\neq i$.
          By our definition of $M$, it follows that $s_i\notin M(w)$ and that no student in $U'\setminus \{d_j\}$ has type~$n^*+i$.
          However, since student~$d_j$ has type~$n^*+i$ but student~$d_i$ does not, replacing $d_i$ with $U'$ will exceed the upper quota of $w$ regarding type~$n^++i$,
          a contradiction.
          \item If $u=s_i$, then since the lower quota and upper quota of $w$ for type~$n^*+m^*+1$ are equal to $k$ but only dummy students have such type,
          it follows that $U'\cap D=\emptyset$.
          However, since the lower quota and the upper quota for type~$\{n^*+i\}$ are equal to $k$ and since no set-student other than $s_i$ has this type, it follows that that $U'$ must contain a dummy student, a contradiction.
        \end{itemize}
	Thus \(M\) is indeed feasible and stable.
	
	Conversely, given a feasible (and stable\footnote{This fact is not necessary for showing that the \textsc{Set Cover} instance~$I$ is a yes instance.}) matching \(M\) for the constructed instance, \(C \coloneqq \{S_i \mid \{s_i, w\} \in M\}\) is a solution to the original \textsc{Set Cover} instance:
	\(|C| \leq k\) is ensured by \(q_w = 2k\) and there are exactly $k$ dummy students in $M(w)$.
        \(\mathcal{U} \subseteq \bigcup_{i : S_i \in C} S_i\) is ensured by \(\lowervec_w[z] = 1\) and \(\typevec_{s_i}[z] = 1 \Leftrightarrow x_z \in S_i\) for all \(z \in [n^*]\).
	
	Obviously in our constructed instance \(m = 1\) and \(\qmax = 2k\) which makes this a valid parameterized reduction from \textsc{Set Cover} parameterized by \(k\) to \feasible, and \diverse respectively parameterized by \(m + \qmax\).
\end{proof}

However, when no ties are present and $\lmax=0$, we are able to provide an FPT algorithm for \diverse{} with respect to $m+\qmax$.
It is based on the following observation.

\begin{observation}\label{obs:diverse:lmax=0-noties:u}
  Assume that there are no ties and $\lmax=0$.
  If student~$u$ is college~$w$'s most preferred student such that $\typevec_{u}\le \uppervec_{w}$,
  then each stable matching must assign a college~$w'$ to $u$ with $w' \succeq_u w$.
\end{observation}

\begin{proof}
  Suppose, for the sake of contradiction, that there exists a stable matching~$M$ with $w \succ_u M(w)$.
  Then, $\{u,w\}$ is blocking~$M$ as $w$ prefers~$u$ to~$M(w)$ and can replace $M(w)$ with $u$ to obtain a new feasible matching.
\end{proof}

\cref{obs:diverse:lmax=0-noties:u} can be used to extend the Gale-Shapley algorithm so as to obtain an FPT algorithm for $m+\qmax$.

\begin{proposition}\label{prop:diverse:lamx=0-noties:m+qmax:FPT}
  If $\lmax=0$, then
  \diverse{} can be solved in $\mathcal{O}(m^{m\cdot \qmax}\cdot n\cdot m\cdot t)$~time.
\end{proposition}

\begin{proof}[Proof Sketch.]
  Let $I=(U,W,(\typevec_u)_{u\in U}$, $(\succ_x)_{x\in U\cup W}$ $(q_w$, $\lowervec_w =\boldsymbol{0},\uppervec_w)_{w\in W})$ be an instance of \diverse{}.
  The idea of our algorithm is to start with a feasible empty matching~$M$,
  and iteratively find an unassigned student~$u\in U$ for which there exists a college~$w\in \acset(u)$ such that 
    \begin{align}
      |M(w)| < q_w \text{ and } \typevec_{u}+\sum_{v\in M(w)}\typevec_v\le \uppervec_w&\label{cond:first}
  \end{align}
  (i.e., $w$ has still some capacity to accommodate $u$ and assigning~$u$ to $w$ does not exceed the upper quotas)
  and \(u\) there is no unassigned student which is strictly preferred over \(u\) by \(w\) satisfying (\ref{cond:first}).
  If there is no such student, then we check whether the constructed matching is stable by using the polynomial-time algorithm given in \cref{lem:lq=0:stable-equiv-form}.
  Otherwise, based on \cref{obs:diverse:lmax=0-noties:u}, every stable matching must assign some college to~$u$.
  Hence, we branch into assigning each one college~$w'$ from~$\acset(u)$ to~$u$, i.e., adding $\{u,w'\}$ to $M$.
  For each possible branch, we check for a next student satisfying Condition~\eqref{cond:first}, and continue as described.
  We return ``no'' if no constructed matching is stable.

  As for the running time, since there are at most $m$ branches for each found student and in each branch at least one college receives one more student,
  the search tree built according to our branching algorithm has depth~$\mathcal{O}(\qmax \cdot m)$ and $\mathcal{O}(m^{\qmax\cdot m})$ nodes.
  Since in each node of the search tree we can find a next student in linear time and at each leaf of the search tree we can check stability in linear time using \cref{lem:lq=0:stable-equiv-form},
  we can decide $I$ in $\mathcal{O}(m^{m\cdot \qmax}\cdot n\cdot m\cdot t)$~time.
\end{proof}

The condition that \(\lmax = 0\) cannot be dropped for the fixed parameter tractability given in Proposition~\ref{prop:diverse:lamx=0-noties:m+qmax:FPT}, as is signified by the following result.

\begin{proposition}\label{prop:diverseties:m+qmax:w[1]-h}
	\diverseties parameterized by \(m + \qmax\) remains W[1]-hard, even when \(\lmax = 0\).
\end{proposition}
\begin{proof}
  We provide a parameterized reduction from the W[1]-complete \textsc{Set Packing} problem, parameterized by the solution size~$k$~\cite{DF13}.
	\decprobnormal{Set Packing}
	{An \(n^*\)-element universe \(\mathcal{U} = \{x_1, \dotsc, x_{n^*}\}\),
          a collection~$\mathcal{S}=\{S_1,\ldots, S_{m^*}\}$ of $m^*$~sets,
          where each \(S_i \subseteq \mathcal{U}\), and \(k \in \mathbb{N}\) with $k \le m$.}
	{It there size-at-least-$k$ \myemph{set packing}, i.e.,
          a subcollection~$C \subseteq \mathcal{S}$ with $|C|\le k$ such that
          for each two sets~$S, S'\in C$ it holds that $S\cap S' =\emptyset$?}
	
	
	Let \(\mathcal{U}, S_1, \dotsc, S_{m^*}\) and \(k\) be a \textsc{Set Packing} instance.
	We construct a \diverseties instance as follows:
	Let \(U = \{u_1, \dotsc, u_{m^*}, r_1, r_2, r_3\}\), \(W = \{w, a, b\}\) and \(t = n^* + 2\).
	For all \(i \in [m^*]\), and \(z \in [n^*]\) set \(\typevec_{u_i}[z] = 1\) whenever \(x_z \in S_i\),
	and set \(\typevec_{u_i}(m^* + 1) = \typevec_{u_i}(m^* + 2) = 0\).
	That is, a student corresponding to \(S_i\) has exactly the types corresponding to \(S_i\)'s elements.
	Further, for each~\(i \in [m^*]\), let \(\acset(u_i) = \{w\}\), which also uniquely determines the preference list of \(u_i\).
	Set all other preferences, types, quotas and capacities according to the following table, where
        the preferences of~$w$ have all students from $\{u_1,\ldots,u_{m^*}\}$ tied at the first position and rank student~$r_2$ at the second position while all other preferences do \emph{not} have ties and are ordered by $\succ$: 
        \begin{center}
        \tabcolsep=3pt        
	\begin{tabular}{@{}>{\columncolor{white}[0pt][\tabcolsep]}l@{}>{\columncolor{white}[0pt][\tabcolsep]}l@{\;}c|l@{}l@{}>{\columncolor{white}[0pt][\tabcolsep]}c@{\;}>{\columncolor{white}[0pt][\tabcolsep]}c}
		\toprule
		S. & Pref. & T. & C. & Pref. & UQ. & C. \\\midrule 
		\highlight{$r_1\colon$} & \highlight{$b \osucc a$} & \highlight{$\{0\}^{n^*}10$} &  \highlight{$ w\colon$}& \highlight{$\{u_1,\dotsc, u_{m^*}\}\osucc r_2$} & \highlight{$\{1\}^{n^* + 2}$} & \highlight{$k$}\\
		\highlight{$r_2\colon$} & \highlight{$b \osucc w \osucc a$} &\highlight{$\{0\}^{n^*}11$} &   \highlight{$a\colon$} & \highlight{$r_1 \osucc r_2 \osucc r_3$} & \highlight{$\{0\}^{n^*}11$} & \highlight{$1$}\\
		\highlight{$r_3\colon$} & \highlight{$a \osucc b$} &\highlight{$\{0\}^{n^*}01$}&  \highlight{$b\colon$} & \highlight{$r_3 \osucc r_2 \osucc r_1$} & \highlight{$\{0\}^{n^*}11$} & \highlight{$2$}\\
		\bottomrule
	\end{tabular}
      \end{center}
        We set the lower quotas to zero.
        This completes the construction of the instance, which can clearly be done in polynomial time.

	We claim that a size-at-least-$k$ set packing for $I$ infers a feasible and stable matching for the constructed instance of \diverseties, and vice versa.
	
	For the first direction, let \(C\) be a set packing with $k$ sets; we delete arbitrary sets from $C$ to ensure that $C$ has size~$k$.
        Then, we claim that the matching~\(M \coloneqq \{\{r_1, b\}, \{r_2, a\}, \{r_3, b\}\} \cup \{\{u_i, w\} \mid S_i \in C\}\) is a solution for the \diverseties instance:
	Feasibility can be easily verified.
	Assume for contradiction, that \(M\) is not stable, i.e., there is some blocking pair \(\{\alpha,\beta\}\) for \(M\). 
	Note that because \(|M(w)| = |\{u_i \mid S_i \in C\}| = |C| = k = q_w\), we can apply Lemma~\ref{lem:aux}\eqref{aux-nec}.
	Hence we know that \(\alpha \in \{u_i \mid i \in [m^*]\}\) and \(\beta = w\).
	Because \(|M(w)| = q_w\), this has to be witnessed by some non-empty subset of \(M(w)\).
	However there cannot be such a set which is strictly preferred over \(\alpha\) which is a contradiction.
	Thus \(M\) is indeed a feasible and stable matching for the constructed \diverseties instance.
	
	Conversely, given a feasible and stable matching \(M\) for the constructed instance, \(C = \{S_i \mid \{u_i, w\} \in M\}\) is a size-$k$ set packing for the original \textsc{Set Packing} instance:
	By Lemma~\ref{lem:aux}\eqref{aux-card}, the stability of \(M\) implies that \(|\{\{u_i, w\} \in M \mid i \in [m^*]\}| = q_w = k\), and thus \(|C| = k\).
	Moreover the upper quotas for \(w\) and the types of the elements of \(C\), ensure that no two sets in $C$ have non-empty intersection,
        for otherwise the upper quota of some type that corresponds to the element in the intersection would be exceeded. 
	Thus \(C\) is indeed a set packing for the original \textsc{Set Packing} instance.
	
	Obviously in our constructed instance \(\lmax = 0\), \(m = 3\) and \(\qmax = k\) which makes this a valid parameterized reduction from \textsc{Set Packing} parameterized by \(k\) to \diverseties with \(\lmax = 0\) parameterized by \(m + \qmax\).
\end{proof}

We have shown that \diverseties{} parameterized by~$n$ is fixed-parameter tractable~(\cref{thm:fptn}) since it admits an exponential-size problem kernel~(\cref{thm:fptn}).
In the following,  we show that a polynomial-size problem kernel is unlikely to exist.
\begin{proposition}\label{prop:nopoly-m+t+qmax}
  \feasible and \diverseties{} do not admit a problem kernel with size polynomially bounded by $m+t+\qmax$, unless \textnormal{NP}~\(\subseteq\)~\textnormal{coNP/P}.
\end{proposition}
\begin{proof}
  To see this for \feasible, we note that the reduction given by \citeauthor{AzizGaspersSunWalsh2019aamas}~\shortcite[Proposition~5.1]{AzizGaspersSunWalsh2019aamas} can be considered as from \textsc{Set Cover}.
  It produces an instance to \feasible where the number~$t$ of types is equal to the size of the universe~$\mathcal{U}$, the number of colleges is one, and the maximum capacity is equal to the set cover size~$k$.
  Since \textsc{Set Cover} does not admit a polynomial-size problem kernel for~$|\mathcal{U}|+k$,
  it follows that neither does \feasible admit a polynomial-size problem kernel for $m+t+\qmax$~\cite{DLS09}.

  The above reasoning can be used to show the same for \diverseties.
\end{proof}

\begin{proposition}\label{prop:nopoly-n}
  \feasible and \diverse{} do not admit a problem kernel with size polynomially bounded by $n$,
  unless \textnormal{NP}~\(\subseteq\)~\textnormal{coNP/P}.
\end{proposition}

\begin{proof}
  This follows immediately from the parameterized reduction from \textsc{Set Cover} to \feasible
  and \diverse given in the proof of Proposition~\ref{prop:feasible-m+qmax:W2-h}, and the fact that unless NP \(\subseteq\) coNP/P, \textsc{Set Cover} parameterized by the size of the size of the set collection~$\mathcal{S}$ does not admit a polynomial-size problem kernel~\cite[Theorem~5]{BreCheHarKraNieSucWoe2014}.
\end{proof}
\fi

\end{document}